\documentclass[10pt]{article}

\usepackage{amssymb}   
\usepackage{amsthm}    
\usepackage{amsmath}   
\usepackage{stmaryrd}  
\usepackage{titletoc}  
\usepackage{mathrsfs}  
\usepackage{graphicx}

\newlength{\defbaselineskip}
\newcommand{\setlinespacing}[1]%
           {\setlength{\baselineskip}{#1 \defbaselineskip}}

\theoremstyle{plain}
\newtheorem{thm}{Theorem}[section]
\newtheorem{cor}[thm]{Corollary}
\newtheorem{lem}[thm]{Lemma}
\newtheorem{prop}[thm]{Proposition}

\theoremstyle{definition}
\newtheorem{defn}{Definition}[section]

\newtheorem{rmk}{Remark}[section]

\newcommand{\eps}{\varepsilon}

\newcommand{\cS}{\mathcal{S}}
\newcommand{\cF}{\mathcal {F}}
\newcommand{\cO}{\mathcal {O}}

\newcommand{\bP}{\mathbb{P}}

\newcommand{\bR}{\mathbb{R}}

\newcommand{\sF}{\mathscr{F}}
\newcommand{\sP}{\mathscr{P}}
\newcommand{\sD}{\mathscr{D}}

\newcommand{\sS}{\mathscr{S}}

\DeclareMathOperator*{\esssup}{ess\,sup}

\makeatletter\@addtoreset{equation}{section} \makeatother
\allowdisplaybreaks
\begin{document}

\title{Forward-Backward Stochastic Differential Systems Associated to Navier-Stokes Equations in the Whole Space 
\thanks{This research is supported by the Natural Science Foundation of China (Grants
\#10325101 and \#11171076), the Science Foundation for Ministry of
Education of China (No.20090071110001), and the Chang Jiang Scholars
Programme. Part of the work was done when the second author visited Department of Mathematics, ETH, Z\"urich in the summer of 2011. The hospitality of ETH is greatly appreciated. He also would like to
 thank Professors Michael Struwe and  Alain-Sol Sznitman for very helpful discussions and comments related to Navier-Stokes equations.
 \textit{E-mail}: \texttt{delbaen@math.ethz.ch} (Freddy Delbaen), \texttt{qiujinn@gmail.com} (Jinniao Qiu), \texttt{sjtang@fudan.edu.cn} (Shanjian Tang).}
}


\author{Freddy Delbaen \and Jinniao Qiu \and Shanjian Tang
}

%


\maketitle

\begin{abstract}

A coupled forward-backward stochastic differential system (FBSDS) is formulated in spaces of fields for the incompressible Navier-Stokes equation in the whole space. It is shown to have a unique local solution, and further if either the Reynolds number is small or the dimension of the forward stochastic differential equation is equal to two, it can be shown to have a unique global solution.  These results are shown with probabilistic arguments to imply the known existence and uniqueness results for the Navier-Stokes equation, and thus provide probabilistic formulas to the latter. Related results and the maximum principle are also addressed for partial differential equations (PDEs) of Burgers' type. Moreover, from truncating the time interval of the above FBSDS, approximate solution  is derived for the Navier-Stokes equation by a new class of FBSDSs and their associated PDEs; our probabilistic formula is also bridged to the  probabilistic Lagrangian representations for the velocity field, given by Constantin and Iyer (Commun. Pure Appl. Math. 61: 330--345, 2008) and Zhang (Probab. Theory Relat. Fields 148: 305--332, 2010) ; finally, the solution of the Navier-Stokes equation is shown to be a critical point of controlled forward-backward stochastic differential equations.

\end{abstract}
 Keywords: forward-backward stochastic differential system, Navier-Stokes equation, Feynman-Kac formula, strong solution, Lagrangian approach, variational formulation.
\section{Introduction} \label{intro}

Consider the following Cauchy problem for deterministic backward Navier-Stokes equation for  the velocity field of an incompressible, viscous fluid:
\begin{equation}\label{backward NS}
  \left\{\begin{array}{l}
    \begin{split}
      &\partial_t u + \frac{\nu}{2}\Delta u + (u\cdot \nabla)u + \nabla p + f=0, \;\; t \leq T;\\
     &\nabla \cdot u = 0, \quad u(T) = G,
    \end{split}
  \end{array}\right.
\end{equation}
which is obtained from the classical Navier-Stokes equation via the time-reversing transformation
$$
    (u,p,f)(t,x)\longmapsto
    (-u,p,f)(T-t,x),\quad \textup{for} \; t \leq T.
$$
 Here, $T>0$, $u$ is the $d$-dimensional velocity field of the fluid, $p$ is the pressure field, $\nu\in(0,\infty)$ is the kinematic viscosity, and $f$ is the external force field which, without any loss of generality, is taken to be divergence free. It is well-known that the Navier-Stokes equation was introduced by Navier~\cite{Navier1822} and Stokes~\cite{Stokes1849} via adding a dissipative term $\nu \Delta u$ as the friction force
to Euler's equation, which is Newton's law for an infinitesimal volume element of the fluid.

Forward-backward stochastic differential equations (FBSDEs) are already well-known nowadays to be connected to systems of nonlinear parabolic partial differential equations (PDEs) (see among many others \cite{Antonelli93,cheridito2007second,delarue2002existence,HuPengFK95,MaProtterYong1994,ParPeng1992QPDE,PardouxTangFK99,Tang_05,YongFK1997}). Within such a theory, the $d$-dimensional Burgers' equation (in the backward form)
\begin{equation}\label{backward burgers-eq}
  \left\{\begin{array}{l}
          \partial_t v + \frac{\nu}{2} \Delta v + (v\cdot \nabla)v  +f=0 , \;\; t \leq T;\\
         v(T) = \phi,
    \end{array}\right.
\end{equation}
as a simplified version of Naview-Stokes equation~\eqref{backward NS}, is associated in a straightforward way to the following coupled FBSDE:
\begin{equation}\label{FBSDE-burgers}
  \left\{\begin{array}{l}
  \begin{split}
  dX_s(t,x)&
            =Y_s(t,x)\,ds+\sqrt{\nu}\,dW_s,\quad s\in[t,T];\\
  X_t(t,x)&=x;\\
  -dY_s(t,x)&
            =f(s,X_s(t,x))\,ds-\sqrt{\nu}Z_s(t,x)\,dW_s;\\
    Y_T(t,x)&=G(X_T(t,x)). \\
    \end{split}
  \end{array}\right.
\end{equation}
They are related to each other by the following:
\begin{align}\label{reltn-intro-burgers}
Y_s(t,x)=v(s,X_s(t,x)),\,Z_s(t,x)=\nabla v(s,X_s(t,x)), \quad s\in [t,T]\times \mathbb{R}^d
\end{align}
and (see \cite{Tang_05})
\begin{align}\label{reltn-intro-burgers}
Y_s(t,X_s^{-1}(t,x))=v(s,x),\,Z_s(t,X_s^{-1}(t,x))=\nabla v(s,x), \quad s\in [t,T]\times \mathbb{R}^d
\end{align}
with $X_{\cdot}^{-1}(t,x)$ being the inverse of the homeomorphism $X_{\cdot}(t,x), x\in \mathbb{R}^d$.

 It is a tradition to  represent solutions of PDEs as the expected  functionals of stochastic processes. The history is long and the literature is huge. Many studies have been devoted to 
  probabilistic representation to solution of Navier-Stokes equation \eqref{backward NS}, with the following three methodologies. The first is {\it the vortex method}, which aims to give a probabilistic representation first for the vorticity field, and then for the velocity field via the Biot-Savart law (which associates the vorticity field directly to the velocity field, see \cite[pages 71-73]{BertozziMajda2002}).
   In the two-dimensional case ($d=2$),  the vorticity turns out to obey a Fokker-Planck type parabolic PDE, and its probabilistic interpretation is straightforward. In this line, see Chorin \cite{Chorin73} who used random walks and a particle limit to represent the vorticity field, and Busnello~\cite{Busnello1999} who used the Girsanov transformation to give a probabilistic representation of the vorticity field, and used the Bismut-Elworthy-Li formula to derive a probabilistic interpretation for the Biot-Savart law. In the three-dimensional case ($d=3$), the vorticity field is still found to evolve as a parabolic PDE, but it is complicated by the addition of the stretching term; Esposito et al. \cite{Esposito-Sciarretta1988,EspositoMarra_1989} proposed a probabilistic representation formula for the vorticity field and then for the velocity field without any further probabilistic representation for the Biot-Savart law.  Busnello et al.~\cite{BusnelloFlandoliRomito05} used the Bismut-Elworthy-Li formula  to give a probabilistic interpretation for the Biot-Savart law and then for the velocity field. The second is {\it the Fourier transformation method}.  Le Jan and Sznitman \cite{LeJanSznitman1997} interpreted the Fourier transformation of the Laplacian of the three-dimensional velocity field in terms of a backward branching process and a composition rule along the associated tree, and got a new existence theorem, and their approach was extensively studied and generalized by others (see, for instance \cite{Bhattacharya-Waymire2003,Ossiander2005}). The third is {\it the Lagrangian flows method}, and see Constantin and Iyer~\cite{Constantin-Iyer-2008} and Zhang~\cite{Zhang-bNS-2010}.

Navier-Stokes equation~\eqref{backward NS} typically has a divergence-free constraint and contains a pressure potential to complement the thus-lost degree of freedom.  Since the pressure in equation~\eqref{backward NS} turns out to be determined by the Poisson equation (as a consequence of the incompressibility):
$$
\Delta {p}=-\textrm{div div }(  {u}\otimes {u}),
$$
the Navier-Stokes equations \eqref{backward NS} has the following equivalent form:
 \begin{equation}\label{NS-Intg-Opr-form}
  \left\{\begin{array}{l}
          \partial_t  {u} + \frac{\nu}{2} \Delta  {u} + ( {u}\cdot \nabla) {u}
           +\nabla(-\Delta)^{-1}\textrm{div div }( {u}\otimes {u}) +f=0 , \;\; t \leq T;\\
          {u}(T) = G.
    \end{array}\right.
\end{equation}
 In comparison with Burgers' equation~\eqref{backward burgers-eq}, it has an extra nonlocal operator appearing in its dynamic equation. To give a fully probabilistic representation for the Navier-Stokes equations, we have to incorporate this additional term. In this paper, we associate the Navier-Stokes equation~\eqref{backward NS} to the following coupled forward-backward stochastic differential system (FBSDS):
\begin{equation}\label{1.1}
  \left\{\begin{array}{l}
  \begin{split}
  dX_s(t,x)&
            =Y_s(t,x)\,ds+\sqrt{\nu}\,dW_s,\quad s\in[t,T];\\
  X_t(t,x)&=x;\\
  -dY_s(t,x)&
            =\left[f(s,X_s(t,x))+\widetilde{Y}_0(s,X_s(t,x))\right]\,ds-\sqrt{\nu}Z_s(t,x)\,dW_s;\\
    Y_T(t,x)&=G(X_T(t,x));\\
  -d\widetilde{Y}_{s}(t,x)&=\!\!\sum_{i,j=1}^d \frac{27}{2s^3}Y^i_tY^j_t(t,x+B_{s})\left(B^i_{\frac{2s}{3}}-B^i_{\frac{s}{3}}\right)
              \left(B_{s}^j-B_{\frac{2s}{3}}^j\right)B_{\frac{s}{3}} \,ds\\
              &\ \ -\widetilde{Z}_s(t,x)dB_s,\quad s\in(0,\infty);\\
    \widetilde{Y}_{\infty}(t,x)&=0.
    \end{split}
  \end{array}\right.
\end{equation}
Here, $B$ and $W$ are two independent  $d$-dimensional standard Brownian motions, $Y$ and $\widetilde{Y}$ satisfy backward stochastic differential equations
 and  $X$ satisfies a forward one. The forward SDE describes a stochastic particle system, and the BSDE in the finite time interval specifies the evolution of the velocity.   The drift part of $\{Y_s(t,x), s\in [t,T]\}$ (see the third equality of FBSDS~\eqref{1.1}) at time $s$  depends on ${\widetilde Y}_0$, and that of $\{{\widetilde Y}_s(t,x), s\in (0,\infty)\}$ depends on $Y_t(t, x+B_s)$, which make our system \eqref{1.1} differ from the conventional coupled forward-backward stochastic differential equations (FBSDEs) (see \cite{Antonelli93,HuPengFK95,MaProtterYong1994,ParPeng1992QPDE,Peng1991_QPDE,PardouxTangFK99,YongFK1997}). Furthermore, a BSDE in the infinite time interval is introduced to express the integral operator $\nabla(-\Delta)^{-1}\textrm{div div}$ in a probabilistic manner, and both backward stochastic differential equations (BSDEs) in  FBSDS~\eqref{1.1} are defined on two different time-horizons $[t,T]$ and $(0, \infty)$.

We have

\begin{thm}\label{thm local-intro}
  Let $G\in H^m_{\sigma},$ and $f\in L^2(0,T;H^{m-1}_{\sigma})$ with $m>d/2$. Then there is  $T_0<T$ which depends on $\|f\|_{L^2(0,T;H^{m-1})}$, $\nu$, $m$, $d$, $T$ and $\|G\|_m$, such that FBSDS \eqref{1.1} admits a unique $H^m$-solution $(X,Y,Z,\widetilde{Y}_0)$ on $(T_0,T]$.
  For the solution, there hold the following representations
  \begin{align}
  Z_t(t,\cdot)=\nabla Y_t(t,\cdot),Y_s(t,\cdot)= Y_s(s,X_s(t,\cdot)) \textrm{ and } Z_{s}(t,\cdot):=Z_s(s,X_s(t,\cdot)),
  \label{thm-relat-Y-Z}
  \end{align}
  for $T_0<t\leq s\leq T$. Moreover, there exists some scalar-valued function $ {p}$ such that $\nabla  {p}=\widetilde{Y}_0$, and  $( {u}, {p})$ with $ {u}(t,x):=Y_t(t,x)$ coincides with the unique strong solution to the Navier-Stokes equation \eqref{backward NS}.
  \end{thm}

  In the last theorem, all unknown forward and backward states of the concerned FBSDS evolve in spaces of fields, and the conditions on  $f$ and $G$ are much weaker than those of the existing related results on coupled FBSDEs (see \cite{Antonelli93,delarue2002existence,HuPengFK95,MaProtterYong1994,PardouxTangFK99,pengwu1999fully,YongFK1997})---which usually require that  $f$ and $G$  are either bounded or uniformly Lipschitz continuous in the space variable $x$. Related results and the maximum principle for PDEs of Burgers' type are also presented in this paper.

 FBSDS~\eqref{1.1} is a complicated version of FBSDE~\eqref{FBSDE-burgers}, including an additional nonlinear and nonlocal term in the drift of the BSDE to keep the backward state living in the divergence-free subspace. While the additional term causes difficulty in formulating probabilistic representations, it helps us to obtain the global solutions if either the Reynolds number is small or the dimension is equal to two.

   Our relationship between FBSDS~\eqref{1.1}  and Navier-Stokes equation~\eqref{backward NS} is shown to imply a probabilistic Lagrangian representation for the velocity field, which coincides with the formulas given by \cite{Constantin-Iyer-2008,Constantin-Iyer-arxiv-2011,Zhang-bNS-2010} and weakens the regularity assumptions required in the references (see Remark \ref{rmk-connct-Lagrgn} below). On the other hand, in the spirit of the variational interpretations for Euler equations by Arnold \cite{Arnold-1966}, Ebin and Marsden~\cite{ebin1970groups} and Bloch et al. \cite{bloch2000optimal},  Inoue and Funaki~\cite{inoue1979new},  Yasue \cite{yasue1983variational} and Gomes~\cite{gomes2005variational} formulated different stochastic variational principles for the Navier-Stokes equations. Along this direction, we give a new stochastic variational formulation for the Navier-Stokes equations on basis of our probabilistic Lagrangian representation.

    Other quite related works include Albeverio and Belopolskaya \cite{albeverio2007generalized} who constructed a weak solution of the 3D Navier-Stokes equation by solving the associated stochastic system with the approach of stochastic flows, Cruzeiro and Shamarova \cite{Cruzeiro-Shamarova-2009} who established a connection between the strong solution to the spatially periodic Navier-Stokes equations and a solution to a system of FBSDEs on the group of volume-preserving diffeomorphisms of a flat torus, and Qiu, Tang and You \cite{QiuTangYou-SPA-2011} who considered a similar non-Markovian FBSDS to ours \eqref{1.1} in the two-dimensional spatially periodic case, and studied the well-posedness of the corresponding backward stochastic PDEs.

    The rest of this paper is organized as follows. In Section 2, we introduce notations and functional spaces, and give auxiliary  results. In Section 3, the solution to FBSDS \eqref{1.1} is defined in suitable spaces of fields, and our main result (Theorem \ref{thm local}) is stated on the FBSDS associated to the Navier-Stokes Equation. In Section 4,  we discuss the coupled FBSDEs for PDEs of Burgers' type, and related results and the maximum principle are presented. In Section 5, Theorem \ref{thm local} is proved,  and the  global existence and uniqueness of the solution  is given if either the Reynolds number is small or the dimension of the forward SDE is equal to two. By truncating the time interval of the FBSDS, we approximate in Section 6 the Navier-Stokes equation by a class of FBSDSs and associated PDEs. In Section 7, from our relationship between FBSDS and the Navier-Stokes equation, we derive a probabilistic Lagrangian representation for the velocity field, which is shown to imply those of~\cite{Constantin-Iyer-2008,Zhang-bNS-2010}, and we also give a variational characterization of the the Navier-Stokes equation. Finally in Section 8 as an appendix,  we prove Lemmas \ref{lem norm-equivalence} and \ref{lem_verif}.

\section{Preliminaries}\label{sec:prelm}
\subsection{Notations}
Let $(\Omega,\bar{\sF},\{\bar{\sF}_t\}_{t\geq0},\bP)$ be a complete filtered
probability space on which are defined two independent $d$-dimensional standard Brownian motions $W=\{W_t:t\in[0,\infty)\}$ and $B=\{B_t:t\in[0,\infty)\}$ such that $\{\bar{\sF}_t\}_{t\geq 0}$ is the natural
filtration generated by $W$ and $B$, and augmented by all the $\bP$-null sets in
$\bar{\sF}$. By $\{\sF\}_{t\geq 0}$ and $\{\sF^B\}_{t\geq 0}$, we denote the natural filtration generated by $W$ and $B$ respectively, and they are both augmented by all the $\bP$-null sets.  $\sP$ is the $\sigma$-algebra of the predictable sets on $\Omega\times[0,T]$ associated with $\{\sF_t\}_{t\geq0}$.

  The set of all the integers is denoted by $\mathbb{Z}$, with $\mathbb{Z}^+$ the subset of the positive elements and $\mathbb{N}:=\mathbb{Z}^+\cup \{0\}$. Denote by $|\cdot|$ (respectively,
$\langle \cdot,\cdot \rangle$ or $\cdot$) the norm (respectively, scalar product) in
finite-dimensional Hilbert space such as $\bR,\bR^k,\bR^{k\times l}$, $k,l\in \mathbb{Z}^+$ and
$$|x|:=\left( \sum_{i=1}^k x_i^2  \right)^{1/2} \quad
\textrm{and}\quad  |y|:=\left(\sum_{i=1}^k\sum_{j=1}^l y_{ij}^2
\right)^{1/2} \quad \textrm{for}~ (x,y)\in \bR^k \times
\bR^{k\times l}.
$$

For each Banach space $(\mathcal {X},\|\cdot\|_{X})$ and real $q\in[1,\infty]$, we denote by $S^q([t,\tau];\mathcal{X})$ the set of $\mathcal{X}$-valued, $\sF_t$-adapted and c\`{a}dl\`{a}g processes $\{X_s\}_{s\in[t,\tau]}$ such that
 \begin{equation*}
  \|X\|_{S^q([t,\tau];\mathcal{X})}:=
  \left\{\begin{array}{l}
          E\Big[  \sup_{s\in[t,\tau]} \|X_s\|_{\mathcal{X}}^q  \Big]^{1/q}<\infty,\quad q\in[1,\infty);\\
         \esssup_{\omega\in\Omega}  \sup_{s\in[t,\tau]} \|X_s\|_{\mathcal{X}}<\infty,\quad q=\infty.
    \end{array}\right.
\end{equation*}
$L_{\sF}^q(t,\tau;\mathcal{X})$ denotes the set of (equivalent classes of) $\mathcal{X}$-valued predictable processes
$\{X_s\}_{s\in[t,\tau]}$ such that
 \begin{equation*}
  \|X\|_{L^q_{\sF}(t,\tau;\mathcal{X})}:=
  \left\{\begin{array}{l}
          E\Big[  \int_t^\tau \|X_s\|_{\mathcal{X}}^q\,ds  \Big]^{1/q}<\infty,\quad q\in[1,\infty);\\
         \esssup_{(\omega,s)\in\Omega\times{[t,\tau]}} \|X_s\|_{\mathcal{X}}  <\infty,\quad q=\infty.
    \end{array}\right.
\end{equation*}
 Both $\left(S^q([t,\tau];\mathcal{X}),\|\cdot\|_{S^q([t,\tau]\mathcal{X})}\right)$ and $\left(L_{\sF}^q(t,\tau;\mathcal{X}),\|\cdot\|_{L_{\sF}^q(t,\tau;\mathcal{X})}\right)$ are Banach spaces.

Define the set of multi-indices
$$\mathcal {A}:=\{\alpha=(\alpha_1,\cdots,\alpha_d): \alpha_1, \cdots, \alpha_d \textrm{
are nonnegative integers}\}.$$ For any $\alpha\in \mathcal{A}$ and
$x=(x_1,\cdots,x_d)\in \bR^d,$ denote
$$ |\alpha|=\sum_{i=1}^d
\alpha_i,\ x^{\alpha}:=x_1^{\alpha_1}x_2^{\alpha_2}\cdots x_d^{\alpha_d},\ D^{\alpha}:=\frac{\partial^{|\alpha|}}
{\partial x_1^{\alpha_1}\partial x_2^{\alpha_2}\cdots\partial x_d^{\alpha_d}}. $$
For  differentiable transformations $\phi,\psi$ on $\bR^d$,  define the Jacobi matrix  $\nabla \phi$ of $\phi$:
\begin{equation*}
  \nabla \phi =
  \left(\begin{array}{l}
  \begin{split}
    &\partial_{x_1}\phi^1,\partial_{x_2}\phi^1,\cdots,\partial_{x_d}\phi^1\\
    &\partial_{x_1}\phi^2,\partial_{x_2}\phi^2,\cdots,\partial_{x_d}\phi^2\\
    &\cdots,\ \ \cdots,\ \ \cdots,\ \ \cdots\\
    &\partial_{x_1}\phi^d,\partial_{x_2}\phi^d,\cdots,\partial_{x_d}\phi^d
    \end{split}
  \end{array}\right)
\end{equation*}
whose transpose is denoted by $\nabla^{\mathcal {T}}\phi$, the divergence
$\textrm{div} \phi=\nabla\cdot\phi$, and the matrix
\begin{equation*}
  \phi\otimes\psi =
  \left(\begin{array}{l}
  \begin{split}
    &\phi^1\psi^1,\phi^1\psi^2,\cdots,\phi^1\psi^d\\
    &\phi^2\psi^1,\phi^2\psi^2,\cdots,\phi^2\psi^d\\
    &\cdots,\ \ \cdots,\ \ \cdots,\ \ \cdots\\
    &\phi^d\psi^1,\phi^d\psi^2,\cdots,\phi^d\psi^d
    \end{split}
  \end{array}\right).
\end{equation*}

Now we extend several  spaces of real-valued functions to those of vector-valued functions. For $l,k\in\mathbb{Z}^+$, we denote by $C^{\infty}_{c}(\bR^l;\bR^k)$ the set of all infinitely
differentiable $\bR^k$-valued functions with compact supports on  $\bR^{l}$ and by $\mathscr{D}'(\bR^l;\bR^k)$ the totality of all the $\bR^k$-valued general functions with each component being Schwartz distribution. For simplicity,  we write $C_c^{\infty}$ and $\mathscr{D}'$ for the case $l=k=d$.
On $\bR^d$ we denote by $\sS$ ($\sS'$, respectively) the set of
all the $\bR^d$-valued functions whose elements are Schwartz functions (tempered
distributions, respectively).
Then the Fourier transform $\cF(f)$ of $f\in \sS$ is given by
$$\cF(f)(\xi)=(2\pi)^{-d/2}\int_{\bR^d}\exp{\left(-\sqrt{-1}\langle x,\,\xi\rangle \right)}f(x)\,dx,
~~~\xi\in \bR^d,$$ and the inverse Fourier transform $\cF^{-1}(f)$ is given
by
$$
    \cF^{-1}(f)(x)=(2\pi)^{-d/2}\int_{\bR^d}\exp{\left(\sqrt{-1}\langle x,\,\xi\rangle\right)}f(\xi)\,d\xi,
~~~x\in \bR^d. $$
Extended to the general function space $\sS'$, the Fourier transform defines an isomorphism from
$\sS'$ onto itself. As usual, for each $s\in \bR$ and $f\in \sS'$, we denote the Bessel potential
$I_{s}(f):=(1-\Delta)^{s/2}f=\cF^{-1}((1+|\xi|^2)^{s/2}\cF(f)(\xi)).$

For $l\in\mathbb{Z}^+$, $m\in \mathbb{N}$, $q\in [1,\infty)$ and  $s\in [1,\infty]$, by $L^s(\bR^l)$ and $H^{m,q}(\bR^l)$ ($L^s$ and $H^{m,q}$,
with a little notional abuse), we denote the usual $\bR^l$-valued Lebesgue and Sobolev spaces on $\bR^d$, respectively. $H^{m,q}$ is equipped with the norm:
 $$
 \|\phi\|_{m,q}:=\Bigg(\|\phi\|_{L^q}^q+\sum_{|\alpha|=1}^m \|D^{\alpha}\phi\|_{L^q}^q  \Bigg)^{1/q},\ \phi\in H^{m,q},
 $$
 which is equivalent to the norm:
 $$\|\phi\|_{m,q}:=\|(1-\Delta)^{\frac{m}{2}}\phi\|_{L^q},\ \phi\in H^{m,q},\quad\textrm{for }q\in (1,\infty).$$
 Both norms will not be distinguished  unless there is a confusion.
 In particular, for the case of $q=2$, $H^{m,2}$ is a Hilbert space with the inner product:
 $$
 \langle\phi,\,\psi\rangle_m:
 =\int_{\bR^d}\langle I_{m/2}\phi(x),\,I_{m/2}\psi(x)\rangle\,dx,\ \phi,\psi\in H^{m,2}.
 $$
 We define the duality between $H^{s,q}$ and $H^{r,q'}$ for $q\in (1,\infty)$ and $q'=q/(q-1)$ as:
 $$
 \langle \phi,\,\psi\rangle_{s,r}:
 =\int_{\bR^d}\langle I_{s/2}\phi(x),\,I_{r/2}\psi(x)\rangle\,dx,\ \phi\in H^{s,q},\psi\in H^{r,q'}.
 $$
 For simplicity, we write the space $H^m$ and the norm $\|\cdot\|_m$ for $H^{m,2}$ and $\|\cdot\|_{m,2}$, respectively.

Define
 $$
 \sD_{\sigma}:=\left\{\phi\in C_c^{\infty}:\ \nabla\cdot\phi=0\right\}.
 $$
 Denote by $H^{m,q}_{\sigma}$  the completion of $\sD_{\sigma}$
 under the norm $\|\cdot\|_{m,q}$,  which is a complete subspace of $H^{m,q}$.

 Now we introduce several spaces of  continuous functions. For  $l\in\mathbb{Z}^+$, $k\in\mathbb{N}$ and domain $\cO\subset \bR^d$, we denote by $C(\cO,\bR^l)$, $C^k(\cO,\bR^l)$ and $C^{k,\delta}(\cO,\bR^l)$ with $\delta\in(0,1)$ the continuous function spaces equipped with the following norms respectively:
\begin{equation*}
  \begin{split}
    &\|\phi\|_{C(\cO,\bR^l)}:=\sup_{x\in \cO}|\phi(x)|,\quad
    \|\phi\|_{C^k(\cO,\bR^l)}:=\|\phi\|_{C(\cO,\bR^l)}+\sum_{|\alpha|=1}^k\|D^{\alpha}\phi\|_{C(\cO,\bR^l)}, \\
    &\|\phi\|_{C^{k,\delta}(\cO,\bR^l)}:=\|\phi\|_{C^{k}(\cO,\bR^l)}
                +\sum_{|\alpha|=k}\sup_{x,y\in\cO,x\neq y}\frac{|D^{\alpha}\phi(x)-D^{\alpha}\phi(y)|}{|x-y|^{\delta}},
  \end{split}
\end{equation*}
 with the convention that $C^0(\cO,\bR^l)\equiv C(\cO,\bR^l)$.  Whenever there is no confusion, we write $C(\bR^d),C^k,$ and $C^{k,\delta}$ for $C(\bR^d,\bR^l),C^k(\bR^d,\bR^l) $ and $C^{k,\delta}(\bR^d,\bR^l)$, respectively.

In an obvious way, we define spaces of Banach space valued functions such as $C(0,T;H^{m,q})$ and $L^r(0,T;H^{m,q})$ for $m\in\mathbb{Z},r,q\in (1,\infty)$, and related local spaces like the following ones:
\begin{align*}
L^r_{\textrm{loc}}(T_0,T;H^{m,q}):&=\mathop{\bigcap}_{T_1\in (T_0, T]}L^r(T_1,T;H^{m,q}),\\
C_{\textrm{loc}}((T_0,T];H^{m,q}):&=\mathop{\bigcap}_{T_1\in (T_0, T]}C([T_1,T];H^{m,q}).
\end{align*}

\subsection{Auxiliary results}

 In the remaining part of the work, we shall use $C$ to denote a constant whose value may vary from line to line, and when needed, a bracket will follow immediately after $C$ to indicate what parameters $C$ depends on. By $A\hookrightarrow B$ we mean that normed space $(A,\|\cdot\|_{A})$ is embedded into  $(B,\|\cdot\|_B)$ with a constant $C$ such that
$$\|f\|_B\leq C \|f\|_A,\,\,\,\forall f\in A.$$

\begin{lem}\label{lem sobolev}
  There holds the following assertions:

  (i)  For integer $n>d/q+k$ with $k\in \mathbb{N}$ and $q\in(1,\infty)$, we have  $H^{n,q}\hookrightarrow C^{k,\delta}$, for any $\delta\in (0,(n-d/q-k)\wedge 1).$

  (ii) If $1<r<s<\infty$ and $m,n\in\mathbb{N}$ such that  $\frac{d}{s}-m=\frac{d}{r}-n$, then $H^{n,r}\hookrightarrow H^{m,s}$.

  (iii) For any $s>d/2$, $H^s$ is a Banach algebra, i.e., there is a constant $C>0$ such that,
  $$
  \|\phi\psi\|_s\leq C\|\phi\|_s\|\psi\|_s,\ \ \forall \phi,\psi\in H^s.
  $$
\end{lem}
The first two assertions are borrowed from the well-known embedding theorem in Sobolev space (see~\cite{TRiebel_83}), and the last one is referred to \cite[Lemma 3.4, Page 98]{BertozziMajda2002}.

\begin{rmk}\label{rmk_m=2}
   Note that $d=2$ or $3$ throughout this work. For any $h,g\in H^2$, we have
   \begin{align*}
     &\|h\cdot g\|_1^2
     \\
     =&\,\|h\cdot g\|_0^2+\sum_{i=1}^d\|\partial_{x_i}h\cdot g\|_0^2
        +\sum_{i=1}^d\|h\cdot\partial_{x_i}g\|_0^2
     \\
     \leq&\,
        C\Big\{
        \|h\|_{C(\bR^d)}^2\|g\|_0^2
        +\|\nabla h\|_0^{2\beta}\|\nabla h\|_1^{2-2\beta}\| g\|_0^{2\beta}\|g\|_1^{2-2\beta}
        +\|h\|_{C(\bR^d)}^2\| g\|_1^2
        \Big\}
        \\
     &\big(\textrm{ using Gagliardo-Nirenberg Inequality (see \cite{FriedmanPDEs,Ladyzhenskaia_68,Nirenberg1959})}\big)
     \\
     \leq&\,
        C \|h\|_2^2\|g\|_1^2,
   \end{align*}
   where $\beta:=1-d/4$, and $H^2\hookrightarrow C^{0,\delta}$ for some $\delta\in(0,1)$. In view of Lemma \ref{lem sobolev}, we have for any integer $m>d/2$,
   $$
   \|hg\|_{m-1}\leq C(m,d) \|h\|_{m}\|g\|_{m-1},\ \forall h\in H^m,g\in H^{m-1}.
   $$
\end{rmk}

Next, we  discuss the the composition of generalized functionals with stochastic flows with Sobolev space-valued coefficients. Assume that $\nu>0$ and that
\begin{equation}\label{Assump on b}
b\in C([0,T];H^m)\cap L^2(0,T;H^{m+1}),\phi\in L^2(0,T;L^2(\bR^d)),\psi\in L^2(\bR^d),
\end{equation}
for some integer $m>d/2$. Consider the following FBSDE:
\begin{equation}\label{FBSDS_simpp}
  \left\{\begin{array}{l}
  \begin{split}
  dX_s(t,x)&
            =b(s,X_s(t,x))\,ds+\sqrt{\nu}\,dW_s,\quad s\in[t,T];\\
  X_t(t,x)&=x;\\
  -dY_s(t,x)&
            =\phi(s,X_s(t,x))\,ds-\sqrt{\nu}Z_s(t,x)\,dW_s,\quad s\in[t,T];\\
    Y_T(t,x)&=\psi(X_T(t,x)).
    \end{split}
  \end{array}\right.
\end{equation}
 Since $H^m\hookrightarrow C^{0,\delta}$ and $H^{m+1}\hookrightarrow C^{1,\delta}$ for $m>d/2$, in view of \cite[Theorems 3.4.1 and 4.5.1]{Kunita_book}, the \emph{forward} SDE is well posed for each $(t,x)\in [0,T]\times\bR^d$,  and the unique solution in relevance to the initial data $(t,x)\in [0,T]\times\bR^d$ defines a stochastic flow of homeomorphisms. Since the function $\phi$ is only  measurable, the following lemma serves to justify the composition $\phi(s,X_s(t,x))$.

 \begin{lem}\label{lem norm-equivalence}
   Assume that $m>d/2$ and $b\in C([0,T];H^m)\cap L^2(0,T;H^{m+1})$. Then
   for all $t\in [0,T], s\in [t, T], (\varphi, \eta)\in L^1(\bR^l)\times L^1([0,T]\times\bR^d;\bR^l)$, $l\in\mathbb{Z}^+$, we have
   \begin{align}
     \kappa \|\varphi\|_{L^1(\bR^l)}
                \leq &  \int_{\bR^d}\!\! E\big[ |\varphi(X_s(t,x))| \big]\,dx
                        \leq  \kappa^{-1}  \|\varphi\|_{L^1(\bR^l)},\label{eq_special}\\
     \lambda \|\eta\|_{L^1([t,T]\times\bR^l)}
                \leq \int_{\bR^d}&  \int_t^T\!\! E\big[ |\eta(s,X_s(t,x))| \big]\,dsdx
                        \leq \lambda^{-1}  \|\eta\|_{L^1([t,T]\times\bR^l)}, \label{eq_time}
   \end{align}
   with $\kappa=e^{-\|\textrm{div}\,b\|_{L^1(t,s;L^{\infty})}}$ and
   $\lambda=e^{-\|\textrm{div}\,b\|_{L^1(t,T;L^{\infty})}}$.
     \end{lem}

  Lemma \ref{lem norm-equivalence} weakens the  assumptions on $b$ of   \cite[Theorem 14.3]{BarlesLesigne_BSDEPDE97_inbook}, where $b(t, \cdot)\equiv b(\cdot)$ is time invariant and is  required to lie in $C^1(\bR^d)\cap L^2(\bR^d)$. Since $b(t,x)$ is not necessarily  uniformly Lipschitz continuous in $x$, the stability of $X$ with respect to the coefficient $b$ has to be proved very carefully  and the proof of \cite[Theorem 14.3]{BarlesLesigne_BSDEPDE97_inbook} has to be generalized accordingly. We give a probabilistic proof in the appendix.

 \begin{rmk}\label{rmk_norm}
   From Lemma \ref{lem norm-equivalence}, we see that Lebesgue's measure transported by the flow $\{X_s(t,x),s\in[t,T]\}$ results in a group of measures $\{\mu_s,s\in[t,T]\}$ satisfying for any Borel measurable set $A\subset \bR^d$,
   $$
   \mu_s(A)=\int_{\bR^d}E\left[1_{A}(X_s(t,x))\right]\,dx.
   $$
   These measures are all equivalent to Lebesgue measure and the exponential rate of compression or dilation are governed by the divergence of $b$. In particular, when $b$ is divergence free, $X_s(t,\cdot)$ preserves the Lebesgue measure for all times. This is similar to that of a system of ordinary differential equations (see \cite{DipernaLions1989}). On the other hand, thanks to Lemma \ref{lem norm-equivalence}, our FBSDE \eqref{FBSDS_simpp} makes sense under assumption \eqref{Assump on b}, i.e., the \emph{forward} SDE is well posed for each $(t,x)\in[0,T]\times\bR^d$ and for each $t\in[0,T]$ the BSDE is well posed for almost every $x\in\bR^d$.
 \end{rmk}

\section{FBSDS associated with Navier-Stokes Equation}\label{sec:conet}

\begin{defn}\label{def Hm solution}
 Let $T_0<T$. $(X,Y,Z,\widetilde{Y}_0)$ is called an  $H^{m}$-solution to  FBSDS~\eqref{1.1} in $[T_0,T]$ if  for almost every $(t,x)\in[T_0,T]\times\bR^d$,
   $$
   (X_{\cdot}(t,x),Y_{\cdot}(t,x),Z_{\cdot}(t,x))\in
   S^2([t,T];\bR^d)\times S^2([t,T];\bR^d)\times L^2_{\sF}(t,T;\bR^{d\times d})
   $$
   and for each $t\in[T_0,T]$  and almost every $x\in\bR^d$, $\{Y_{s}(t,x),\,s\in[t,T]\}\in L^2(t,T;L^{\infty}(\Omega;\bR^d))$, such that all the stochastic differential equations of \eqref{1.1} hold in It\^o's sense and $\widetilde{Y}_0(t,x):=\lim_{\eps \downarrow 0} E\widetilde{Y}_{\eps}(t,x)$ exists for almost every $(t,x)\in(T_0,T]\times\bR^d$ with $\widetilde{Y}_0\in L^2_{\textrm{loc}}(T_0,T;H^{m-1})$.
\end{defn}

Our main result is stated as follows.

\begin{thm}\label{thm local}
  Let $\nu>0, G\in H^m_{\sigma},$ and $f\in L^2(0,T;H^{m-1}_{\sigma})$ with $m>d/2$. Then there is $T_0<T$ which depends on $\|f\|_{L^2(0,T;H^{m-1})}$, $\nu$, $m$, $d$, $T$ and $\|G\|_m$, such that FBSDS~\eqref{1.1} has a
  unique $H^m$-solution $(X,Y,Z,\widetilde{Y}_0)$ on $(T_0,T]$ with the function
  $$\{Y_t(t,x),\,(t,x)\in(T_0,T]\times\bR^d\}\in C_{loc}((T_0,T];H^m_{\sigma})\cap L^2_{loc}(T_0,T;H^{m+1}_{\sigma}).$$
   Moreover, we have the following representations
  \begin{align}
   Z_t(t,\cdot)=\nabla Y_t(t,\cdot),Y_s(t,\cdot)= Y_s(s,X_s(t,\cdot)) \textrm{ and } Z_{s}(t,\cdot):=Z_s(s,X_s(t,\cdot)),
  \label{thm-relat-Y-Z}
  \end{align}
  for $T_0<t\leq s\leq T$, and there is some scalar-valued function ${p}$ such that $\nabla {p}=\widetilde{Y}_0$ and $(Y,Z,{p})$ satisfies
  \begin{align}
    Y_r(r,X_r(t,x))=&
    \,G(X_T(t,x))
    +\int_r^T \left[f(s,X_s(t,x))+\nabla {p}(s,X_s(t,x))    \right]\,ds
    \nonumber\\
    &-\sqrt{\nu}\int_r^T Z_s(s,X_s(t,x))\,dW_s,\,T_0<t\leq r\leq T,\,a.e.x\in\bR^d,\,a.s.,
  \label{Eq thm bNS-FBSDS}
  \end{align}
and  $({u},{p})$ with ${u}(t,x):=Y_t(t,x)$ is the unique strong solution to Navier-Stokes equation~\eqref{backward NS} on $(T_0,T]$.
  \end{thm}

As indicated in the introduction,  Navier-Stokes equation~\eqref{backward NS} is equivalent to \eqref{NS-Intg-Opr-form}:
 \begin{equation*}
  \left\{\begin{array}{l}
          \partial_t  {u} + \frac{\nu}{2} \Delta {u} + ( {u}\cdot \nabla) {u}
           +\nabla(-\Delta)^{-1}\textrm{div div }( {u}\otimes {u}) +f=0 , \;\; t \leq T;\\
          {u}(T) = G.
    \end{array}\right.
\end{equation*}
To give a fully probabilistic solution of Navier-Stokes equation~\eqref{backward NS}, we shall first give a probabilistic representation for the nonlocal operator $\nabla(-\Delta)^{-1}\textrm{div div}$. Note that a different probabilistic formulation for $\nabla {p}=\nabla(-\Delta)^{-1}\textrm{div div }( {u}\otimes {u})$ was given by Albeverio and Belopolskaya \cite{albeverio2007generalized} for $d=3$.

\begin{lem}\label{lem singular operator}
  For $\phi,\psi\in H^m$ with $m> \frac{d}{2}+1$, the following BSDE :
   \begin{equation}\label{bsde lem SIO}
  \left\{\begin{array}{l}
  \begin{split}
	  -d\widetilde{Y}_s(x)&=\frac{27}{2s^3}\sum_{i,j=1}^d\phi^i\psi^j(x+B_s)\left(B_s^j-B_{\frac{2s}{3}}^j\right)
\left(B^i_{\frac{2s}{3}}-B^i_{\frac{s}{3}}\right)
B_{\frac{s}{3}} \,ds\\
              &\ \ -\widetilde{Z}_s(x)dB_s,\ s\in(0,\infty);                         \\
	      \widetilde{Y}_{\infty}(x)&=0\ \
    \end{split}
  \end{array}\right.
\end{equation}
is well-posed on $(0,\infty)$ and $\widetilde{Y}_0(x):=\lim_{\eps\downarrow 0}E\widetilde{Y}_\eps(x)$ exists for each $x\in\bR^d$. Moreover,
$\widetilde{Y}_0\in C(\bR^d)$ and
\begin{equation*}
	\widetilde{Y}_0(x)=\nabla (-\Delta)^{-1}\textrm{div }\textrm{div}(\phi\otimes\psi)(x)
  =\sum_{i,j=1}^d\nabla (-\Delta)^{-1}\partial_{x^i}\partial_{x^j}(\phi^i(x)\psi^j(x)),\ \forall x\in\bR^d.
\end{equation*}
\end{lem}

\begin{proof}
  For $m>d/2+1$, $H^m$ is a Banach algebra embedded into $H^{2,\gamma}$ for some $\gamma>d$ and also into $C^{1,\delta}(\bR^d)$ for some $\delta\in(0,1)$. Thus, $\phi^i\psi\in H^m\cap H^{m,1}$ and $\partial_{x^i}\partial_{x^j}(\phi^i(x)\psi^j(x))\in H^{0,\gamma}(\bR)\cap H^{0,1}(\bR)$, $i,j=1,\cdots,d$.

   For each  $\eps>0$,
	\begin{equation*}
		\begin{split}
			&E\left[\int_{\eps}^{\infty} \frac{27}{2s^3}\left|\phi^i\psi^j(x+B_s)\left(B^i_{\frac{2s}{3}}-B^i_{\frac{s}{3}}\right)
              	\left(B_s^j-B_{\frac{2s}{3}}^j\right)B_{\frac{s}{3}}\right| \,ds\right]\\
	      \leq &
	      	\ \|\phi^i\psi^j\|_{L^{\infty}}
			\int_{\eps}^{\infty} \frac{27}{2s^3}E\left[\left|\left(B^i_{\frac{2s}{3}}-B^i_{\frac{s}{3}}\right)
              	\left(B_s^j-B_{\frac{2s}{3}}^j\right)B_{\frac{s}{3}}\right|\right] \,ds\\
	      \leq &
	      \ C\,\|\phi\otimes\psi\|_2 \int_{\eps}^{\infty}\frac{1}{s^{3/2}} \ ds<\infty,\quad i,j=1,\cdots,d.
      \end{split}
	\end{equation*}
	Thus, BSDE \eqref{bsde lem SIO} is well-posed on $[\eps,\infty]$.

  Note that
  \begin{align*}
	  \widetilde{Y}_{\eps}(x)=E\left[\int_{\eps}^{\infty}  \frac{27}{2s^3}\sum_{i,j=1}^d\phi^i\psi^j(x+B_s)\left(B^i_{\frac{2s}{3}}-B^i_{\frac{s}{3}}\right)
              \left(B_s^j-B_{\frac{2s}{3}}^j\right)B_{\frac{s}{3}}   \,ds\Big|\sF^B_{\eps}\right].
  \end{align*}
  Applying the integration-by-parts formula, we obtain
  \begin{align*}
      &E\left[ \frac{27}{2s^3}\phi^i\psi^j(x+B_s)\left(B^i_{\frac{2s}{3}}-B^i_{\frac{s}{3}}\right)
              \left(B_s^j-B_{\frac{2s}{3}}^j\right)B_{\frac{s}{3}}   \right]\\
      =&
      \frac{27}{2s^3}\!\int_{\bR^d}\!\int_{\bR^d}\!\int_{\bR^d}\!\!\!
      \phi^i\psi^j(x+y+z+r)y^iz^jr^k(2\pi s/3)^{-3d/2}e^{-\frac{3(|y|^2+|z|^2+|r|^2)}{2s}}\,dydzdr  \\
      =&
      -\frac{9}{2s^2}\!\int_{\bR^d}\!\int_{\bR^d}\!\int_{\bR^d}\!\!\!
      \phi^i\psi^j(y)z^jr^k(2\pi s/3)^{-3d/2}\partial_{y^i}e^{-\frac{3(|y-x-z-r|^2+|z|^2+|r|^2)}{2s}}\,dydzdr\\
      =&\cdots\\
      =&
      \frac{1}{2s}\!\int_{\bR^d}\!\!\!
      \partial_{x^i}\partial_{x^j}(\phi^i\psi^j)(y+x)y^k(2\pi s)^{-d/2}e^{-\frac{|y|^2}{2s}}\,dy\\
      =&
      \frac{1}{2s}E\left[\partial_{x^i}\partial_{x^j}(\phi^i\psi^j)(x+B_s)B^k_s\right],\quad s>0,\,i,j,k=1,\cdots,d.
  \end{align*}
  As for $i,j=1,\cdots,d$
  \begin{equation}
    \begin{split}
      &s^{-1}\int_{\bR^d}
      |\partial_{x^i}\partial_{x^j}(\phi^i\psi^j)(y+x)y(2\pi s)^{-d/2}e^{-\frac{|y|^2}{2s}}|\,dy\\
      \leq \,&
      {C}\,s^{-\frac{d}{2}-1}\|\partial_{x^i}\partial_{x^j}(\phi^i\psi^j)\|_{0,q}\sqrt{s}s^{\frac{d}{2}(1-\frac{1}{q})}
      \\
      \leq \,&
      C\,\|\partial_{x^i}\partial_{x^j}(\phi^i\psi^j)\|_{0,q}s^{-\frac{1}{2}-\frac{d}{2q}},\ q\in[1,\gamma],
    \end{split}
  \end{equation}
  and
  \begin{align}
      &\int_0^{\infty}s^{-1}E\left[\left|\partial_{x^i}\partial_{x^j}(\phi^i\psi^j)(x+B_s)B^k_s
      \right|\right]\,ds
      \nonumber\\
      =\,&
      \int_0^{\infty}s^{-1}\int_{\bR^d}
      |\partial_{x^i}\partial_{x^j}(\phi^i\psi^j)(y+x)y^k(2\pi s)^{-d/2}e^{-\frac{|y|^2}{2s}}|\,dy\,ds\nonumber\\
      \leq\,&
      C\int_0^1\|\partial_{x^i}\partial_{x^j}(\phi^i\psi^j)\|_{0,\gamma}\,s^{-\frac{1}{2}-\frac{d}{2\gamma}}\,ds
      +C\int_1^{\infty}\|\partial_{x^i}\partial_{x^j}(\phi^i\psi^j)\|_{0,1}\,s^{-\frac{d+1}{2}}\,ds
      \nonumber\\
      \leq\,&
      C\left(\|\partial_{x^i}\partial_{x^j}(\phi^i\psi^j)\|_{m-2}
      +\|\partial_{x^i}\partial_{x^j}(\phi^i\psi^j)\|_{m-2,1}\right),\label{a}
  \end{align}
  we have
  \begin{align}
	    \lim_{\eps\downarrow 0} E\left[ \widetilde{Y}_{\eps}^k(x)\right]
	    =\,&\lim_{\eps\downarrow 0} \int_{\eps}^{\infty}\frac{1}{2s}\int_{\bR^d}
      \partial_{x^i}\partial_{x^j}(\phi^i\psi^j)(y+x)y^k(2\pi s)^{-d/2}e^{-\frac{|y|^2}{2s}}\,dyds
      \nonumber\\
      =\,&
      \int_{0}^{\infty}\frac{1}{2s}\int_{\bR^d}
      \partial_{x^i}\partial_{x^j}(\phi^i\psi^j)(y+x)y^k(2\pi s)^{-d/2}e^{-\frac{|y|^2}{2s}}\,dyds
      \nonumber\\
      &\textrm{(by \eqref{a} and Fubini Theorem)}
      \nonumber\\
      =\,&
      \int_{\bR^d}\partial_{x^i}\partial_{x^j}(\phi^i\psi^j)(x+y)y^k
      \int_0^{\infty}\frac{1}{2s}(2\pi s)^{-d/2}e^{-\frac{|y|^2}{2s}}\,dsdy
      \nonumber\\
      =\,&
  \frac{\Gamma(d/2)}{2\pi^{d/2}}
  \int_{\bR^d}\partial_{x^i}\partial_{x^j}(\phi^i\psi^j)(x+y)\frac{y^k}{|y|^d}
      dy
      \nonumber\\
      =\,&
      \partial_{x^k}(-\Delta)^{-1}\partial_{x^i}\partial_{x^j}
      (\phi^i\psi^j)(x),\nonumber
  \end{align}
  which coincides with the convolution representation of the operator $\nabla(-\Delta)^{-1}$ described in \cite[Page 31]{BertozziMajda2002}.
  Hence,
  BSDE \eqref{bsde lem SIO} are well-posed on $(0,\infty)$ and by \eqref{a}, one has $\widetilde{Y}_0\in C(\bR^d)$ due to the continuity of the translation operator on $L^p(\bR^d)$, $p\in[1,\infty)$. Moreover, we have
  $$\widetilde{Y}_0(x)=\nabla (-\Delta)^{-1}\partial_{x^i}\partial_{x^j}(\phi^i(x)\psi^j(x))=\lim_{\eps\downarrow 0} E\left[ \widetilde{Y}_{\eps}(x) \right], \ \forall x\in\bR^d.$$
The proof is complete.
\end{proof}

\begin{rmk}\label{rmk2 after lem SIO}
 Applying the integration-by-parts formulas in the above proof, we see that BSDE \eqref{bsde lem SIO} gives a probabilistic representation for the operator $\nabla(-\Delta)^{-1}\textrm{div div }$ in the spirit of the Bismut-Elworthy-Li formula (see \cite{ElworthyLi1994}). Its generator does not contain any of its own unknowns and is trivial in its form, while the existence of its solution goes beyond existing results on infinite horizon BSDEs (see \cite{peng2000infinite} and references therein) as the generator may fail to be integrable on the whole time horizon $[0,\infty)$. In fact, the operator $\textbf{P}:=\textbf{I}-\nabla \Delta^{-1} \textrm{div}$
 is the Leray-Hodge projection onto the space of divergence free vector fields, where $\textbf{I}$ is the identity operator.  Define $\textbf{P}^{\perp}:=\textbf{I}-\textbf{P}$.  We have in Lemma \ref{lem singular operator} that $\widetilde{Y}_0=-\textbf{P}^{\perp}(\textrm{div}(\phi\otimes\psi)).$
   Indeed, the singular integral operator $\textbf{P}$ (see \cite{BertozziMajda2002,Stein1970}) is a bounded transformation in  $H^{n,q}$ for $q\in (1,\infty)$ and $n\in\mathbb{Z}$.
    Note that for any $ g\in H^{m}_{\sigma} $, integration-by-parts formula yields
 $$
 \langle \textbf{P}^{\perp}(\textrm{div}(\phi\otimes\psi)), \, g\rangle_{m-2,m}
 =0.
 $$
\end{rmk}

\begin{rmk}\label{rmk after lem SIO}
  There is a scalar function $\eta$  such that $\widetilde{Y}_0=\nabla \eta$. It is sufficient to take
  $$\eta(x)=:(-\Delta)^{-1}\partial_{x^i}\partial_{x^j}(\phi^i\psi^j)(x)\in H^{m,2}(\bR)$$  by the theory of second order Elliptic PDEs (see \cite{GilbargTrud1983}).
\end{rmk}

\begin{rmk}\label{rmk-lem-SIO-eps}
  For any $\eps>0$, we have by Minkowski inequality
  \begin{align}
  \left\|E\left[\widetilde{Y}_{\eps}\right]\right\|_{0}
  =&
  \left\|  \sum_{i,j=1}^d
   E\bigg[
      \int_{\eps}^{\infty} \frac{27}{2s^3}
    \phi^i \psi^j(\cdot+B_s)\left(B^i_{\frac{2s}{3}}-B^i_{\frac{s}{3}}\right)
      	\left(B_s^j-B_{\frac{2s}{3}}^j\right)B_{\frac{s}{3}} \,ds
\bigg]\right\|_0
      \nonumber\\
      \leq &
      \sum_{i,j=1}^d \|\phi^i\psi^j\|_0
      E
      \int_{\eps}^{\infty} \frac{27}{2s^3}
\left|\left(B^i_{\frac{2s}{3}}-B^i_{\frac{s}{3}}\right)
      	\left(B_s^j-B_{\frac{2s}{3}}^j\right)B_{\frac{s}{3}} \right|\,ds
      \nonumber\\
      \leq&
      \frac{27}{\sqrt{\eps}}\sum_{i,j=1}^d \|\phi^i\psi^j\|_0.
  \end{align}
  Putting
  $$
  \textbf{P}^{\eps}(\phi\otimes\psi)=E\left[\widetilde{Y}_{\eps}\right], \quad \forall \phi,\psi\in H^m,\,m> d/2+1,
  $$
  we have
  $$
  \|\textbf{P}^{\eps}(\phi\otimes\psi)\|_k\leq \frac{C}{\sqrt{\eps}}
  \|\phi\otimes\psi\|_k,\quad \forall \,0\leq k\leq m,
  $$
  with the constant $C$ being independent of $\eps$. Then, the operator $\textbf{P}^{\eps}$ can be seen as a regular approximation of $-\textbf{P}^{\perp}\textrm{div}$. This approximation will be used to study  approximate solution of Navier-Stokes equation in Section \ref{sec:numr-apprx}.
\end{rmk}

\section{FBSDEs for PDEs of Burgers' type}

The PDE of Burgers' type:
  \begin{equation}\label{backward-burgers}
  \left\{\begin{array}{l}
          \partial_t u + \frac{\nu}{2} \Delta u + \big((b+\alpha u)\cdot \nabla\big)u +cu+\phi=0 , \;\; t \leq T;\\
          u(T) = \psi,
    \end{array}\right.
\end{equation}
is easily  connected to the following coupled FBSDE:
\begin{equation}\label{FBSDS_simp}
  \left\{\begin{array}{l}
  \begin{split}
  dX_s(t,x)&
            =\left[b(s,X_s(t,x))+\alpha Y_s(t,x)\right]\,ds+\sqrt{\nu}\,dW_s,\quad s\in[t,T];\\
  X_t(t,x)&=x;\\
  -dY_s(t,x)&
            =\left[\phi(s,X_s(t,x))+c(s,X_s(t,x))Y_s(t,x)\right]\,ds-\sqrt{\nu}Z_s(t,x)\,dW_s,\quad s\in[t,T];\\
    Y_T(t,x)&=\psi(X_T(t,x)),\\
    \end{split}
  \end{array}\right.
\end{equation}
where $\nu>0$ and $\alpha$ are constants. The classical Burgers' equation is the case where $\alpha=1$, $b\equiv 0$ and $c\equiv 0$.

 \begin{defn} Let $T_0<T$.
   We say $(X,Y,Z)$ is a solution to FBSDE~\eqref{FBSDS_simp} on $[T_0,T]$ if  for each $t\in [T_0,T]$  and almost every $x\in\bR^d$,
   $$
   (X_{\cdot}(t,x),Y_{\cdot}(t,x),Z_{\cdot}(t,x))\in
   S^2([t,T];\bR^d)\times S^2([t,T];\bR^d)\times L^2_{\sF}(t,T;\bR^{d\times d})
   $$
   and
 such that the \emph{forward} SDE and BSDE on $[t,T]$  hold in It\^o's sense. If for each $t\in[T_0,T]$ and almost every $x\in\bR^d$, we further have
 \begin{equation}\label{eq-def-loc-bd-sol}
     Y_{\cdot}(t,x)\in L^2(t,T;L^{\infty}(\Omega;\bR^d)),
   \end{equation}
   then $(X,Y,Z)$ is called a strengthened  solution.
\end{defn}

\begin{lem}\label{lem_verif}
  Let $b\in C([T_0,T];H^m)\cap L^2(T_0,T;H^{m+1}),\phi\in L^2(T_0,T;H^{m-1})$, and $\psi\in H^m$, with $m>d/2$ and $T_0\in[0,T)$. Then,  the following FBSDE:
   \begin{equation}\label{FBSDE_verif1}
  \left\{\begin{array}{l}
  \begin{split}
  dX_s(t,x)&
            =b(s,X_s(t,x))\,ds+\sqrt{\nu}\,dW_s,\quad T_0\leq t\leq s\leq T;\\
  X_t(t,x)&=x;\\
  -dY_s(t,x)&
            =\phi(s,X_s(t,x))\,ds-\sqrt{\nu}Z_s(t,x)\,dW_s,\quad s\in[t,T];\\
    Y_T(t,x)&=\psi(X_T(t,x))
    \end{split}
  \end{array}\right.
\end{equation}
has a unique  solution $(X,Y,Z)$ such that the function
$$\{Y_t(t,x), (t,x)\in [T_0,T]\times \bR^d\}\quad \in \quad C([T_0,T];H^m)\cap L^{2}(T_0,T;H^{m+1}).$$
 For each $t\in[T_0,T]$, almost all $x\in\bR^d$ and all $r\in [t,T]$, we have
  \begin{align}
  Y_r(r,X_r(t,x))
  &
  =\psi(X_T(t,x))+\!\int_r^T\!\!\!\phi(s,X_s(t,x))\,ds
  -\sqrt{\nu}\int_r^T\!\!\!Z_s(s,X_s(t,x))\,dW_s,\,a.s.,
  \label{PDE_Mild}\\
    Z_t(t,x)&=\nabla Y_t(t,x),\, (Y_r(t,x),Z_r(t,x))=(Y_r,Z_r)(r,X_r(t,x)),\,a.s.. \label{eq expresion}
  \end{align}
Moreover, for any $t\in[T_0,T]$
\begin{equation}\label{eq:lemma6:energy}
  \begin{split}
    &\|{Y}_t(t)\|_m^2+\nu\int_{t}^{T}\|{Z}_s(s)\|_m^2\,ds
    =
    \ \|{Y}_T(T)\|_{m}^2
    +2\int_{t}^{T}
    \langle  \phi(s)+{Z}_sb(s),\,{Y}_s(s)
      \rangle_{m-1,m+1}\,ds.
  \end{split}
\end{equation}
\end{lem}

A proof is sketched in the appendix for the reader's convenience,  though it might exist elsewhere.

\begin{prop}\label{prop FBSDEandPDE}
  Assume that $\psi\in H^m$, $b\in C([0,T];H^{m})\cap L^2(0,T;H^{m+1})$, $c\in L^2(0,T;H^{m}(\bR^{d\times d}))$ and $\phi\in L^2(0,T;H^{m-1})$ with $m>d/2$. Then there is $T_0<T$ which depends on $\|\psi\|_{m}$, $\|b\|_{C([0,T];H^{m})}$, $\|c\|_{L^2(0,T;H^{m})}$, $\|\phi\|_{L^2(0,T;H^{m-1})}$, $\alpha$, $\nu$, $m$, $d$ and $T$, such that FBSDE~\eqref{FBSDS_simp} has a unique strengthened solution on $(T_0,T]$ and if $\alpha=0$, the existence time interval is $[0,T]$.  Moreover,

  (i)
  $$
  \left\{Y_t(t,x),\,(t,x)\in(T_0,T]\times\bR^d\right\}\,\,\in\,\,C_{loc}((T_0,T];H^m)\cap L^2_{loc}(T_0,T;H^{m+1});
  $$

  (ii)
  for almost every $x\in\bR^d$ and all $s\in [t,T]$,
  \begin{align}
    Y_s(s,X_s(t,x))
  =&\psi(X_T(t,x))+\!\int_s^T\!\!\!\left(\phi+cY_r\right)(r,X_r(t,x))\,dr
  -\sqrt{\nu}\int_s^T\!\!\!Z_r(r,X_r(t,x))\,dW_r,\,a.s.,
  \label{PDE_Mild_prop}\\
   Z_t(t,x)=&\nabla Y_t(t,x),\,(Y_s(t,x),\,Z_s(t,x))=(Y_s,\, Z_s)(s,X_s(t,x)),\,a.s.;\label{eq prop express}
  \end{align}

  (iii) for any $t\in(T_0, T]$, we have the following energy equality:
  \begin{equation}\label{Eq prop Energy}
    \begin{split}
      \|Y_t(t)\|_m^2
      =&
      \|\psi\|_m^2+2\int_t^T\langle Z_s(b+\alpha Y_s)(s),\,Y_s(s) \rangle_{m-1,m+1}\,ds\\
      &\ +2\int_t^T \langle c\,Y_s(s)+\phi(s),\, Y_s(s)\rangle_{m-1,m+1}\,ds
      -\nu\int_{t}^T\|Z_s(s)\|_m^2\,ds.
    \end{split}
  \end{equation}
In particular, if $m>d/2+1$, the strengthened solution on the time interval $(T_0,T]$ is the unique solution as well.

In addition, $Y_t(t,x)$ is the unique strong solution to PDE~\eqref{backward-burgers} on $(T_0,T]$.
\end{prop}

\begin{rmk}
When $\alpha=0$, $b\equiv 0$ and $c\equiv 0$, PDE \eqref{backward-burgers} becomes the classical d-dimensional Burgers' equation.  With the method of stochastic flows, Constantin and Iyer \cite{Constantin-Iyer-2008} and Wang and Zhang \cite{wang2012probabilistic} give different stochastic representations for the local regular solutions of Burgers' equations based on stochastic Lagrangian paths. Moreover, in Wang and Zhang \cite{wang2012probabilistic} the global existence results are also presented under certain assumptions on the coefficients. To focus on the Navier-Stokes equation, we shall not search such global results for the PDEs of Burgers' type in this work. Nevertheless, the conditions on the coefficients herein are much weaker than those in \cite{Constantin-Iyer-2008,wang2012probabilistic} where $G$ and $f(t,\cdot)$ are continuous in the space variable and take values in $C^{k+1,\alpha}$ and $H^{k+3,q}\,(\hookrightarrow C^{k+2})$ with $(k,\alpha,q)\in \mathbb{N}\times (0,1)\times (d\vee 2,\infty)$, respectively.  We also note that Cruzeiro and Shamarova \cite{cruzeiro2011_FB_burgers} through a forward-backward stochastic system describes a probabilistic representation of $H^{n}$-regular ($n>\frac{d}{2}+2$) solutions for the spatially periodic forced Burgers' equations.
\end{rmk}

\begin{rmk}\label{rmk_Theta_semiM}
   In Proposition \ref{prop FBSDEandPDE}, we see $Y_t(t,x)$ and $Z_t(t,x)$ are all deterministic functions on $(T_0,T]\times\bR^d$. Therefore, for each semi-martingale $\{X_s'(t,x), s\in [t, T]\}$ of the form
  $$
  X_s'(t,x)=x+\int_t^s\varphi_r(t,x)\,dr+\int_t^s\sqrt{\nu}\,dW_r,\quad T_0<t\leq s\leq T
  $$
  with $\{\varphi_s(t,x), s\in [t,T]\}$ being  bounded and  predictable, it is interesting to understand $(Y_s$, $ Z_s)(s$, $X_s'(t,x))$  in FBSDE framework. Indeed, define the following equivalent probability measure:
  \begin{equation*}
  \begin{split}
  \frac{d\mathbb{Q}^{t,x}}{d\mathbb{P}} :=&\exp\biggl(\frac{1}{\sqrt{\nu}}\int_{t}^T \big[(b+\alpha  Y_s)(s,X_s'(t,x))-\varphi_s(t,x)\big]\,dW_s
  \\
  &
  \quad\quad\quad\quad
  -\frac{1}{2\nu}\int_{t}^T|(b+\alpha Y_s)(s,X_s'(t,x))-\varphi_s(t,x)|^2\,ds\biggr)\, .
  \end{split}
\end{equation*}
  Then in view of Girsanov theorem,  there is a standard Brownian motion $(W', \mathbb{Q}^{t,x})$ such that
  \begin{align*}
    &X_s'(t,x)=x+\int_t^r (b+\alpha Y_r)(r,X_r'(t,x))\,dr+\int_t^s\sqrt{\nu}\,dW_r',\ T_0<t\leq s\leq T.
  \end{align*}
   Then, we obtain
  \begin{align*}
    Y_{\tau}(\tau,X_{\tau}'(t,x))=
    &\ \psi(X_T'(t,x))+\!\int_{\tau}^T\!\!\!\left(\phi+c\,Y_s\right)(s,X_s'(t,x))\,ds
    -\sqrt{\nu}\int_{\tau}^T\!\!\!Z_s(s,X_s'(t,x))\,dW'_s\\
    =&\
    \psi(X_T'(t,x))+\int_{\tau}^T\!\!\!\!
    \Big\{\big[\phi+c\,Y_s+Z_s(b+\alpha Y_s)\big](s,X_s'(t,x))-Z_s\varphi_s(t,x)\Big\}\,ds
  \\
    &-\sqrt{\nu}\int_{\tau}^T Z_s(s,X_s'(t,x))\,dW_s.
  \end{align*}
\end{rmk}

\begin{proof}[ of Proposition \ref{prop FBSDEandPDE}]

\textbf{Step 1. Existence of the solution. } Choose $(b_n,c_n,\phi_n,\psi_n)\in C_c^{\infty}(\bR^{1+d};\bR^d)\times C_c^{\infty}(\bR^{1+d};\bR^{d\times d})\times C_c^{\infty}(\bR^{1+d};\bR^d)\times C_c^{\infty}(\bR^d;\bR^d)$ such that, putting $(\delta b_n,\delta c_n,\delta \phi_n,\delta\psi)=(b-b_n,c-c_n,\phi-\phi_n,\psi-\psi_n)$, we have
\begin{equation*}
\begin{split}
  &\lim_{n\rightarrow \infty}\left[\|\delta b_n\|_{C([0,T];H^{m})}
  +\|\delta b_n\|_{L^2(0,T;H^{m+1})}
  +\|\delta c_n\|_{L^2(0,T;H^{m})}
  +\|\delta\phi_n\|_{L^2(0,T;H^{m-1})}
  +\|\delta \psi_n\|_{m}\right]=0,\\
  &\|b_n\|_{C([0,T];H^{m})}\leq C \|b\|_{C([0,T];H^{m})},\, \|\phi_n\|_{L^2(0,T;H^{m-1})}\leq C \|\phi\|_{L^2(0,T;H^{m-1})},\,\,\,
  \|\psi_n\|_{m}\leq C \|\psi\|_{m},
\end{split}
\end{equation*}
$\|c_n\|_{L^2(0,T;H^{m})}\leq C \|c\|_{L^2(0,T;H^{m})}$,
and $\|b_n\|_{L^2(0,T;H^{m+1})}\leq C \|b\|_{L^2(0,T;H^{m+1})}$, where $C$ is a universal constant  being independent of $n$.
By the existing FBSDE theory (for instance, see \cite{MaProtterYong1994}), for each $n$, FBSDE \eqref{FBSDS_simp} with $(b,c,\phi,\psi)$ being replaced by smooth triple $(b_n,c_n,\phi_n,\psi_n)$ admits a local  solution $(X^n,Y^n,Z^n)$ on some time interval $(\tau,T]$ such that
$({Y^n},{Z^n})$ satisfies \eqref{eq prop express}.
 Then we have by Lemma \ref{lem_verif},
\begin{align*}
    &\| {Y^n_s}(s)\|_m^2+\nu\int_s^T\| {Z^n_r}(r)\|_m^2\,dr
    \\
    =&\ \|\psi_n\|_m^2
   		 +\int_s^T\!\!\!  2\left(\langle {Z^n_r}(b_n+\alpha {Y^n_r})(r),\,
    	{Y^n_r}(r)\rangle_{m-1,m+1}
    	+\langle (\phi_n(r)+c_nY^n_r(r), \,
    	{Y^n_r}(r)\rangle_{m-1,m+1}\right)\,dr
    \\
    \leq&
        \,C\!\!\int_s^T\!\!\!
        \Big(\|(b_n+\alpha {Y^n_r})(r)\|_{m}\|{Z^n_r}(r)\|_{m-1}\|{Y^n_r}(r)\|_{m+1}
        +\|c_n(r)\|_m\|Y_r(r)\|^2_m
        +\|\phi_n(r)\|_{m-1} \|{Y^n_r}(r)\|_{m+1} \Big)dr
    \\
    	&\,+\|\psi_n\|_m^2\\
    \leq&\
        C\bigg\{
        \int_s^T\!\!\! \left(\|b_n(r)\|_m^2
        +\|c_n(r)\|_m+1\right)\|{Y^n_r}(r)\|_m^2\,dr
        +\alpha^2\int_s^T\!\! \| {Y^n_r}(r)\|_{m}^4\,dr
        +\int_s^T\!\!\! \|\phi(r)\|_{m-1}^2\,dr\bigg\}\\
        &
        +\frac{\nu}{2}\int_s^T\|{Z^n_r}(r)\|_m^2\,dr
        +C\|\psi\|_m^2.
\end{align*}
Gronwall inequality implies
\begin{align}
    &\| {Y^n_s}(s)\|_m^2+{\nu}\int_s^T\|{Z^n_r}(r)\|_m^2\,dr
    \nonumber\\
    \leq&\,
    C\exp\left\{\int_0^T\left(\|b_n(r)\|^2_m+\|c_n(r)\|_m+1\right)dr\right\}
    \left(\|\psi\|_m^2+\|\phi\|^2_{L^2(0,T;H^m)}
        +\alpha^2\int_s^T\!\! \| {Y^n_r}(r)\|_{m}^4\,dr\right)
    \nonumber\\
     \leq&\,
    C\exp\left\{\int_0^TC\left(\|b(r)\|^2_m+\|c(r)\|^2_m+1\right)dr\right\}
    \left(\|\psi\|_m^2+\|\phi\|^2_{L^2(0,T;H^m)}
        +\alpha^2\int_s^T\!\! \| {Y^n_r}(r)\|_{m}^4\,dr\right)
    \nonumber\\
    \leq&\,
    C_0\left(1
        +\alpha^2\int_s^T\!\! \| {Y^n_r}(r)\|_{m}^4\,dr\right)
    \label{est-alpha-pf-prop}
\end{align}
with the constant $C_0$ depending only on $\|\phi\|_{L^2(0,T;H^{m-1})}$, $\|b\|_{C([0,T];H^{m})}$, $\|c\|_{L^2(0,T;H^{m})}$, $\|\psi\|_m$, $\nu$, $m$, $d$ and $T$. Hence, applying Gronwall inequality again and setting
$$\tau_0=\left(T-\frac{1}{|C_0\alpha|^2}\right)\vee 0,$$
we have for any $s\in(\tau_0,T]$,
\begin{equation}\label{eq local est}
  \sup_{r\in[s,T]}\| {Y^n_r}(r)\|_m^2+{\nu}\int_s^T\|{Z^n_r}(r)\|_m^2\,dr\leq \frac{C_0}{1-|\alpha C_0|^2(T-s)}.
\end{equation}
As a consequence, we are allowed to choose a uniform existing time interval $(\tau,T]$ for all $(X^n,Y^n,Z^n)$, $n\in\mathbb{Z}^+$.

Put
$$
(X^{nk},Y^{nk},Z^{nk},b_{nk},c_{nk},\phi_{nk},\psi_{nk})
=
(X^n,Y^n,Z^n,b_n,c_n,\phi_n,\psi_n)-(X^k,Y^k,Z^k,b_k,c_k,\phi_k,\psi_k).
$$
Then for each  $\eps\in(0,T-\tau)$,  we have, for any $s\in(\tau+\eps,T]$
\begin{align*}
    &\| {Y^{nk}_s}(s)\|_m^2+\nu\int_s^T\|{Z^{nk}_r}(r)\|_m^2\,dr
    \\
    =&\ \|\psi_{nk}\|_m^2
    +\int_s^T\!\!\!  2\left(\langle ({Z^{n}_r}b_{nk}+{Z^{nk}_r}b_k
    +c_{nk}Y^n_r+c_kY^{nk}_r+\phi_{nk})(r), \,  {Y^{nk}_r}(r)\rangle_{m-1,m+1}
    \right)\,dr
    \\
    &
    +\int_s^T\!\!\!  2\alpha\langle {Z^{n}_r}{Y^{nk}_r}(r)+{Z^{nk}_r}{Y^{k}_r}(r),\,   {Y^{nk}_r}(r)\rangle_{m-1,m+1}\,dr\\
    \leq& \,  \|\psi_{nk}\|^2_m+\frac{\nu}{2}\int_s^T\!\!\!(\|{Z^{nk}_r}(r)\|_{m}^2+\|{Y^{nk}_r}(r)\|_m^2)\,dr\\
    &\
    +C\bigg\{
    \alpha^2\int_s^T\left( \|{Y^{nk}_r}(r)\|_{m}^2\|{Y^{n}_r}(r)\|_{m}^2+
    \|{Y^{nk}}(r)\|_{m}^2\|{Y^{k}_r}(r)\|_{m}^2\right)\,dr
    \\
    &\
    +\int_s^T\!\!\left(\|\phi_{nk}\|_{m-1}^2
    + \big(\|b_{nk}\|_{m}^2+\|c_{nk}\|_m^2)\|{Y^{n}_r}\|_{m}^2
    +\big(\|b_{k}\|_{m}^2+\|c_k\|_m\big)\|{Y^{nk}_r}\|_{m}^2\right)(r)\,dr
      \bigg\}
    \\
    \leq&\
    \|\psi_{nk}\|^2_m+\int_s^T\!\bigg[\frac{\nu}{2}\|{Z^{nk}_r}\|_{m}^2
    +C
    \left( \|b_{nk}\|_{m}^2+\|c_{nk}\|_{m}^2+\|\phi_{nk}\|_{m-1}^2
    +(1+\|c_k\|_m)\|{Y^{nk}_r}\|_{m}^2\right)\bigg](r)\,dr
    .
\end{align*}
Consequently, for a constant $C$ which is independent of $n$ and $k$, we have
\begin{equation}\label{eq_nk}
  \begin{split}
    &\sup_{s\in[T_0+\eps,T]}\|{Y^{nk}_r}(s)\|_m^2+\nu\int_{T_0+\eps}^T\|{Z^{nk}_r}(r)\|_m^2\,dr
    \\
    \leq
    &
    \,
C\bigg\{\|\psi_{nk}\|^2_m
    +
    \int_0^T\left( \|b_{nk}\|_{m}^2+\|c_{nk}\|_{m}^2+\|\phi_{nk}\|_{m-1}^2  \right)(r)\,dr
    \bigg\}
    \longrightarrow 0\textrm{ as }n,k\rightarrow\infty.
  \end{split}
\end{equation}
Then $\{({Y^{n}_r(r,x)},{Z^{n}_r(r,x)}),(t,x)\in[\tau+\eps,T]\times\bR^d\}_{ n\in\mathbb{Z}^{+}}$ is a Cauchy sequence in $C([\tau+\eps,T];H^m) \times L^2(\tau+\eps,T;H^{m})$ for any $\eps\in(0,T-\tau)$, having a limit denoted by $(\zeta(r,x), \nabla\zeta(r,x))$.

On the other hand,  FBSDE
\begin{equation}\label{FBSDS_simp2}
  \left\{\begin{array}{l}
  \begin{split}
  dX_s(t,x)&
            =\left[b(s,X_s(t,x))+\alpha \zeta(s,X_s(t,x))\right]\,ds+\sqrt{\nu}\,dW_s\,;\\
  X_t(t,x)&=x;\\
  -dY_s(t,x)&
            =(\phi+c\,\zeta)(s,X_s(t,x))\,ds-\sqrt{\nu}Z_s(t,x)\,dW_s,\quad s\in[t,T];\\
    Y_T(t,x)&=\psi(X_T(t,x)),\\
    \end{split}
  \end{array}\right.
\end{equation}
admits a unique solution $(X,Y,Z)$ on $(\tau,T]$, which by Lemma \ref{lem_verif} satisfies \eqref{PDE_Mild} and \eqref{eq expresion}. Setting $k\rightarrow \infty$ first and then $n\rightarrow \infty$
in \eqref{eq_nk}, we have  $\zeta(t,x)=Y_t(t,x)$ and $\nabla \zeta(t,x)=Z_t(t,x)$ for almost every $(t,x)\in (\tau,T]\times\bR^d$. Again from Lemma \ref{lem_verif}, we see that
the triple ($X_s(t,x)$, $\zeta(s,X_s(t,x))$, $\nabla\zeta(s,X_s(t,x))$) solves FBSDE \eqref{FBSDS_simp} and satisfies all the assertions of this proposition except the uniqueness and the relation to PDE \eqref{backward-burgers}, which are left to the next steps.
Moreover, through a bootstrap argument, we can extend the existing interval to a maximal one denoted by $(T_0,T]$ with $T_0$ depending on $\|\psi\|_{m}$, $\|b\|_{C([0,T];H^{m})}$, $\|c\|_{L^2(0,T;H^{m})}$, $\|\phi\|_{L^2(0,T;H^{m-1})}$,
 $\alpha$, $m$, $d$, $\nu$ and $T$. In particular, if $\alpha=0$, it follows from estimates \eqref{est-alpha-pf-prop} and \eqref{eq local est} that the existence time interval is $[0,T]$.

 \textbf{Step 2. Uniqueness. }
   First, for $m>d/2+1$, as $H^{m-1}\hookrightarrow C^{0,\delta}(\bR^d)$ for some $\delta\in(0,1)$, our BSDE is well-posed for each $x\in\bR^d$. Then It\^o's formula yields  that for $T_0<t\leq s\leq T$,
   \begin{align*}
   &\left|Y_s(t,x)\right|^2+E\bigg[\nu\int_s^T\!\!|Z_r(t,x)|^2dr   \Big|\sF_s\bigg]\\
   &=
   		\,E\bigg[ |\psi(X_T(t,x))|^2
   		+2\int_s^T\!\!\langle Y_r(t,x),\,\phi(r,X_r(t,x))+c(r,X_r(t,x))Y_r(t,x) \rangle\,dr
   		 \Big|\sF_s\bigg]\\
   &\leq\,
   		C(m,d)\bigg\{
   		\|\psi\|_m^2
   +\|\phi\|^2_{L^2(0,T;H^{m-1})}	+\int_s^T\!\!\left(1+\|c(r)\|_m\right)E\Big[\left|Y_r(t,x)\right|^2\big|\sF_s\Big]\,dr
   		\bigg\},
   \end{align*}
   which by Gronwall inequality implies that for each $(t,x)\in(T_0,T]\times\bR^d$, there holds almost surely
   \begin{align*}
   \left|Y_s(t,x)\right|^2\leq C(m,d)\exp \left\{ T+T^{1/2}\|c\|_{L^2(0,T;H^m)}\right\}
   \left(\|\psi\|_m^2+\|\phi\|^2_{L^2(0,T;H^{m-1})}\right), \quad \forall\,s\in[t,T].
   \end{align*}
   Thus, each solution turns out to be a strengthened one.  Hence, we need only prove the uniqueness of the strengthened solution for general $m>d/2$.

   Let $(X,Y,Z)$ be any strengthened solution of \eqref{FBSDS_simp} on $(T_0,T]$.
 For each $(t,x)\in(T_0,T]\times\bR^d$, define the following equivalent probability measure:
   \begin{align*}
   d\mathbb{Q}^{t,x}
   :=\textrm{exp} \bigg(&
   -\frac{1}{\sqrt{\nu}}\int_{t}^T \!\!\left[b(s,X_s(t,x))+\alpha Y_s(t,x)\right]\,dW_s\\
   &-\frac{1}{2\nu}\int_{t}^T\!\!|b(s,X_s(t,x))+\alpha Y_s(t,x)|^2\,ds \bigg)   d\mathbb{P},
   \end{align*}
   where we note that $Y_{\cdot}(t,x)\in L^2(t,T;L^{\infty}(\Omega;\bR^d))$.
Then  FBSDE \eqref{FBSDS_simp} reads
\begin{equation*}
  \left\{\begin{array}{l}
  \begin{split}
  dX_s(t,x)&
            =\sqrt{\nu}\,dW'_s,\quad s\in[t,T];\\
  X_t(t,x)&=x;\\
  -dY_s(t,x)&
            =\Big[\phi(s,X_s(t,x))+c(s,X_s(t,x))Y_s(t,x)+Z_s(t,x)\big(b(s,X_s(t,x))+\alpha Y_s(t,x)\big)\Big]\,ds\\
            &\quad-\sqrt{\nu}Z_s(t,x)\,dW'_s,\quad s\in[t,T];\\
    Y_T(t,x)&=\psi(X_T(t,x)),\\
    \end{split}
  \end{array}\right.
\end{equation*}
where $(W', \mathbb{Q}^{t,x})$ is a standard Brownian motion.

 Borrowing the notations from \textbf{Step 1},  define
 \begin{equation}\label{eq def tildeYZ}
 \widetilde{Y}^n_s(t,\cdot)={Y^n_s}(s,X_s(t,\cdot))\textrm{ and }\widetilde{Z}^n_s(t,\cdot)={Z^n_s}(s,X_s(t,\cdot)).
 \end{equation}
 For simplicity, we assume $\tau=T_0$. As $m>d/2$ and $H^m\hookrightarrow C^{0,\delta}(\bR^d)$ for some $\delta\in(0,1)$, there is a constant $N^{t}$ such that
 $$\sup_{n}\left(\sup_{s\in[t,T],x\in\bR^d} \big|Y^n_s(s,x)|+ \int_t^T\sup_{x\in\bR^d}\big| {Z^n_s}(s,x)\big|^2 \,ds \right)\leq N^t.$$
 By Remark \ref{rmk_Theta_semiM}, we have $\textrm{ for almost all }x\in\bR^d$,
 \begin{align*}
   \widetilde{Y}^n_s(t,x)
   =&\,\int_s^T\!\!\!
   \big[\widetilde{Z}^n_r(t,x)\big(b_n(r,X_r(t,x))+\alpha \widetilde{Y}^n_r(t,x)\big)+
   c(r,X_r(t,x))\widetilde{Y}^n_r(t,x)
   +\phi_n(r,X_r(t,x))\big]\,dr\\
   &\,+\psi_n(X_T(t,x))-\sqrt{\nu}\int_s^T\!\!\! \widetilde{Z}^n_r(t,x)\,dW'_r,\ t\leq s\leq T.
 \end{align*}

 Note that both $(\widetilde{Y}^n_{\cdot}(t,x),\widetilde{Z}^n_{\cdot}(t,x))$ and $((Y_{\cdot}(t,x),Z_{\cdot}(t,x)))$ belong to $ S^2([t,T];\bR^d)\times L^2_{\sF}(t,T;\bR^{d\times d})$.
 Put
 $
 (\delta Y^n,\delta Z^n)
 =
 (\widetilde{Y}^n-Y,\widetilde{Z}^n-Z)
 $.
 Then It\^o's formula yields
  \begin{align*}
   &E_{\mathbb{Q}^{t,x}}\left[
   |\delta Y^n_s(t,x)|^2+\nu\int_s^T|\delta Z^n_r(t,x)|^2\,dr
   \right]
   \\
   =&\
   2E_{\mathbb{Q}^{t,x}}
   \int_s^T\!\!\!\langle \delta Y^n_r(t,x),\,
   \widetilde{Z}^n_r(t,x)\big(b_n(r,X_r(t,x))+\alpha \widetilde{Y}^n_r(t,x)\big)
   +c_n(r,X_r(t,x))\widetilde{Y}^n_r(t,x)+\delta\phi_n(r,X_r(t,x))
   \\
   &\,\quad\quad\quad-Z_r(t,x)\big(b(r,X_r(t,x))+\alpha Y_r(t,x)\big)
   -c(r,X_r(t,x))Y_r(t,x)
      \rangle \,dr+E_{\mathbb{Q}^{t,x}}|\delta\psi_n(X_T(t,x))|^2
   \\
   =&\
   E_{\mathbb{Q}^{t,x}}|\delta\psi_n(X_T(t,x))|^2
   +
   2E_{\mathbb{Q}^{t,x}}
   \int_s^T\!\!\!\langle
   \delta Z^n_r(t,x)b(r,X_r(t,x))+\widetilde{Z}^n_r(t,x)\delta b^n(r,X_r(t,x))
   +\alpha \delta Z^n_rY_r(t,x)
    \\
   &\,+\alpha \widetilde{Z}^n_r \delta Y^n_r(t,x)
   +c(r,X_r(t,x))\delta Y^n_r(t,x)+\delta c_n(r,X_r(t,x))\widetilde{Y}^n_r(t,x)
    +\delta \phi_n(r,X_r(t,x)),\,\delta Y^n_r(t,x)\rangle \,dr
   \\
   \leq&\,C\|\delta\psi_n\|_m^2
        +
        \frac{\nu}{2}E_{\mathbb{Q}^{t,x}}
        \int_s^T\!\!\!\big|\delta Z^n_r(t,x)\big|^2\,dr
        +C E_{\mathbb{Q}^{t,x}}\int_s^T\!\!\!\Big(\big|\delta \phi_n(r,X_r(t,x))\big|^2
        +\big|\widetilde{Z}^n_r(t,x)\big|^2\|\delta b_n\|_{C([0,T];H^m)}^2\\
        &~+\big|\widetilde{Y}_r^n(t,x)\big|^2\|\delta c_n(r)\|_m^2+\big|\delta Y^n_r(t,x)\big|^2\big(1+\|b(r)\|_m^2+\|c(r)\|_m^2+\|{Z^n_r}(r)\|^2_{C(\bR^d)}+|Y_r(t,x)|^2\big)
        \Big)\,dr\\
   \leq&\,C\|\delta\psi_n\|_m^2
        +CN^t\left(\|\delta b_n\|_{C([0,T];H^m)}^2+\|\delta c_n\|^2_{L^2(0,T;H^m)}\right)
        +C E_{\mathbb{Q}^{t,x}}\int_s^T\!\!\!\Big(\big|\delta \phi_n(r,X_r(t,x))\big|^2
        \\
        &
        +\big|\delta Y^n_r(t,x)\big|^2\big(1+\|b(r)\|_m^2+\|c(r)\|_m^2+\|{Z^n_r}(r)\|^2_{C(\bR^d)}
        +\|Y_r(t,x)\|_{L^{\infty}(\Omega;\bR^d)}^2\big)
        \Big)\,dr.\\
        &\,+
        \frac{\nu}{2}E_{\mathbb{Q}^{t,x}}
        \int_s^T\big|\delta Z^n_r(t,x)\big|^2\,dr.
 \end{align*}
 By Gronwall inequality, we obtain
 \begin{align*}
     &\sup_{s\in[t,T]}
         E_{\mathbb{Q}^{t,x}}\left[
        |\delta Y^n_s(t,x)|^2+\nu\int_s^T|\!\!\!\delta Z^n_r(t,x)|^2\,dr
        \right]    \\
     \leq&\
     C\Big\{
     \|\delta\psi_n\|_{m}^2
     +\|\delta b_n\|_{C([0,T];H^{m})}^2
     +\|\delta c_n\|^2_{L^2(0,T;H^m)}
      +E_{\mathbb{Q}^{t,x}}\int_s^T\!\!\!
   |\delta \phi_n(r,X_r(t,x))|^2\,dr
 \Big\},
 \end{align*}
 where the constant $C$ depends only on $N^t$, $\|Y(t,x)\|_{L^2(t,T;L^{\infty}(\Omega;\bR^d))}$, $T$, $\|b\|_{C([0,T];H^m)}$, $\|c\|_{L^2(0,T;H^m)}$, $m$, $d$, $\nu$ and $\alpha$, and is independent of $n$. As
 $$
 \int_{\bR^d}
 E_{\mathbb{Q}^{t,x}}\int_s^T
   |\delta \phi_n(r,X_r(t,x))|^2\,drdx
   =\|\delta \phi_n\|^2_{L^2(t,T;L^2(\bR^d))}
   \longrightarrow 0,\quad \textrm{as }n\rightarrow \infty,
 $$
extracting a subsequence if necessary, we have
$$
\sup_{s\in[t,T]}
         E_{\mathbb{Q}^{t,x}}\left[
        |\delta Y^n_s(t,x)|^2+\nu\int_s^T|\delta Z^n_r(t,x)|^2\,dr
        \right]
        \longrightarrow 0,\textrm{ a.e. }x\in\bR^d,\quad \textrm{as }n\rightarrow\infty.
$$

 Thus, in view of \eqref{eq def tildeYZ}, we conclude that for each $t\in(T_0,T]$ and almost every $x\in\bR^d$, there holds almost surely
 $$
 Y_s(t,x)=\zeta(s,X_s(t,x))\textrm{ and }Z_s(t,x)=\nabla \zeta(s,X_s(t,x)),\quad t\leq s\leq T.
 $$
 Therefore, any  strengthened solution of FBSDE \eqref{FBSDS_simp} on $(T_0,T]$ must have the form described as above. Now, let $(X,Y,Z)$ and $(\bar{X},\bar{Y},\bar{Z})$ be any two strengthened solutions of FBSDE \eqref{FBSDS_simp} on $(T_0,T]$. By previous argument we have
 \begin{equation}
   \begin{split}
     &Y_s(t,x)=\zeta(s,X_s(t,x)),\ Z_s(t,x)=\nabla \zeta(s,X_s(t,x)),\\
     &\bar{Y}_s(t,x)=\zeta(s,X_s(t,x)),\ \bar{Z}_s(t,x)=\nabla \zeta(s,X_s(t,x)).
   \end{split}
 \end{equation}
 Hence, by Lemma \ref{lem_verif} we must have $(X,Y,Z)\equiv(\bar{X},\bar{Y},\bar{Z})$.

Finally, similar to the proof of Theorem \ref{thm local}, we verify that $Y_t(t,x)$ is the unique strong solution of  PDE \eqref{backward-burgers} on $(T_0,T]$.  The proof is complete.
 \end{proof}

\begin{rmk}\label{rmk after prop alpha=0}
From the proof of Proposition~\ref{prop FBSDEandPDE}, we have
  $$
  T_0\leq 0\vee \left[ T-\frac{1}{|\alpha C_0|^2}  \right].
  $$
  If $T|\alpha C_0|^2<1$  in \eqref{eq local est},
  the existence time interval of the strengthened solution can be taken as $[0,T]$.
\end{rmk}
\begin{rmk}
  If $\widetilde{Y}_0(\cdot,\cdot)$ of FBSDS~\eqref{1.1} lies in $L^2(T_0,T; H^{m-1})$, then by Proposition \ref{prop FBSDEandPDE}, $Y_t(t,x)$ of FBSDS \eqref{1.1} is deterministic and belongs to $L^2(T_0,T;H^{m+1})$. Therefore, by Lemma \ref{lem norm-equivalence} and Proposition \ref{prop FBSDEandPDE}, Definition \ref{def Hm solution} makes sense.
\end{rmk}

From Proposition \ref{prop FBSDEandPDE}, we have the following characterization of an $H^m$-solution to FBSDS~\eqref{1.1} for $m> d/2$, whose proof is omitted.

\begin{cor}\label{cor-prop-FBSDEs} Let $T_0<T$.
Under assumptions of Theorem \ref{thm local}, $(X,Y,Z,\widetilde{Y}_0)$ is  an $H^{m}$-solution of  FBSDS \eqref{1.1} on  $(T_0,T]$ if
  and only if
  $(X,Y,Z,\widetilde{Y}_0)$ is a solution to FBSDS~\eqref{1.1} on  $(T_0,T]$ with the function $Y_t(t,x)$ lying in $ C_{\textrm{loc}}((T_0,T];H^m)\cap L^2_{\textrm{loc}}(T_0,T;H^{m+1})$ and
    \begin{align*}
  Z_t(t,\cdot)=\nabla Y_t(t,\cdot),Y_s(t,\cdot)= Y_s(s,X_s(t,\cdot)) \textrm{ and } Z_{s}(t,\cdot)=Z_s(s,X_s(t,\cdot)), \,\,t\leq s\leq T.
  \end{align*}.
\end{cor}

Immediately from Proposition \ref{prop FBSDEandPDE}, we have the maximum principle for PDEs of Burgers' type.
\begin{cor}\label{MP-burgers}
  In Proposition \ref{prop FBSDEandPDE}, assuming further $\phi\in L^1(0,T;L^{\infty})$, we have
  \begin{align*}
  \sup_{x\in\bR^d} \left|Y_t(t,x)\right|
  \leq&\,\exp\left\{ C(m,d)\|c\|_{L^1(t,T;H^m)} \right\}\left(\sup_{x\in\bR^d} \left|\psi(x)\right| +
  \int_t^T \esssup_{x\in\bR^d}\left|\phi(s,x)\right|\,ds\right).
  \end{align*}
  In particular, if $c\equiv 0$,  there holds for any $t\in (T_0,T]$, $j=1,\cdots,d$,
  \begin{align*}
  \sup_{x\in\bR^d} Y_t^j(t,x)
  \leq&\,\sup_{x\in\bR^d} \psi^j(x) +
  \int_t^T \esssup_{x\in\bR^d}\phi^j(s,x)\,ds\\
  \textrm{and}\quad \sup_{x\in\bR^d} \left|Y_t^j(t,x)\right|
  \leq&\,\sup_{x\in\bR^d} \left|\psi^j(x)\right| +
  \int_t^T \esssup_{x\in\bR^d}\left|\phi^j(s,x)\right|\,ds.
  \end{align*}
  \end{cor}

\section{Proof of Theorem \ref{thm local} and global results}
\label{subsec:pf-thm}

\subsection{Proof of Theorem \ref{thm local}}
By Proposition \ref{prop FBSDEandPDE} and Corollary \ref{cor-prop-FBSDEs}, to prove the existence and uniqueness of the $H^m$-solution of FBSDS \eqref{1.1}, it is equivalent to find a unique $u\in C((T_0,T];H^m_{\sigma})\cap L^2(T_0,T;H^{m+1}_{\sigma})$ on some time interval $(T_0,T]$ such that
$$
Y_s(t,\cdot)= u(s,X_s(t,\cdot)) \textrm{ and } Z_{s}(t,\cdot)=\nabla u(s,X_s(t,\cdot)), \,\,T_0<t\leq s\leq T.
$$
 Therefore, in view of the energy equality \eqref{Eq prop Energy} and the probabilistic representation for the operator $\nabla (-\Delta)^{-1}\textrm{div }\textrm{div}$  in Lemma \ref{lem singular operator}, we can use similar techniques of the energy method for the Navier-Stokes equations (see \cite{BertozziMajda2002}) to prove the existence and uniqueness of the $H^m$-solution for FBSDS \eqref{1.1}. To give a self-contained proof, we provide the following two-step iteration scheme.

\begin{proof}[ of Theorem \ref{thm local}]
  First, for each $v\in C([0,T];H^m_{\sigma})\cap L^2(0,T;H^{m+1}_{\sigma})$ and $\zeta\in C([0,T];H^m)\cap L^2(0,T;H^{m+1})$, consider the following FBSDS:
  \begin{equation}\label{Eq FBSDE1 in thm}
  \left\{\begin{array}{l}
  \begin{split}
  dX_s(t,x)&
            =v(s,X_s(t,x))\,ds+\sqrt{\nu}\,dW_s,\quad s\in[t,T];\\
  X_t(t,x)&=x;\\
  -dY_s(t,x)&
            =\left[f(s,X_s(t,x))+\widetilde{Y}_0(s,X_s(t,x))\right]\,ds-\sqrt{\nu}Z_s(t,x)\,dW_s;\\
    Y_T(t,x)&=G(X_T(t,x));\\
  -d\widetilde{Y}_s(t,x)&=\frac{27}{2s^3}\sum_{i,j=1}^dv^i\zeta^j(t,x+B_s)\left(B^i_{\frac{2s}{3}}-B^i_{\frac{s}{3}}\right)
              \left(B_s^j-B_{\frac{2s}{3}}^j\right)B_{\frac{s}{3}} \,ds\\
              &\ \ -\,\widetilde{Z}_s(t,x)dB_s,\quad s\in(0,\infty);\\
    \widetilde{Y}_{\infty}(t,x)&=0.\ \
    \end{split}
  \end{array}\right.
\end{equation}
By BSDE theory (see \cite{Karoui_Peng_Quenez,ParPeng_90}) and Lemmas \ref{lem norm-equivalence} and \ref{lem singular operator}, FBSDS \eqref{Eq FBSDE1 in thm} has a unique solution $(X^{v,\zeta}$, $Y^{v,\zeta}$, $Z^{v,\zeta}$, $\widetilde{Y}^{v,\zeta}_0)$.  From Lemma \ref{lem singular operator} and Remark \ref{rmk2 after lem SIO}, we have
 $$Y_0(t,x)=-\textbf{P}^{\perp}\,\textrm{div}(v\otimes\zeta)(t,x)
 =-\textbf{P}^{\perp}((v\cdot\nabla)\zeta)(t,x),
 $$
where we have used the fact that $\textrm{div}(v)=0$.
 By Lemma \ref{lem_verif}, the function
 $\{Y^{v,\zeta}_t(t,x),(t,x)\in[0,T]\times\bR^d\}$ lies in $C([0,T];H^m)\cap L^2(0,T;H^{m+1})$ and ${Z_t^{v,\zeta}(t,x)}=\nabla {Y^{v,\zeta}_t(t,x)}$.

For any $\zeta_i\in C([0,T];H^m)\cap L^2(0,T;H^{m+1})$, $i=1,2$,
put
$$
(\delta {Y^{v,\zeta}},\,\delta {Z^{v,\zeta}},\, \delta  \zeta):=({Y^{v,\zeta_1}}-{Y^{v,\zeta_2}},\,
{Z^{v,\zeta_1}}-{Z^{v,\zeta_2}},\,\zeta_1-\zeta_2).
$$
By Lemma \ref{lem_verif} and Remark \ref{rmk2 after lem SIO}, we have
\begin{align}
    &\|\delta {Y^{v,\zeta}_s}(s)\|_m^2+\nu\int_s^T\|\delta {Z_r^{v,\zeta}}(r)\|_m^2\,dr
    \nonumber\\
    =&
        \int_s^T \!\!\!2\langle \delta {Z_r^{v,\zeta}}v(r),\,\delta  {Y_r^{v,\zeta}}(r)\rangle_{m-1,m+1}\, dr
     -\int_s^T\!\!\!2\langle \textbf{P}^{\perp}\left((v\cdot\nabla)\delta\zeta\right)(r),
        \, \delta {Y_r^{v,\zeta}}(r)\rangle_{m-1,m+1} \,dr
    \nonumber\\
    \leq &\
        \frac{\nu}{2}\int_s^T(\|\delta {Y_r^{v,\zeta}}(r)\|_{m}^2+\|\delta {Z_r^{v,\zeta}}(r)\|_m^2)\,dr
        \nonumber\\
        &\
        +C(\nu)\left(
        \int_s^T\!\!\!\left( \|v(r)\|_m^2\|\delta {Y_r^{v,\zeta}}(r)\|_m^2
        + \|v(r)\|_{m}^2\|\delta\zeta(r)\|_{m}^2\right)\,dr
        \right).\nonumber
\end{align}
Using Gronwall inequality, we obtain
\begin{equation}
  \begin{split}
    \sup_{s\in[t,T]}\|\delta {Y_s^{v,\zeta}}(s)\|_m^2
    +\int_t^T\|\delta {Z_r^{v,\zeta}}(r)\|_m^2\,dr
    \leq
     C (T-t)\|\delta \zeta\|^2_{C([t,T];H^{m})}
  \end{split}
\end{equation}
with the constant $C$ depending on $m,d,\nu,\|v\|_{C([0,T];H^m)}$ and $T$. Then, by the contraction mapping principle we can choose a small enough positive constant $\eps\leq T$ depending only on  $m,d,\nu,\|v\|_{C([0,T];H^m)}$ and $T$, such that there exists a unique function $\bar{\zeta}\in C([T-\eps,T];H^m)\cap L^2(T-\eps,T;H^{m+1})$ satisfying
$$
\left({Y_r^{v,\bar{\zeta}}(r,x)},\,{Z_r^{v,\bar{\zeta}}(r,x)}\right)
=\left(\bar{\zeta},\nabla\bar{\zeta}\right)(r,x),
\textrm{ in }C([T-\epsilon,T];H^m)\times L^2(T-\epsilon,T;H^{m}).$$
Then by Lemmas \ref{lem singular operator} and \ref{lem_verif}, we have for almost all $x\in\bR^d$,
  \begin{align*}
      {Z_t^{v,\bar{\zeta}}}(t,x)
    =\nabla{Y_t^{v,\bar{\zeta}}}&(t,x), \big( Y^{v,\bar{\zeta}}_r(t,x),Z^{v,\bar{\zeta}}_r(t,x)\big)
    =\big({Y_r^{v,\bar{\zeta}}},{Z_r^{v,\bar{\zeta}}}\big)(r,X_r(t,x)),a.s.,\\
  {Y_r^{v,\bar{\zeta}}}(r,X_r(t,x))
    =&\,G(X_T(t,x))-\sqrt{\nu}\int_r^T\!\!{Z_s^{v,\bar{\zeta}}}(s,X_s(t,x))\,dW_s
\\
  &+\int_r^T \!\!\!\left[ f(s,X_s(t,x))-\textbf{P}^{\perp}((v\cdot \nabla){Y_s^{v,\bar{\zeta}}})(s,X_s(t,x))  \right]\,ds,\,a.s..
  \end{align*}

  For each $(t,x)\in[T-\eps,T)\times\bR^d$, define the following equivalent probability measure:
   $$
   d\mathbb{Q}^{t,x}
   :=\exp{\left(   -\frac{1}{\sqrt{\nu}}\int_{t}^T v(s,X_s(t,x))\,dW_s
   -\frac{1}{2\nu}\int_{t}^T|v(s,X_s(t,x))|^2\,ds  \right)  }\, d\mathbb{P}.
   $$
  Then  FBSDS \eqref{Eq FBSDE1 in thm} reads
  \begin{equation}\label{Eq FBSDE1 in thm_2}
  \left\{\begin{array}{l}
  \begin{split}
  dX_s(t,x)&
            =\sqrt{\nu}\,dW'_s,\quad s\in[t,T];\\
  X_t(t,x)&=x;\\
  -dY_s(t,x)&
            =\left[f+Z_sv+\widetilde{Y}_0\right](s,X_s(t,x))\,ds
            -\sqrt{\nu}Z_s(t,x)\,dW'_s;\\
    Y_T(t,x)&=G(X_T(t,x));\\
  -d\widetilde{Y}_s(t,x)&=\frac{27}{2s^3}\sum_{i,j=1}^dv^i Y_t^j(t,x+B_s)\left(B^i_{\frac{2s}{3}}-B^i_{\frac{s}{3}}\right)
              \left(B_s^j-B_{\frac{2s}{3}}^j\right)B_{\frac{s}{3}} \,ds\\
              &\ \ -\widetilde{Z}_s(t,x)dB_s,\quad s\in(0,\infty);                          \\
    \widetilde{Y}_{\infty}(t,x)&=0,\ \
    \end{split}
  \end{array}\right.
\end{equation}
   where $(W', \mathbb{Q}^{t,x})$ is a standard Brownian motion. By taking the divergence operator on both sides of the BSDE in the finite time interval in the above FBSDS, we deduce that
   $$\{Y_r^{v,\bar{\zeta}}(r,x),\,(r,x)\in[t,T]\times\bR^d\}\quad \in \quad   C([t,T];H^m_{\sigma})\cap L^2(t,T;H^{m+1}_{\sigma}).$$

    On the other hand, by Lemma  \ref{lem_verif} and  Remarks \ref{rmk_m=2} and \ref{rmk2 after lem SIO}, we have
\begin{align*}
    &\|{Y_s^{v,\bar{\zeta}}}(s)\|_{m}^2+\nu\int_s^T\|{Z_r^{v,\bar{\zeta}}}(r)\|_m^2\,dr
    \\
    =&\
    \|G\|_m^2+
        \int_s^T\!\!\!2\left(\langle (v\cdot\nabla){Y_r^{v,\bar{\zeta}}}(r),\,
        {Y_r^{v,\bar{\zeta}}}(r)\rangle_{m-1,m+1}
        +\langle f(r),\,{Y_r^{v,\bar{\zeta}}}(r)\rangle_{m-1,m+1}\right)\,dr
    \\
    &-\int_s^T\!\!\!2\langle \textbf{P}^{\perp}\textrm{div}\,\left(v\otimes {Y_r^{v,\bar{\zeta}}}
    \right)(r),\,
        {Y_r^{v,\bar{\zeta}}}(r)\rangle_{m-1,m+1} \,dr
    \\
    &
    \textrm{(in view of Remark \ref{rmk2 after lem SIO})}
    \\
        =&\
    \|G\|_m^2+
        \int_s^T\!\!\!2\left(\langle (v\cdot\nabla){Y_r^{v,\bar{\zeta}}}(r),\,
        {Y_r^{v,\bar{\zeta}}}(r)\rangle_{m-1,m+1}
        +\langle f(r),\,{Y_r^{v,\bar{\zeta}}}(r)\rangle_{m-1,m+1}\right)\,dr
    \\
    \leq&\
    \|G\|_m^2+
    C\left(
    \int_s^T\|v(r)\|^2_m\|{Y_r^{v,\bar{\zeta}}}(r)\|_m^2\,dr+\frac{1}{\nu}\int_s^T\|f(r)\|_{m-1}^2\,dr
    \right)
    \\
    &\ +\frac{\nu}{2}\int_s^T(\|{Y_r^{v,\bar{\zeta}}}(r)\|_m^2
    +\|{Z_r^{v,\bar{\zeta}}}(r)\|_{m}^2)\,dr
\end{align*}
which together with Gronwall inequality implies
\begin{align}
    &\sup_{s\in[t,T]}\|{Y_s^{v,\bar{\zeta}}}(s)\|_{m}^2
    +\frac{\nu}{2}\int_t^T\|{Z_r^{v,\bar{\zeta}}}(r)\|_m^2\,dr
    \nonumber\\
    \leq &\ C\left(\|f\|^2_{L^2(0,T;H^{m-1})}+\|G\|_m^2\right)
    \exp{\left(C(\|v\|^2_{C([t,T];H^m)}+\nu)(T-t)\right)}.
  \label{eq est in thm}
\end{align}
Hence, through a bootstrap argument, we conclude that
there exists a unique function $\widetilde{\zeta}\in C([0,T];H_{\sigma}^m)$
 satisfying $({Y_r^{v,\widetilde{\zeta}}(r,x)},{Z_r^{v,\widetilde{\zeta}}(r,x)})
 =(\widetilde{\zeta},\nabla\widetilde{\zeta})(r,x)$ in $C([0,T];H_{\sigma}^m)\times L^2(0,T;H_{\sigma}^m)$, and again by Lemma \ref{lem_verif}, we conclude that $$(X^{v},Y^{v},Z^{v},\widetilde{Y}^{v}_0)
 :=(X^{v,\widetilde{\zeta}},Y^{v,\widetilde{\zeta}},Z^{v,\widetilde{\zeta}},\widetilde{Y}^{v,\widetilde{\zeta}}_0)$$
 is the unique $H^m$-solution of the following FBSDS:
  \begin{equation}\label{Eq FBSDE2 in thm}
  \left\{\begin{array}{l}
  \begin{split}
  dX_s(t,x)&
            =v(s,X_s(t,x))\,ds+\sqrt{\nu}\,dW_s,\quad 0\leq t\leq s\leq T;\\
  X_t(t,x)&=x;\\
  -dY_s(t,x)&
            =\left[f(s,X_s(t,x))+\widetilde{Y}_0(s,X_s(t,x))\right]\,ds-\sqrt{\nu}Z_s(t,x)\,dW_s;\\
    Y_T(t,x)&=G(X_T(t,x));\\
  -d\widetilde{Y}_s(t,x)&=\frac{27}{2s^3}\sum_{i,j=1}^dv^iY_t^j(t,x+B_s)\left(B^i_{\frac{2s}{3}}-B^i_{\frac{s}{3}}\right)
              \left(B_s^j-B_{\frac{2s}{3}}^j\right)B_{\frac{s}{3}} \,ds\\
              &\ \ -\widetilde{Z}_s(t,x)dB_s,\quad s\in(0,\infty);                          \\
    \widetilde{Y}_{\infty}(t,x)&=0.
    \end{split}
  \end{array}\right.
\end{equation}

Choose two positive real numbers $R$ and $\eps$ ($\eps<T$) whose values are to be determined later, and define
\begin{align*}
U_R^{\eps}:=\Big\{
u&\,\in \,C([T-\eps,T];H^m_{\sigma})\cap L^2(T-\eps,T;H^{m+1}_{\sigma}):
\\
&\|u\|^2_{C([T-\eps,T];H^m_{\sigma})}
+\frac{\nu}{2}\|\nabla u\|^2_{L^2(T-\eps,T;H^{m+1}_{\sigma})}   \leq R^2
\Big\}.
\end{align*}
For any $v\in U_R^{\eps}$, there holds the following estimate by \eqref{eq est in thm}:
\begin{align}
    &\sup_{s\in[T-\eps,T]}\|{Y^v_s}(s)\|_{m}^2
    +\frac{\nu}{2}\int_{T-\eps}^T\|{Z^v_s}(s)\|_m^2\,dr
    \leq
    \ C(m,d,\nu,\|f\|^2_{L^2(0,T;H^{m-1})},\|G\|_m^2,T)e^{CR^2\eps}.
    \label{eq est2 in thm}
\end{align}
Choosing $R$ to be big enough and $\eps$ to be small enough, we have
$$
\sup_{s\in[T-\eps,T]}\|{Y^v_s}(s)\|_{m}^2+\frac{\nu}{2}\int_{T-\eps}^T\|{Z^v_r}(r)\|_m^2\,dr\leq R^2.
$$
On the other hand, for any $v_1,v_2\in U_R^{\eps}$, setting
$$(\delta {Y^v},\delta {Z^v},\delta v):=({Y^{v_1}}-{Y^{v_2}},{Z^{v_1}}-{Z^{v_2}},v_1-v_2),$$ we have
\begin{align*}
    &\|\delta {Y^v_s}(s)\|_m^2+\nu\int_s^T\|\delta {Z^v_r}(r)\|_m^2\,dr
    \\
    =&
    \  2\int_s^T\!\!\!\left(\langle (\delta v\cdot\nabla) {Y_r^{v_1}}(r),\, \delta {Y^v_r}(r)    \rangle_{m-1,m+1}
    +\langle  (v_2\cdot\nabla)\delta {Y^v_r}    (r),\, \delta {Y^v_r}(r)    \rangle_{m-1,m+1}
    \right)\,dr
    \\
    \leq&
    \ C(\nu,m,d)\left(
    \int_s^T\|\delta v(r)\|_{m}^2\|{Y_r^{v_1}}(r)\|_m^2 \,dr
    +\int_s^T\|v_2(r)\|_m^2\|\delta {Y_r^v}(r)\|_m^2\,dr
    \right)
    \\
    &+\frac{\nu}{2}\int_s^T(\|\delta {Y_r^v}(r)\|_{m}^2+\|\delta {Z^v_r}(r)\|_m^2)\,dr,
\end{align*}
which together with the Gronwall-Bellman inequality, implies
\begin{equation*}
  \begin{split}
    \sup_{s\in[T-\eps,T]}\|\delta {Y^v_s}(s)\|^2_m+
    \frac{\nu}{2}\int_{T-\eps}^T\|\delta {Z^v_r}(r)\|^2_m\,dr
    \leq\
    C(\nu,m,d)R^2e^{R^2T}\eps\|\delta v\|_{C([T-\eps,T];H^m_{\sigma})}^2.
  \end{split}
\end{equation*}
Therefore, if we choose $\eps$ to be small enough, the solution map $\Psi: v(t,x)\mapsto {Y^v_t}(t,x)$ is a contraction mapping on the  complete metric space $U_R^{\eps}$ and then through a bootstrap argument, we obtain a unique function $\bar{u}\in C_{loc}((T_0,T];H^m_{\sigma})\cap L^2_{loc}(T_0,T;H^{m+1}_{\sigma})$ satisfying $({Y^{\bar{u}}_t(t,x)},{Z_t^{\bar{u}}(t,x)})
=(\bar{u}(t,x),\nabla \bar{u}(t,x))$ on $(T_0,T]\times \bR^d$ with $T_0$ depending on $\nu,m,d,T,\|G\|_m$ and $\|f\|_{L^2(0,T;H^{m-1})}$. By Proposition \ref{prop FBSDEandPDE}, Corollary \ref{cor-prop-FBSDEs} and the contraction mapping principle,
$$(X,Y,Z,\widetilde{Y}_0):=(X^{\bar{u}},Y^{\bar{u}},Z^{\bar{u}},\widetilde{Y}^{\bar{u}}_0)$$
is the unique  $H^m$-solution of FBSDS \eqref{1.1} and there holds \eqref{thm-relat-Y-Z} and \eqref{Eq thm bNS-FBSDS}.

 From Remarks \ref{rmk2 after lem SIO} and \ref{rmk after lem SIO}, we deduce that there exists some $p\in L^2_{\textrm{loc}}(T_0,T;H^m)$ such that $\widetilde{Y}_0=\nabla p$. For each $(t,x)\in (T_0,T]\times\bR^d$, define the following equivalent probability $\mathbb{Q}^{t,x}$:
    \begin{align}
   d\mathbb{Q}^{t,x}
   :=\exp\left(
   -\frac{1}{\sqrt{\nu}}\int_{t}^T Y_s(s,X_s(t,x))\,dW_s
   -\frac{1}{2\nu}\int_{t}^T|Y_s(s,X_s(t,x))|^2\,ds     \right) \, d\mathbb{P}.
   \end{align}\label{Girsanov-Trans-proof-thm-loc}
   Then we have
     \begin{equation*}
  \left\{\begin{array}{l}
  \begin{split}
  dX_s(t,x)&
            =\sqrt{\nu}\,dW'_s,\quad s\in[t,T];\\
  X_t(t,x)&=x;\\
  -dY_s(t,x)&
            =\left[(Y_s\cdot\nabla) Y_s+f+\nabla p\right](s,X_s(t,x))\,ds-\sqrt{\nu}\nabla Y_s(s,X_s(t,x))\,dW'_s;\\
    Y_T(t,x)&=G(X_T(t,x)),
    \end{split}
  \end{array}\right.
\end{equation*}
 where $(W', \mathbb{Q}^{t,x})$ is a standard Brownian motion.
 For any $\zeta\in C_c^{\infty}(\bR^{1+d};\bR^d)$, It\^o's formula yields that
 \begin{equation*}
   \zeta(s,X_s(t,x))
   =\zeta(T,X_T(t,x))-\!\!\int_{s}^T\!\!\!\!(\partial_r+\frac{\nu}{2}\Delta)\zeta(r,X_r(t,x))dr
   -\sqrt{\nu}\!\!\int_s^T\!\!\!\!\!\nabla \zeta(r,X_r(t,x))\,dW_r',
 \end{equation*}
 and thus,
 \begin{align}
   &E_{\mathbb{Q}^{t,x}}[\langle Y_t,\,\zeta\rangle (t,x)]+\nu E\bigg[\int_t^T\langle \nabla\zeta,\,\nabla Y_s\rangle (s,X_s(t,x))\,ds \bigg]\nonumber\\
 =&\,E_{\mathbb{Q}^{t,x}}\bigg[
 \int_t^T(\langle -\partial_s\zeta-\frac{\nu}{2}\Delta\zeta,\ Y_s\rangle +\langle \zeta,\ (Y_s\cdot\nabla)Y_s+\nabla p+f\rangle )(s,X_s(t,x))\,ds \nonumber\\
 &\,\,+\langle \zeta(T,X_T(t,x)),\,G(x_T(t,x))\rangle
 \bigg].\nonumber
 \end{align}
Integrating both sides in $x$, we obtain
\begin{align*}
  \langle\zeta(t),\,Y_t(t)\rangle_0
  =&\!\int_t^T\!\!\!\big[-\langle \partial_s\zeta(s),\,Y_s(s)\rangle_0
  +\langle \zeta(s),\,\frac{\nu}{2}\Delta Y_s(s)+(Y_s\cdot\nabla)Y_s(s)   +\nabla p(s)+f(s)\rangle_0\big]\,ds\\
  &+ \langle\zeta(T),\,G  \rangle_0
\end{align*}
Hence,  $(Y_t(t,x),p(t,x))$ is a strong solution to Navier-Stokes equation \eqref{backward NS} (see \cite{Temam84,rT95}). Due to the reversibility
 of the above procedure, the uniqueness of the $H^m$-solution of FBSDS \eqref{1.1} implies that of the strong solution for Navier-Stokes equation \eqref{backward NS} as well. The proof is complete.
\end{proof}
\begin{rmk}\label{rmk2-thm}
   In the above proof, Proposition \ref{prop FBSDEandPDE} plays a crucial role in the characterization of the solution of
   FBSDS \eqref{1.1} (see Corollary \ref{cor-prop-FBSDEs}). This characterization  together with the contraction  principle serves to guarantee the existence and uniqueness of the $H^m$-solution of FBSDS~\eqref{1.1}.
\end{rmk}

\subsection{Global results}\label{sec:global-reslt}

By Lemma \ref{lem singular operator},  Proposition \ref{prop FBSDEandPDE} and Corollary \ref{cor-prop-FBSDEs}, basing on the energy equality \eqref{Eq prop Energy} we can also obtain the global results in a similar way to the energy method for the Navier-Stokes (see \cite[Page 86--134]{BertozziMajda2002}).
First, for the two dimensional case, we can obtain the following global result in a similar way to \cite[Section 3.3]{BertozziMajda2002}. We omit the proof herein.

\begin{prop}\label{prop-2d}
  Let $d=2$, $m\geq 3$ and $G\in H^m_{\sigma}$. Then FBSDS~\eqref{1.1} with $f=0$ admits a unique $H^m$-solution
  $(X,Y,Z,\widetilde{Y}_0)$ on $[0,T]$.
\end{prop}

The other global result is for the case of small Reynolds numbers. Let us work on the $d$-dimensional torus $\mathbb{T}^d=\mathbb{R}^d/(L\times \mathbb{Z}^d)$ where $L>0$ is a fixed length scale. Denote by $(H^{n,q}_{\sigma}(\mathbb{T}^d;\mathbb{R}^d),\|\cdot\|_{n,q;\mathbb{T}^d})$ the $\bR^d$-valued Sobolev space on $\mathbb{T}^d$, each element of which is divergence free. For $q=2$, write $(H^{n}_{\sigma}(\mathbb{T}^d;\mathbb{R}^d),\|\cdot\|_{n;\mathbb{T}^d})$ for simplicity.
Let $f=0$ and $G\in H^{m}_{\sigma}(\mathbb{T}^d;\bR^d)$  of zero mean  for $m>d/2$.  In a similar way to Theorem \ref{thm local}, FBSDS \eqref{1.1} defined on torus admits a unique $H^m$-solution $(X,Y,Z,\widetilde{Y}_0)$ on some interval $(T_0,T]$. Moreover, we have
\begin{equation*}
  \begin{split}
    &\|Y_t(t)\|^2_{m;\mathbb{T}^d}+\nu\int_t^T\|Z_s(s)\|_{m;\mathbb{T}^d}^2 \,ds\\
    =&
    \ \|G\|_{m;\mathbb{T}^d}^2
    +2\int_t^T\langle Z_s(s)Y_s(s),\, Y_s(s)   \rangle_{m-1,m+1;\mathbb{T}^d}\,ds
    \\
    \leq&
    \ \|G\|_{m;\mathbb{T}^d}^2
    +C\int_t^T\|Y_s(s)\|_{m;\mathbb{T}^d}\|Z_s(s)
    \|_{m-1;\mathbb{T}^d}\|Y_s(s)\|_{m+1;\mathbb{T}^d} \,ds
    \\
     \leq&
    \ \|G\|_{m;\mathbb{T}^d}^2
    +C\int_t^T\|Y_s(s)\|_{m;\mathbb{T}^d}\|Z_s(s)
    \|_{m-1;\mathbb{T}^d}\Big(\|Y_s(s)\|_{m;\mathbb{T}^d}+ \|Z_s(s)
    \|_{m;\mathbb{T}^d}\Big)\,ds
    \\
    \leq&
    \ \|G\|_{m;\mathbb{T}^d}^2
    +\widetilde{C}L\int_t^T\|Y_s(s)\|_{m;\mathbb{T}^d}\|Z_s(s)\|^2_{m;\mathbb{T}^d} \,ds,
  \end{split}
\end{equation*}
where $Z_s(s,\cdot)=\nabla Y_s(s,\cdot)$ and we have used the Poincar$\acute{\textrm{e}}$ inequality
\begin{equation*}
  \begin{split}
    &\|Y_s(s,\cdot)\|_{m;\mathbb{T}^d}
    \leq C L\|\nabla Y_s(s,\cdot)\|_{m;\mathbb{T}^d},\quad  s\in[t,T],
  \end{split}
\end{equation*}
with $Y_s(s, \cdot)$ being mean zero and the constant $\widetilde{C}$ being independent of $L$.
Thus,
\begin{equation*}
  \begin{split}
    &\|Y_t(t)\|^2_{m;\mathbb{T}^d}
    +\int_t^T(\nu-\widetilde{C}L\|Y_s(s)\|_{m;\mathbb{T}^d})\|Z_s(s)\|_{m;\mathbb{T}^d}^2 \,ds
    \leq \|G\|_{m;\mathbb{T}^d}^2.
  \end{split}
\end{equation*}
If we take the Reynolds number
$
R:=L\nu^{-1}\|G\|_{m;\mathbb{T}^d}<\widetilde{C}^{-1},
$
then for this local solution $(X,Y,Z,\widetilde{Y}_0)$ we always have
$$
\|Y_t(t,\cdot)\|_{m;\mathbb{T}^d}\leq \|G\|_{m;\mathbb{T}^d},\ t\in (T_0, T].
$$
Using bootstrap arguments,  the local solution can be extended to be a global one. In summary, we have

\begin{prop}\label{prop-smal-reynld}
  Assume that $f=0$ and $G\in H^{m}_{\sigma}(\mathbb{T}^d;\bR^d)$ ($m>d/2$) is mean zero. Then FBSDS \eqref{1.1} admits a unique $H^m$-solution $(X,Y,Z,\widetilde{Y}_0)$ on some time interval $(T_0,T]$ with $Y_t(t,x)$ being spacial mean zero. Moreover,
  there exists a positive constant $R_0$ ($={\widetilde{C}}^{-1}$ as above)
  such that if the Reynolds number $R<R_0$, the  local $H^m$-solution can be extended to be a time global one and for this global $H^m$-solution we have
  $$
  \|Y_t(t)\|_{m;\mathbb{T}^d}\leq \|G\|_{m;\mathbb{T}^d},\ \textrm{for any }t\in[0,T].
  $$
\end{prop}

\begin{rmk}
  FBSDS~\eqref{1.1} is a complicated version of FBSDE~\eqref{FBSDE-burgers}, including an additional nonlinear and nonlocal term in the drift of the BSDE to keep the backward state living in the divergence-free subspace. While the additional term causes difficulty in formulating probabilistic representations, it helps us to obtain the global solutions in Propositions \ref{prop-2d} and \ref{prop-smal-reynld}.
\end{rmk}

\section{Approximation of the Navier-Stokes equation}
\label{sec:numr-apprx}
In view of FBSDS \eqref{1.1} and Theorem \ref{thm local},  the Navier-Stokes equation is approximated in this section by truncating the time interval of the BSDE associated with $\widetilde{Y}$.

\begin{lem}
  For  $k\in\mathbb{N}$, $\gamma,\alpha\in(0,1)$, $
  \gamma\leq\alpha$, there is a constant $C$  such that
  \begin{align*}
    &\|\phi\|_{C^{k,\gamma}}\leq C \|\phi\|_{C^{k,\alpha}},
    \quad \forall \phi\in C^{k,\alpha},\\
    &
    \|\phi\|_{C^{k,\alpha}}\leq C \left(\|\phi\|_{C^{k,\gamma}} +
    \|\nabla\phi\|_{C^{k}}\right),
    \quad \forall \phi\in  C^{k+1}.
  \end{align*}
\end{lem}
It is an immediate consequence of the interpolation inequalities of Gilbarg and Trudinger \cite[Lemma 6.32]{GilbargTrud1983}.

To approximate the Navier-Stokes equations, we truncate the time interval of the infinite-time-interval BSDE of FBSDS \eqref{1.1}.
\begin{lem}\label{lem-numr}
  For any $\phi,\psi\in C^{k,\alpha}$, $k\in\mathbb{Z}^+$, $\alpha\in(0,1)$, the following BSDE
     \begin{equation}\label{bsde-lem-numr}
  \left\{\begin{array}{l}
  \begin{split} -d\widetilde{Y}_s(x)&=\frac{27}{2s^3}\sum_{i,j=1}^d\phi^i\psi^j(x+B_s)\left(B_s^j-B_{\frac{2s}{3}}^j\right)
\left(B^i_{\frac{2s}{3}}-B^i_{\frac{s}{3}}\right)
B_{\frac{s}{3}} \,ds\\
              &\ \ -\widetilde{Z}_s(t,x)dB_s,\ s\in(0,\infty)                          \\
	      \widetilde{Y}_{\infty}(x)&=0\ \
    \end{split}
  \end{array}\right.
\end{equation}
is well-posed on the time interval $(0,\infty)$ and $\widetilde{Y}_0(x):=\lim_{\eps\downarrow 0} E\widetilde{Y}_{\eps}(x)$ exists for each $x\in\bR^d$. Moreover,
\begin{equation*}
	\widetilde{Y}_0=\nabla (-\Delta)^{-1}\textrm{div }\textrm{div}(\phi\otimes\psi)
\in C^{k-1,\frac{\alpha}{2}},
\end{equation*}
and
\begin{equation}
\label{lem-numr-estimate}
\|\widetilde{Y}_0\|_{C^{k-1,\frac{\alpha}{2}}}\leq \,C\|\phi\|_{C^{k,\alpha}}\|\psi\|_{C^{k,\alpha}},
\end{equation}
with the positive constant $C$ independent of $\phi$ and $\psi$.
\end{lem}
\begin{proof}[Sketched only]
For any $\eps\in(0,1)$ and $N\in(1,\infty)$, in a similar way to the proof of Lemma \ref{lem singular operator},
  \begin{align}
    &\bigg| E\int_{\eps}^{N}
    \frac{27}{2s^3}\phi^i\psi^j(x+B_s)\left(B_s^j-B_{\frac{2s}{3}}^j\right)
\left(B^i_{\frac{2s}{3}}-B^i_{\frac{s}{3}}\right)
B_{\frac{s}{3}} \,ds\bigg|
\nonumber\\
=&\,\bigg|
E\bigg(\int_{\eps}^1+\int_1^N\bigg)
    \frac{27}{2s^3}\phi^i\psi^j(x+B_s)\left(B_s^j-B_{\frac{2s}{3}}^j\right)
\left(B^i_{\frac{2s}{3}}-B^i_{\frac{s}{3}}\right)
B_{\frac{s}{3}} \,ds\bigg|
\nonumber\\
=&\,\bigg|
E\int_{\eps}^1\frac{9}{2s^2}\left[\nabla(\phi^i\psi^j)(x+B_s)
-\nabla(\phi^i\psi^j)(x)\right]
\left(B_s^j-B_{\frac{2s}{3}}^j\right)
\left(B^i_{\frac{2s}{3}}-B^i_{\frac{s}{3}}\right)\,ds
\nonumber\\
&
+
E\int_1^N
    \frac{27}{2s^3}\phi^i\psi^j(x+B_s)\left(B_s^j-B_{\frac{2s}{3}}^j\right)
\left(B^i_{\frac{2s}{3}}-B^i_{\frac{s}{3}}\right)
B_{\frac{s}{3}} \,ds
\bigg|\nonumber\\
\leq &\,
C \|\phi\otimes\psi\|_{C^{1,\alpha}}
\int_{\eps}^1\frac{1}{s^{1-\frac{\alpha}{2}}}\,ds
+C\|\phi\otimes\psi\|_{L^{\infty}}
\int_1^N\frac{1}{s^{\frac{3}{2}}}\,ds
\nonumber\\
\leq &\,
C \|\phi\otimes\psi\|_{C^{1,\alpha}}
\left(
2-\eps^{\frac{\alpha}{2}}-\frac{1}{\sqrt{N}}
\right),\quad i,j=1,\cdots,d.  \label{est-lem-numr-eps-N}
  \end{align}
   Letting $N\rightarrow \infty$ and $\eps\rightarrow 0$, we conclude that BSDE \eqref{bsde-lem-numr} is well-posed on the time interval $(0,\infty)$ and $\widetilde{Y}_0(x):=\lim_{\eps\downarrow 0} E \widetilde{Y}_{\eps}(x)$ exists for each $x\in\bR^d$.

  On the other hand, for each $x,y\in\bR^d$, $i,j=1,\cdots,d$,
  \begin{align}
    &\bigg|E\int_1^{N}\frac{27}{2s^3}\left(\phi^i\psi^j(x+B_s)-\phi^i\psi^j(y+B_s)\right)
    \left(B_s^j-B_{\frac{2s}{3}}^j\right)
\left(B^i_{\frac{2s}{3}}-B^i_{\frac{s}{3}}\right)
B_{\frac{s}{3}} \,ds\bigg|
\nonumber\\
&\leq\,C|x-y|^{\frac{\alpha}{2}}
 \|\phi\otimes\psi\|_{C^{1,{\alpha}}}^{\frac{1}{2}}
\|\phi\otimes\psi\|_{L^{\infty}}^{\frac{1}{2}}
 E\int_1^{N}\frac{1}{s^3}
 \left|\left(B_s^j-B_{\frac{2s}{3}}^j\right)
\left(B^i_{\frac{2s}{3}}-B^i_{\frac{s}{3}}\right)
B_{\frac{s}{3}}\right|\,ds
\nonumber\\
&\leq\,C|x-y|^{\frac{\alpha}{2}}
 \|\phi\otimes\psi\|_{C^{1,\alpha}}\left(1-\frac{1}{\sqrt{N}}  \right)
 \label{est-lem-numr-N}
  \end{align}
  and
  \begin{align}
    &\bigg| E\int_{\eps}^1\frac{9}{2s^2}\left[\nabla(\phi^i\psi^j)(x+B_s)
-\nabla(\phi^i\psi^j)(y+B_s)\right]
\left(B_s^j-B_{\frac{2s}{3}}^j\right)
\left(B^i_{\frac{2s}{3}}-B^i_{\frac{s}{3}}\right)\,ds\bigg|
\nonumber\\
&=\,\bigg|
E\int_{\eps}^1\frac{9}{2s^2}\nabla\left[
\phi^i(x+B_s)\left(\psi^j(x+B_s)-\psi^j(y+B_s)\right)
+\left(\phi^i(x+B_s)-\phi^i(y+B_s)\right)\psi^j(y+B_s)\right]
\nonumber\\
&\quad\quad\quad\quad
\left(B_s^j-B_{\frac{2s}{3}}^j\right)
\left(B^i_{\frac{2s}{3}}-B^i_{\frac{s}{3}}\right)\,ds\bigg|
\nonumber\\
&=\,\bigg|
E\int_{\eps}^1\frac{9}{2s^2}\nabla\Big[
\left(\phi^i(x+B_s)-\phi^i(x)\right)\left(\psi^j(x+B_s)-\psi^j(y+B_s)\right)
\nonumber\\
&\quad\quad\quad\quad
+\phi^i(x)\left( \psi^j(x+B_s)-\psi^j(x)-\psi^j(y+B_s)+\psi^j(y)  \right)
\nonumber\\
&\quad\quad\quad\quad
+\left(\phi^i(x+B_s)-\phi^i(y+B_s)\right)\left(\psi^j(y+B_s)-\psi^j(y)\right)
\nonumber\\
&\quad\quad\quad\quad
+\psi^j(y)\left( \phi^i(x+B_s)-\phi^i(x)-\phi^i(y+B_s)+\phi^i(y)  \right)
\Big]
\left(B_s^j-B_{\frac{2s}{3}}^j\right)
\left(B^i_{\frac{2s}{3}}-B^i_{\frac{s}{3}}\right)\,ds\bigg|
\nonumber\\
&\leq\,
C |x-y|^{\frac{\alpha}{2}}
\|\phi\|_{C^{1,\alpha}}\|\psi\|_{C^{1,\alpha}}
E\int_{\eps}^1\frac{1}{s^2}\left||B_s|^{\frac{\alpha}{2}}
\left(B_s^j-B_{\frac{2s}{3}}^j\right)
\left(B^i_{\frac{2s}{3}}-B^i_{\frac{s}{3}}\right)\right|\,ds
\nonumber\\
&\leq\,
C |x-y|^{\frac{\alpha}{2}}
\|\phi\|_{C^{1,\alpha}}\|\psi\|_{C^{1,\alpha}}
\left(
1-\eps^{\frac{\alpha}{4}}
\right),\label{est-lem-numr-eps}
  \end{align}
  where for $h=\phi^i$, $\nabla\phi^i$, $\psi^j$ or $\nabla\psi^j$, we note that
  \begin{align*}
    &\left| h(x+B_s)-h(x)-h(y+B_s)+h(y)\right|
    \\
    \leq
    &
     \left| h(x+B_s)-h(x)-h(y+B_s)+h(y)\right|^{\frac{1}{2}}
    \left(\left| h(x+B_s)-h(x)\right|^{\frac{1}{2}}+\left|h(y+B_s)+h(y)\right|^{\frac{1}{2}}
    \right)
    \\
    \leq
    &
    4\left\|  h  \right\|_{C^{1,\alpha}}
    \left| B_s\right|^{\frac{\alpha}{2}}|x-y|^{\frac{\alpha}{2}}.
  \end{align*}
  Hence, combining \eqref{est-lem-numr-eps-N}, \eqref{est-lem-numr-N} and \eqref{est-lem-numr-eps}, we obtain
  \begin{align*}
    \|\widetilde{Y}_0\|_{C^{0,\frac{\alpha}{2}}}\leq \,C\|\phi\|_{C^{1,\alpha}}\|\psi\|_{C^{1,\alpha}}.
  \end{align*}
  Taking $k-1$-th derivatives in the above arguments, we prove \eqref{lem-numr-estimate}.
\end{proof}

\begin{rmk}\label{rmk-lem--eps-NN}
  In view of \eqref{est-lem-numr-N} and \eqref{est-lem-numr-eps} of the above proof, we can deduce easily that for any $\eps\in(0,1)$ and $N\in(1,\infty)$,
  \begin{align*}
    \left\|E\left[\widetilde{Y}_{\eps}- \widetilde{Y}_{N}\right]\right\|_{C^{k-1,\frac{\alpha}{2}}}
    \leq\,C \|\phi\|_{C^{k,\alpha}}\|\psi\|_{C^{k,\alpha}}
    \left(
2-\eps^{\frac{\alpha}{4}}-\frac{1}{\sqrt{N}}
\right)
  \end{align*}
  with the constant $C$ independent of $\phi$, $\psi$, $\eps$ and $N$. Moreover, in a similar way to the above proof, we obtain
  \begin{align*}
    \|E\left[\widetilde{Y}_{\eps}- \widetilde{Y}_{N}-\widetilde{Y}_0\right]\|_{C^{k-1,\frac{\alpha}{2}}}\leq \,C\left( \eps^{\frac{\alpha}{4}}+\frac{1}{\sqrt{N}} \right)\|\phi\|_{C^{k,\alpha}}\|\psi\|_{C^{k,\alpha}},
  \end{align*}
  with the constant $C$ independent of $\phi$, $\psi$, $\eps$ and $N$.
\end{rmk}

Define the heat kernel
$$
\mathcal{H}^{\nu}(t,x):
=\frac{1}{(2\pi \nu t)^{\frac{d}{2}}} {\rm exp}\Big(-\frac{|x|^2}{2\nu t}\Big),
$$
and the convolution
\begin{align*}
  \mathcal{H}^{\nu}(t)\ast g(x)
  =\int_{\bR^d}\mathcal{H}^{\nu}(t,x-y)g(y)\,dy,\quad \forall\,g\in C(\bR^d).
\end{align*}

\begin{lem}\label{lem-HK}
There is a constant $C$ such that for any $\phi\in C^{k,\gamma}$ with $k\in \mathbb{N}$ and $\gamma\in(0,1)$,
\begin{align}
\|\mathcal{H}^{\nu}(t)\ast\phi\|_{C^{k+1,\gamma}}
\leq & \,C\left(1+\frac{1}{\sqrt{\nu t}}\right)\|\phi\|_{C^{k,\gamma}};
\label{est-lem-app-Hk-1}
\\
\sum_{i,j=1}^d\|\partial_{x^i}\partial_{x^j} \mathcal{H}^{\nu}(t)\ast\phi\|_{C^{k}}
\leq & \,\frac{C}{(\nu t)^{1-\frac{\gamma}{4}}} \|\phi\|_{C^{k,\frac{\gamma}{2}}};
\label{est-lem-app-Hk-2}
\\
\|\mathcal{H}^{\nu}(t)\ast\phi\|_{C^{k+1,\gamma}}
\leq & \,C\left(1+\frac{1}{\sqrt{\nu t}}+\frac{1}{(\nu t)^{1-\frac{\gamma}{4}}}\right)\|\phi\|_{C^{k,\frac{\gamma}{2}}}.
\label{est-lem-app-Hk-3}
\end{align}

\end{lem}

\begin{proof}[Sketched only]
The estimate \eqref{est-lem-app-Hk-1} follows from
\begin{align*}
\left|\mathcal{H}^{\nu}(t)\ast \phi(x)\right|
  =\left|
  \int_{\bR^d}\mathcal{H}^{\nu}(t,x-y)\phi(y)\,dy\right|
   \,\leq\, \|\phi\|_{C(\bR^d)},\quad \forall\,x\in\bR^d
\end{align*}
and for any $x,z\in\bR^d$,
\begin{align*}
\left|\nabla\mathcal{H}^{\nu}(t)\ast \phi(x)-\nabla\mathcal{H}^{\nu}(t)\ast \phi(z)\right|
=&
\left|
  \int_{\bR^d} \frac{y}{(2\pi )^{\frac{d}{2}}(\nu t)^{\frac{d}{2}+1}} {\rm exp}\Big(-\frac{|y|^2}{2\nu t}\Big)\left(\phi(x-y)-\phi(z-y)\right)\,dy\right|
\\
\leq &\,|x-z|^{\gamma} \|\phi\|_{C^{0,\gamma}}
  \int_{\bR^d} \frac{|y|}{(2\pi)^{\frac{d}{2}}(\nu t)^{\frac{d}{2}+1}} {\rm exp}\Big(-\frac{|y|^2}{2\nu t}\Big)\,dy
  \\
\leq &\,\frac{C}{\sqrt{\nu t}}\,|x-z|^{\gamma} \|\phi\|_{C^{0,\gamma}}.
\end{align*}

For any $x\in\bR^d$ and $i,j=1,\cdots,d$,
\begin{align*}
\left|\partial_{x^i}\partial_{x^j} \mathcal{H}^{\nu}(t)\ast\phi(x)\right|
=&
\left|
\int_{\bR^d} \frac{1}{(2\pi \nu t)^{\frac{d}{2}}}\left(\frac{|y|^2}{(t\nu)^2}-\frac{1}{\nu t}\right) {\rm exp}\Big(-\frac{|y|^2}{2\nu t}\Big)\left(\phi(x-y)-\phi(x)\right)\,dy
\right|
\\
\leq&\,
\|\phi\|_{C^{0,\frac{\gamma}{2}}}
\int_{\bR^d} \frac{|y|^{\frac{\gamma}{2}}}{(2\pi \nu t)^{\frac{d}{2}}}\left(\frac{|y|^2}{(t\nu)^2}+\frac{1}{\nu t}\right) {\rm exp}\Big(-\frac{|y|^2}{2\nu t}\Big)\,dy
\\
\leq &
\frac{C}{(\nu t)^{1-\frac{\gamma}{4}}} \|\phi\|_{C^{0,\frac{\gamma}{2}}} ,
\end{align*}
which implies estimate \eqref{est-lem-app-Hk-2}.

Finally,
 \begin{align*}
   \|\mathcal{H}^{\nu}(t)\ast\phi\|_{C^{k+1,\gamma}}
   \leq\,&
   C\left( \|\mathcal{H}^{\nu}(t)\ast\phi\|_{C^{k+1,\frac{\gamma}{2}}}
   +\|\nabla \mathcal{H}^{\nu}(t)\ast\phi\|_{C^{k+1}}   \right)
   \,\,\,\textrm{(by Lemma \ref{lem-numr})}
   \\
   \leq\,&
   C\left( \|\mathcal{H}^{\nu}(t)\ast\phi\|_{C^{k+1,\frac{\gamma}{2}}}
   +\sum_{i,j=1}^d\|\partial_{x^i}\partial_{x^j} \mathcal{H}^{\nu}(t)\ast\phi\|_{C^{k}}   \right)
   \\
   \leq\,&
   C\left(1+
   \frac{1}{\sqrt{\nu t}}
   +\frac{1}{(\nu t)^{1-\frac{\gamma}{4}}}\right)\|\phi\|_{C^{k,\frac{\gamma}{2}}}
      \,\,\,\,   \textrm{(by estimates \eqref{est-lem-app-Hk-1} and \eqref{est-lem-app-Hk-2}).}
 \end{align*}
 The proof is completed.
\end{proof}

For each $N\in (1,\infty)$, define
\begin{align*}
  \textbf{P}_N(\phi\otimes\psi)=E\left[\widetilde{Y}_{\frac{1}{N}}
  -\widetilde{Y}_N\right],\quad \phi,\psi\in H^m,\,m> \frac{d}{2}+1,
\end{align*}
where $\widetilde{Y}_{\cdot}$ satisfies BSDE \eqref{bsde-lem-numr}. In view of Remark \ref{rmk-lem-SIO-eps}, we have
\begin{align}
  \left\|\textbf{P}_N(\phi\otimes\psi)\right\|_{k}
  \leq C\bigg(\frac{1}{\sqrt{N}}+ \sqrt{N}  \bigg)\|\phi\otimes\psi\|_k,\quad \,0\leq k\leq m.
\end{align}
In a similar way to Theorem \ref{thm local}, we have

\begin{thm}\label{thm-N}
Let $\nu>0, G\in H^m,$ and $f\in L^2(0,T;H^{m-1})$ with $m>d/2$. Then there is $T_0<T$ which depends on $\|f\|_{L^2(0,T;H^{m-1})}$, $\nu$, $m$, $d$, $T$, $N$ and $\|G\|_m$, such that FBSDS
\begin{equation}\label{FBSDS-N}
  \left\{\begin{array}{l}
  \begin{split}
  dX_s(t,x)&
            =Y_s(t,x)\,ds+\sqrt{\nu}\,dW_s,\quad s\in[t,T];\\
  X_t(t,x)&=x;\\
  -dY_s(t,x)&
            =\left[f(s,X_s(t,x))+\widetilde{Y}_0(s,X_s(t,x))\right]\,ds-\sqrt{\nu}Z_s(t,x)\,dW_s;\\
    Y_T(t,x)&=G(X_T(t,x));\\
  -d\widetilde{Y}_{s}(t,x)&=\!\!\sum_{i,j=1}^d \frac{27}{2s^3}Y^i_tY^j_t(t,x+B_{s})\left(B^i_{\frac{2s}{3}}-B^i_{\frac{s}{3}}\right)
              \left(B_{s}^j-B_{\frac{2s}{3}}^j\right)B_{\frac{s}{3}} \mathbb{I}_{\big[\frac{1}{N},N\big]}(s) \,ds\\
              &\ \ -\widetilde{Z}_s(t,x)dB_s,\quad s\in(0,\infty);\\
    \widetilde{Y}_{\infty}(t,x)&=0.
    \end{split}
  \end{array}\right.
\end{equation}
 has a unique $H^m$-solution $(X,Y,Z,\widetilde{Y}_0)$ on $(T_0,T]$ with
  $$ \{Y_t(t,x),\,(t,x)\in(T_0,T]\times\bR^d\}\quad\in\quad C_{loc}((T_0,T];H^m)\cap L^2_{loc}(T_0,T;H^{m+1}).$$
  Moreover, we have
  \begin{align}
  Z_t(t,\cdot)=\nabla Y_t(t,\cdot),Y_s(t,\cdot)= Y_s(s,X_s(t,\cdot)) \textrm{ and } Z_{s}(t,\cdot)=Z_s(s,X_s(t,\cdot)),
  \label{eq-thmN-relat-Y-Z}
  \end{align}
  for $T_0<t\leq s\leq T$, $(Y,Z,\widetilde{Y}_0)$ satisfies
  \begin{align}
    Y_r(r,X_r(t,x))=&
    \,G(X_T(t,x))
    +\int_r^T \left[f(s,X_s(t,x))+\widetilde{Y}_0(s,X_s(t,x))    \right]\,ds
    \nonumber\\
    &-\sqrt{\nu}\int_r^TZ_s(s,X_s(t,x))\,dW_s,\,T_0<t\leq r\leq T,\,a.e.x\in\bR^d,\,a.s.,
  \label{Eq thm bNS-pdeN}
  \end{align}
and ${u}^N(r,x):=Y_r(r,x)$ is the unique strong solution of the following PDE:
  \begin{equation}\label{Eq thm pdeN}
  \left\{\begin{array}{l}
          \partial_t  {u}^N + \frac{\nu}{2} \Delta  {u}^N + ( {u}^N\cdot \nabla) {u}^N + \textbf{P}_N( {u}^N\otimes {u}^N) +f=0 , \;\; T_0<t \leq T;\\
           {u}^N(T) = G.
    \end{array}\right.
    \end{equation}
  \end{thm}

  The proof is similar to that of Theorem \ref{thm local}, and is omitted here.

Letting $r=t$, using Girsanov transformation in a similar way to \eqref{Girsanov-Trans-proof-thm-loc} of the proof for Theorem \ref{thm local} and then taking expectations on both sides of
\eqref{Eq thm bNS-pdeN}, we have
\begin{align}
 {u}^N(t,x)=\mathcal{H}^{\nu}(T-t)\ast G(x)+\int_t^T\mathcal{H}^{\nu}(s-t)\ast \left(f+( {u}^N\cdot \nabla) {u}^N + \textbf{P}_N( {u}^N\otimes {u}^N)\right)(s,x)\,ds. \label{Repres-FBSDS-N}
\end{align}
Assume
\begin{align}\label{ass-hoelder}
G\in C^{k,\alpha},\,f\in C([0,T];C^{k-1,\alpha}),\quad k\in\mathbb{Z}^+,\alpha\in(0,1).
\end{align}
By Lemmas \ref{lem-numr} and \ref{lem-HK} and in view of Remark \ref{rmk-lem--eps-NN}, we get
\begin{align}
\| {u}^N(t)\|_{C^{k,\alpha}}
\leq& C \|G\|_{C^{k,\alpha}}
+C\int_{t}^T\bigg[
\Big(1+\frac{1}{\sqrt{s-t}}\Big)\big(\|f(s)+( {u}^N\cdot \nabla) {u}^N(s)\|_{C^{k-1,\alpha}}
\nonumber\\
&\quad\quad\quad
+\Big(1+\frac{1}{\sqrt{s-t}}+\frac{1}{(s-t)^{1-\frac{\alpha}{4}}}\Big)
\|\textbf{P}_N( {u}^N\otimes {u}^N)(s)\|_{C^{k-1,\frac{\alpha}{2}}}\bigg]ds
\nonumber
\\
\leq& C \|G\|_{C^{k,\alpha}}
+C\int_{t}^T\bigg[
\Big(1+\frac{1}{\sqrt{s-t}}\Big)\big(\|f(s)\|_{C^{k-1,\alpha}}+\| {u}^N(s)\|^2_{C^{k,\alpha}}\big)
\nonumber\\
&\quad\quad\quad
+\Big(1+\frac{1}{\sqrt{s-t}}+\frac{1}{(s-t)^{1-\frac{\alpha}{4}}}\Big)
\| {u}^N(s)\|^2_{C^{k,\alpha}}\bigg]ds
\nonumber\\
\leq&
C \|G\|_{C^{k,\alpha}}
+C\|f\|_{C([0,T];C^{k-1,\alpha})}
+C\int_{t}^T\!\!
\Big(1+\frac{1}{(s-t)^{1-\frac{\alpha}{4}}}\Big)
\| {u}^N(s)\|^2_{C^{k,\alpha}}\,ds,\label{deduction-apprx-N}
\end{align}
which by Gronwall inequality implies that
\begin{align}
\sup_{s\in[t,T]}\| {u}^N(s)\|_{C^{k,\alpha}}
\leq \frac{C\big( \|G\|_{C^{k,\alpha}}
+\|f\|_{C([0,T];C^{k-1,\alpha})}\big)}{1-C^2( \|G\|_{C^{k,\alpha}}
+\|f\|_{C([0,T];C^{k-1,\alpha})}) \big[ (T-t)+(T-t)^{\frac{\alpha}{4}}  \big]  },
\label{est-app-N-hoelder}
\end{align}
where $T-t$ is small enough and the constant $C$ is independent of $t$ and $N$. In view of \eqref{est-lem-app-Hk-2} of Lemma \ref{lem-HK}, we further have $ {u}^N(t)\in C^{k+1}$ when $t$ is away from $T$.

From estimate \eqref{est-app-N-hoelder} and Theorem \ref{thm-N}, similar to Theorem \ref{thm local}, we have the following corollary.

\begin{cor}\label{cor-app-N}
Let $\nu>0$. Under assumption \eqref{ass-hoelder}, there is $T_1<T$ which depends on $\|f\|_{C([0,T];C^{k-1,\alpha})}$, $\nu$, $T$ and $\|G\|_{C^{k,\alpha}}$, such that FBSDS~\eqref{FBSDS-N}
 has a unique solution $(X,Y,Z,\widetilde{Y}_0)$ on $(T_1,T]$ with $$\{Y_r(r,x),\,(r,x)\in(T_1,T]\times\bR^d\}\,\,\in\,\, C_{loc}((T_1,T];C^{k,\alpha})\cap C_{loc}((T_1,T);C^{k+1}) \textrm{ and }Z_r(r,x)=\nabla Y_r(r,x).$$
  Moreover, we have~\eqref{eq-thmN-relat-Y-Z},
  $(Y,Z,\widetilde{Y}_0)$ satisfies BSDE \eqref{Eq thm bNS-pdeN},
and  ${u}^N(r,x):=Y_r(r,x)$ is the unique solution of the PDE \eqref{Eq thm pdeN}.
\end{cor}

Since $C^{l,\alpha}\cap H^{m}$ is dense in $C^{l,\alpha}$ for any $m>\frac{d}{2}$ and $l\in\mathbb{N}$, by Theorem \ref{thm-N} we can prove the existence of the solution $(X,Y,Z,\widetilde{Y}_0)$ through standard density arguments.  In view of representation \eqref{Repres-FBSDS-N},
we can prove the uniqueness of the solution through a priori estimates in a similar way to \eqref{deduction-apprx-N}. From estimate \eqref{est-app-N-hoelder},
the unique  solution can be extended to the maximal time interval $(T_1,T]$. The proof of Corollary \ref{cor-app-N} is omitted. It is worth noting that  $T_1$ is independent of $N$ in Corollary~\ref{cor-app-N}, while in Theorem~\ref{thm-N} $T_0$ depends on $N$.

Now we shall use the solution $ {u}^N$ of PDE \eqref{Eq thm pdeN} to approximate the velocity field $ {u}$ of Navier-Stokes equation \eqref{backward NS}.

\begin{thm}\label{thm-apprx-N}
Let $\nu>0, G\in H^m_{\sigma},$ and $f\in C([0,T];H^{m-1}_{\sigma})$ with $m>\frac{d}{2}+1$. Let
$$ {u}\in C_{loc}((T_0,T];H^m_{\sigma})\cap L^2_{loc}(T_0,T;H^{m+1}_{\sigma})$$ be the strong solution of Navier-Stokes equation \eqref{backward NS} in Theorem \ref{thm local}. Since
$$H^m\hookrightarrow C^{k-1,\alpha},\quad\textrm{ for some }k\in\mathbb{Z}^+ \textrm{ and }\alpha\in(0,1),$$
 we are allowed to assume that $ {u}^N\in C_{loc}((T_1,T];C^{k,\alpha})\cap C_{loc}((T_1,T);C^{k+1})$ be the solution of PDE \eqref{Eq thm pdeN} in Corollary \ref{cor-app-N}. Then, for any $t\in (T_0\wedge T_1, T]$, there exists a constant $C$ independent of $N$ such that
\begin{align}
  \left\| {u}- {u}^N\right\|_{C(t,T;C^{k,\alpha})}
  \leq\frac{C}{N^{\frac{\alpha}{4}}}.\label{est-apprx-thm-apprx-N}
\end{align}
\end{thm}

\begin{proof}
In a similar way to \eqref{Repres-FBSDS-N}, we get for any $\tau\in[t,T]$
\begin{align*}
 {u}(\tau,x)=\mathcal{H}^{\nu}(T-\tau)\ast G(x)+\int_{\tau}^T\mathcal{H}^{\nu}(s-\tau)\ast \left(f+( {u}\cdot \nabla) {u} - \textbf{P}^{\perp}\,\textrm{div}\,( {u}\otimes {u})\right)(s,x)\,ds.
\end{align*}
  Putting $\delta u= {u}^N- {u}$, we have for $\tau\in[t,T]$
  \begin{align*}
  &\delta u(\tau,x)\\
  =&\!
  \int_{\tau}^T\!\!\!\mathcal{H}^{\nu}(s-\tau)\ast \left(( {u}^N\cdot \nabla) {u}^N-( {u}\cdot \nabla) {u} +
  \textbf{P}^N( {u}^N\otimes {u}^N)
  + \textbf{P}^{\perp}\,\textrm{div}\,( {u}\otimes {u})\right)(s,x)\,ds
  \\
  =&\!
  \int_{\tau}^T\!\!\!\!\mathcal{H}^{\nu}(s-\tau)\ast \Big((\delta {u}\cdot \nabla) {u}^N\!\!+( {u}\cdot \nabla)\delta {u} +
  (\textbf{P}^N\!\!+\textbf{P}^{\perp}\,\textrm{div}\,)( {u}\otimes {u})+
  \textbf{P}^N(\delta {u}\otimes {u}^N\!\!+ {u}\otimes\delta {u})
  \Big)(s,x)\,ds.
    \end{align*}
   From Lemmas \ref{lem-numr} and \ref{lem-HK} and Remark \ref{rmk-lem--eps-NN}, it follows that
   \begin{align*}
   &\|\delta u(\tau)\|_{C^{k,\alpha}}\\
   \leq & \,C
   \int_{\tau}^T  \!\!\Big(1+\frac{1}{(s-\tau)^{1-\frac{\alpha}{4}}}  \Big)\Big(
   \big\|(\delta {u}\cdot \nabla) {u}^N(s)+( {u}\cdot \nabla)\delta {u}(s)
   +
  (\textbf{P}^N\!\!+\textbf{P}^{\perp}\,\textrm{div}\,)( {u}\otimes {u})(s)
  \\
  &\quad\quad\quad+
  \textbf{P}^N(\delta {u}\otimes {u}^N\!\!+ {u}\otimes\delta {u})(s)
   \big\|_{C^{k-1,\frac{\alpha}{2}}}
   \Big)\,ds
   \\
   \leq &\, C
   \int_{\tau}^T \! \!\!\Big(1+\frac{1}{(s-\tau)^{1-\frac{\alpha}{4}}}  \Big)\Big(
   \left(
   \left\| {u}(s)\right\|_{C^{k,{\alpha}}} +\left\| {u}^N(s)\right\|_{C^{k,{\alpha}}}
  \right) \left\|\delta u(s)\right\|_{C^{k,{\alpha}}}
   +\frac{1}{N^{\frac{\alpha}{4}}}  \left\| {u}(s)\right\|^2_{C^{k,{\alpha}}}
   \!\!\Big)\,ds
   \\
   \leq &\,C
   \int_{\tau}^T \! \!\!\Big(1+\frac{1}{(s-\tau)^{1-\frac{\alpha}{4}}}  \Big)\Big(
    \left\|\delta u(s)\right\|_{C^{k,{\alpha}}}
   +\frac{1}{N^{\frac{\alpha}{4}}}
   \Big)\,ds,
   \end{align*}
   which  implies the estimate \eqref{est-apprx-thm-apprx-N} by Gronwall inequality. We complete the proof.
\end{proof}

\begin{rmk}\label{rmk-numrl}
  In view of Theorem \ref{thm-apprx-N}, we can approximate numerically the strong solution of Navier-Stokes equation \eqref{backward NS}, by approximating  the PDE \eqref{Eq thm pdeN}.  By Theorem \ref{thm-N} and Corollary \ref{cor-app-N}, we rewrite FBSDS \eqref{FBSDS-N} into the following form
  \begin{equation}\label{FBSDS-numrl}
  \left\{\begin{array}{l}
  \begin{split}
  dX_s(t,x)&
            =Y_s(s,X_s(t,x))\,ds+\sqrt{\nu}\,dW_s,\quad s\in[t,T];\\
  X_t(t,x)&=x;\\
  -dY_s(s,X_s(t,x))&
            =\left[f(s,X_s(t,x))+\textbf{P}^N(Y_s\otimes Y_s)(s,X_s(t,x))\right]\,ds-\sqrt{\nu}Z_s(t,x)\,dW_s;\\
    Y_T(T,x)&=G(x);\\
  \textbf{P}^N(Y_s\otimes Y_s)(s,x)&
  =\!\!\sum_{i,j=1}^d E\!\int_{\frac{1}{N}}^{N}\!\! \frac{27}{2r^3} Y_s^i Y_s^j(s,x+B_{r})
  \left(B^i_{\frac{2r}{3}}-B^i_{\frac{r}{3}}\right)
              \left(B_{r}^j-B_{\frac{2r}{3}}^j\right)B_{\frac{r}{3}} \,dr ;\\
  &=
  \!\!\sum_{i,j=1}^d E\!\int_{\frac{1}{3N}}^{\frac{N}{3}}\!\!
  \frac{3}{2r^3} Y_s^i Y_s^j(s,x+\bar{B}_{r}+\tilde{B}_{r}+\hat{B}_r)
  \bar{B}^i_{r} \tilde{B}^j_{r}\hat{B}_{r}
              \,dr,
    \end{split}
  \end{array}\right.
\end{equation}
where $\bar{B}$, $\tilde{B}$ and $\hat{B}$ are three independent d-dimensional Brownian motions. For the numerical approximation theory of FBSDEs, we refer to \cite{Bender-Zhang-08,Delarue-Menozzi-06,Delarue-Menozzi-08} and references therein. Indeed, in the spirit of Delarue and Menozzi \cite{Delarue-Menozzi-06,Delarue-Menozzi-08}, we can define roughly the following algorithm:
\begin{align*}
  &\forall \, x\in\bR^d, \quad \bar{u}^N(T,x)=G(x),
  \\
  &\forall \, k\in [0,\tilde{N}-1]\cap \mathbb{Z},\quad \forall\,x\in \Xi,
  \\
  &\quad \mathcal{J}(t_k,x)=\bar{u}^N(t_{k+1},x)h+\sqrt{\nu} \Delta W_{t_k},
  \\
  &\quad
  \mathcal{P}^N(t_{k},x)=
  \!\!\sum_{i,j=1}^d E\!\int_{\frac{1}{3N}}^{\frac{N}{3}}\!\!
  \frac{3}{2r^3}(\bar{u}^N)^i(\bar{u}^N)^j(t_{k+1},x+\bar{B}_{r}+\tilde{B}_{r}+\hat{B}_r)
  \bar{B}^i_{r} \tilde{B}^j_{r}\hat{B}_{r}
              \,dr,
  \\
  &\quad \bar{u}^N(t_k,x)
  =E\bar{u}^N(t_{k+1},x+\mathcal{J}(t_k,x))
  +h\left(
  f(t_k,x)+\mathcal{P}^N(t_{k},x)
  \right),
\end{align*}
where $h=\frac{T}{\tilde{N}}$, $t_k=kh$ ($k\in [0,\tilde{N}-1]\cap \mathbb{Z}$) and $\Xi=\delta \mathbb{Z}^d$ is the infinite Cartesian grid of step $\delta>0$. Compared with Delarue and Menozzi \cite{Delarue-Menozzi-06,Delarue-Menozzi-08}, we omit the projection mapping on the grid, quantized algorithm for the Brownian motions and the approximations for the diffusion coefficient of the BSDEs in \eqref{FBSDS-numrl}.  We can analyse the above algorithm in a similar way to Delarue and Menozzi \cite{Delarue-Menozzi-06,Delarue-Menozzi-08}, nevertheless, we shall not search such numerical applications in this work.   For more details on the forward-backward algorithms for quasi-linear PDEs and associated FBSDEs, we refer to Delarue and Menozzi \cite{Delarue-Menozzi-06,Delarue-Menozzi-08}, where the Burgers' equation and the deterministic KPZ equation are analyzed as numerical examples.
\end{rmk}

\section{Two related topics}
\subsection{Connections with the Lagrangian approach}

With the Lagrangian approach, Constantin and Iyer \cite{Constantin-Iyer-2008,Constantin-Iyer-arxiv-2011} derived a stochastic representation for the incompressible Navier-Stokes equations based on stochastic Lagrangian paths and gave a self-contained proof of the existence. Later, Zhang \cite{Zhang-bNS-2010} considered a backward analogue and provided short elegant proofs for the classical existence results. In this section, we shall derive from our representation (see Theorem \eqref{thm local}) an analogous Lagrangian formula, through which we show the connections with the Lagrangian approach.

Let $\nu>0, G\in H^m_{\sigma},$ and $f\in L^2(0,T;H^{m-1}_{\sigma})$ with $m>d/2$. By Theorem \ref{thm local}, the following FBSDS
\begin{equation}\label{FBSDS-Lag}
  \left\{\begin{array}{l}
  \begin{split}
  dX_s(t,x)&
            =Y_s(t,x)\,ds+\sqrt{\nu}\,dW_s,\quad s\in[t,T];\\
  X_t(t,x)&=x;\\
  -dY_s(t,x)&
            =\left[f(s,X_s(t,x))+\widetilde{Y}_0(s,X_s(t,x))\right]\,ds-\sqrt{\nu}Z_s(t,x)\,dW_s;\\
    Y_T(t,x)&=G(X_T(t,x));\\
  -d\widetilde{Y}_{s}(t,x)&=\!\!\sum_{i,j=1}^d \frac{27}{2s^3}Y^i_tY^j_t(t,x+B_{s})\left(B^i_{\frac{2s}{3}}-B^i_{\frac{s}{3}}\right)
              \left(B_{s}^j-B_{\frac{2s}{3}}^j\right)B_{\frac{s}{3}} \,ds\\
              &\ \ -\widetilde{Z}_s(t,x)dB_s,\quad s\in(0,\infty);\\
    \widetilde{Y}_{\infty}(t,x)&=0.
    \end{split}
  \end{array}\right.
\end{equation}
admits a unique $H^m$-solution $(X,Y,Z,\widetilde{Y}_0)$ on some time interval $(T_0,T]$, with
  \begin{align}
  Z_t(t,\cdot)=\nabla Y_t(t,\cdot),Y_s(t,\cdot)= Y_s(s,X_s(t,\cdot)) \textrm{ and } Z_{s}(t,\cdot):= Z_s(s,X_s(t,\cdot)),
  \label{thm-relat-Y-Z-Lag}
  \end{align}
  and there exists $p\in L^2(T_0,T;H^{m})$ such that $\nabla p := \widetilde{Y}_0$ and $(u,p)$ coincides  with the unique strong solution to Navier-Stokes equation:
  \begin{equation}\label{Eq thm bNS-Lag}
  \left\{\begin{array}{l}
          \partial_t {u} + \frac{\nu}{2} \Delta {u} + ({u}\cdot \nabla){u} + \nabla{p} +f=0 , \;\; T_0<t \leq T;\\
         \nabla \cdot {u} = 0, \quad {u}(T) = G.
    \end{array}\right.
    \end{equation}
    For each $t\in (T_0,T]$ a.e. $x\in\bR^d$, define the following equivalent probability $\mathbb{Q}^{t,x}$:
    \begin{align}
   d\mathbb{Q}^{t,x}
   :=\exp\left(
   -\frac{1}{\sqrt{\nu}}\int_{t}^T Y_s(s,X_s(t,x))\,dW_s
   -\frac{1}{2\nu}\int_{t}^T| Y_s(s,X_s(t,x))|^2\,ds     \right) \, d\mathbb{P}.
   \end{align}
   Then we have
     \begin{equation*}
  \left\{\begin{array}{l}
  \begin{split}
  dX_s(t,x)&
            =\sqrt{\nu}\,dW'_s,\quad s\in[t,T];\\
  X_t(t,x)&=x;\\
  -dY_s(t,x)&
            =\left[(Y_s\cdot\nabla) Y_s+f+\nabla p\right](s,X_s(t,x))\,ds-\sqrt{\nu}\nabla Y_s(s,X_s(t,x))\,dW'_s;\\
    Y_T(t,x)&=G(X_T(t,x)),
    \end{split}
  \end{array}\right.
\end{equation*}
where $(W', \mathbb{Q}^{t,x})$ is a standard Brownian motion.

    Consider the following BSDE
      \begin{equation}\label{BSDE-Lag}
  \left\{\begin{array}{l}
  \begin{split}
          -d\overline{Y}_s(t,x)=&\,[f(s,X_s(t,x))+Z^{\mathcal{T}}_s(t,x)\overline{Y}_s(t,x)]\,ds
          -\sqrt{\nu}\,\overline{Z}_s(t,x)dW_s, \;\; T_0<t\leq s \leq T;\\
        \overline{Y}_T(t,x) = &\,G(X_T(t,x)).
    \end{split}
    \end{array}\right.
    \end{equation}
    Putting
    \begin{align*}
      (\delta Y,\delta Z)_s(t,x)
      =( Y-\overline{Y}, Z-\overline{Z})_s(t,x),
    \end{align*}
    we have $\delta Y_T(t,x)=0$ and
    \begin{align*}
          -&d\delta {Y}_s(t,x)
          \\
          =&
          \big[\widetilde{Y}_0(s,X_s(t,x))
          -Z_s^{\mathcal{T}}\big(Y-\delta Y\big)_s(t,x)\big]\,ds
          -\sqrt{\nu}\,\delta {Z}_s(t,x)dW_s\\
          =&\big[(\nabla p-\nabla^{\mathcal{T}}Y_sY_s)(s,X_s(t,x))
          +\nabla^{\mathcal{T}}Y_s(s,X_s(t,x))\delta Y_s(t,x)\big]dt
         -\sqrt{\nu}\,\delta {Z}_s(t,x)dW_s
          \\
          =&\Big[\nabla\Big( p-\frac{1}{2}|Y_s|^2\Big)(s,X_s(t,x))
          +\nabla^{\mathcal{T}}Y_s(s,X_s(t,x))\delta Y_s(t,x)\Big]dt
         -\sqrt{\nu}\,\delta {Z}_s(t,x)dW_s\\
         =&\Big[\nabla\Big( p-\frac{1}{2}|Y_s|^2\Big)(s,X_s(t,x))
          +\nabla^{\mathcal{T}}Y_s(s,X_s(t,x))\delta Y_s(t,x)
          +\delta {Z}_s(t,x)Y_s(s,X_s(t,x)) \Big]dt\\
         &\,-\sqrt{\nu}\,\delta {Z}_s(t,x)dW'_s,
         \\
          =&\Big[\nabla\Big( p-\frac{1}{2}|Y_s|^2\Big)(s,X_s(t,x))
          +\nabla^{\mathcal{T}}Y_s{\delta Y_s}(s,X_s(t,x))
          +{\delta {Z}_s}Y_s(s,X_s(t,x)) \Big]dt\\
         &\,-\sqrt{\nu}\,\delta {Z}_s(t,x)dW'_s,
     \end{align*}
     where by Proposition \ref{prop FBSDEandPDE}, we have
     \begin{align}
       {\delta Z_t}(t,\cdot)=\nabla {\delta Y_t}(t,\cdot),\delta Y_s(t,\cdot)= {\delta Y_s}(s,X_s(t,\cdot)) \textrm{ and } \delta Z_{s}(t,\cdot):={\delta Z_s}(s,X_s(t,\cdot)).
     \end{align}
    On the other hand, through basic calculations it is easy to check that
    \begin{align}\label{Geom-gauge}
    v(r,x):=E\int_r^T\!\!\Big(  p-\frac{1}{2}|Y_s|^2 \Big)(s,X_s(r,x))\,ds,\quad \forall\,r\in(T_0,T],
    \end{align}
    satisfies BSDE
    \begin{equation}\label{BSDE-Lag-v}
  \left\{\begin{array}{l}
  \begin{split}
          -dv(s,X_s(t,x))=&\Big[p-\frac{1}{2}|Y_s|^2+
          \nabla v Y_s\Big](s,X_s(t,x))\,ds
          -\sqrt{\nu}\,\nabla v (s,X_s(t,x))dW'_s\\
          =&
          \Big[p-\frac{1}{2}|Y_s|^2\Big](s,X_s(t,x))\,ds
          -\sqrt{\nu}\,\nabla v (s,X_s(t,x))dW_s;\\
        v(T,x) = &0.
    \end{split}
    \end{array}\right.
    \end{equation}
    In view  of the following relation
    $$
    \nabla (\nabla v) Y_t(t,x)+\nabla^{\mathcal{T}} Y_t \nabla v(t,x)
    =\nabla\left(  (Y_t\cdot \nabla)v(t,x)  \right),
    $$
    we further check that $\nabla v(t,x)={\delta Y_t(t,x)}$. Therefore, we have
    \begin{prop}\label{prop-webber}
      Let $\nu>0, G\in H^m_{\sigma},$ and $f\in L^2(0,T;H^{m-1}_{\sigma})$ with $m>d/2$. Let $(X,Y,Z,\widetilde{Y}_0)$ be the unique  $H^m$-solution of FBSDS \eqref{FBSDS-Lag} on some time interval $(T_0,T]$ and $(\overline{Y},\overline{Z})$ satisfy BSDE \eqref{BSDE-Lag-vice}. Then, the strong solution of Navier-Stokes equation \eqref{Eq thm bNS-Lag} admits a probabilistic representation:
      \begin{align}
      u(t,x)=Y_t(t,x)= \overline{Y}_t(t,x)+\nabla v(t,x)= \textbf{P}\,\overline{Y}_t(t,x),\quad (t,x)\in(T_0,T]\times\bR^d,
      \label{Webber-formula-Lag}
    \end{align}
    with $\overline{Y}$ and $v$ satisfying BSDEs \eqref{BSDE-Lag} and \eqref{BSDE-Lag-v} respectively.
    \end{prop}

\begin{rmk}\label{rmk-connct-Lagrgn}
In view of the relation \eqref{thm-relat-Y-Z-Lag}, we rewrite \eqref{BSDE-Lag} into
\begin{equation}\label{BSDE-Lag-vice}
  \left\{\begin{array}{l}
  \begin{split}
          -d\overline{Y}_s(t,x)=&\,[f(s,X_s(t,x))+\nabla^{\mathcal{T}}u(s,X_s(t,x))\overline{Y}_s(t,x)]\,ds
          -\sqrt{\nu}\,\overline{Z}_s(t,x)dW_s;\\
        \overline{Y}_T(t,x) = &\,G(X_T(t,x)).
    \end{split}
    \end{array}\right.
    \end{equation}
    It follows that
        \begin{align*}
      \overline{Y}_t(t,x)
      =
      E\bigg[ \nabla^{\mathcal{T}}X_T(t,x) G (X_T(t,x))+ \int_t^T\!\!\nabla^{\mathcal{T}}X_s(t,x)f(s,X_s(t,x))\,ds
      \bigg],\quad T_0<t\leq T,
    \end{align*}
    and thus,
    \begin{align}
      u(t,x)=\textbf{P} E \bigg[ \nabla^{\mathcal{T}}X_T(t,x) G (X_T(t,x))+ \int_t^T\!\!\nabla^{\mathcal{T}}X_s(t,x)f(s,X_s(t,x))\,ds
      \bigg],\quad T_0<t\leq T,\label{webber-formula-2}
    \end{align}
    with
    \begin{equation}\label{SDE-Lag-SDE}
  dX_s(t,x)
            =u(s,X_s(t,x))\,ds+\sqrt{\nu}\,dW_s,\ \, s\in[t,T];
  \quad X_t(t,x)=x;
\end{equation}
and
      \begin{equation*}
          d\,\nabla^{\mathcal{T}}X_s(t,x)
          =
          \nabla^{\mathcal{T}}X_s(t,x)\nabla^{\mathcal{T}}u(s,X_s(t,x))\,ds
          , \;\; s\in[t,T];
        \quad\nabla^{\mathcal{T}} X_t(t,x) =
        I_{d\times d}.
    \end{equation*}
        Hence, by the relation \eqref{Webber-formula-Lag} or \eqref{webber-formula-2}, we derive a probabilistic representation along the stochastic particle systems for the strong solutions of Navier-Stokes equations. In fact, the representation formula \eqref{webber-formula-2} is analogous to those of Constantin and Iyer \cite{Constantin-Iyer-2008} and Zhang \cite{Zhang-bNS-2010} with the stochastic flow methods. Nevertheless, for the coefficients we only need $G\in H^m$ and especially, $f\in L^2(0,T;H^{m-1})$ which fails to be continuous when $m=2$, and these conditions are much weaker than those of \cite{Constantin-Iyer-2008,Zhang-bNS-2010} where $G$ and $f(t,\cdot)$ are spatially Lipschitz continuous and valued in $C^{k+1,\alpha}$ and $H^{k+2,q}(\bR^d)\,(\hookrightarrow C^{k+1,\alpha})$ respectively, with some $(k,\alpha,q)\in \mathbb{N}\times (0,1)\times (d,\infty)$.
\end{rmk}

\begin{rmk}
  The relation \eqref{Webber-formula-Lag} of Proposition \ref{prop-webber} gives the Helmholtz-Hodge decomposition of $u$.  ${\overline{Y}_t(t,x)}$ is called the magnetization variable (see \cite{chorin1994vorticity}) or impulse density borrowed from notions for Euler equation (see \cite{russo1999impulse}). For more different choices of the decompositions like \eqref{Webber-formula-Lag}, we refer to \cite{russo1999impulse}.
\end{rmk}

\subsection{A stochastic variational formulation for Navier-Stokes equations}

  Euler equation has an interesting variational interpretation (see Arnold \cite{Arnold-1966}, Ebin and Marsden \cite{ebin1970groups} and Bloch et al. \cite{bloch2000optimal}). The Navier-Stokes equation has been interpreted from various variational viewpoints by Inoue and Funaki \cite{inoue1979new} (on the group of \emph{random} orientation preserving diffeomorphisms for the weak solutions), Yasue \cite{yasue1983variational} (on a certain class of volume preserving diffusion processes on a compact manifold), and Gomes \cite{gomes2005variational} (on the divergence-free fields and the random fields which satisfy ordinary differential equations with random coefficients for the smooth solutions). In what follows, we shall interpret the strong solution to the Cauchy problem of the Navier-Stokes equation as a critical point to our controlled FBSDEs. Such a formulation is based on the  probabilistic representation of Proposition \ref{prop-webber} and seems to be new.

Assume that real number $\nu>0$, integer $m>d/2$, and the deterministic function $G\in \cap_{k>0}H^k_{\sigma}$ 
 of a compact support. Consider the following cost functional:
\begin{equation}\label{SOCP}
\begin{split}
  J(u,b,g)
  =&\frac{1}{2}\!
  \int_{\bR^d}\!\int_0^T\!\!\!|u(t,x)|^2\,dtdx
  +E\!
  \int_{\bR^d}\!\int_0^T\!\!\! \langle Y(t,x),\,b(t,x)-u(t,X(t,x)) \rangle  \,dtdx\\
  &
-E\int_{\bR^d}\langle G(X(T,x)),\,X(T,x)\rangle\,dx
\end{split}
\end{equation}
subject to

(i) $u\in L^2(0,T;H^{m+1}_{\sigma})\cap C([0,T];H^m_{\sigma}) $, $g\in L^2_{\sF}(0,T;H^{m})$ and  $ b \in L^2_{\sF}(0,T;H^{m+1})$;

(ii) $(X,Y,Z)\in L^2_{\sF}(0,T;L^2_{loc})\times \cS^2_{\sF}(0,T;L^2) \times L^2_{\sF}(0,T;L^2) $ satisfies the following FBSDE:
\begin{equation*}
  \left\{\begin{array}{l}
  \begin{split}
        X(t,x)=&\,x+\int_0^t\!\!\! b(s,x) \,ds+\sqrt{\nu}\,W_t;\\
        Y(t,x)=&\,G(X(T,x))+\int_t^T \!\!\!g(s,x)\,ds
          -\int_{t}^T\!\!\!Z(s,x)\,dW_s.
    \end{split}
    \end{array}\right.
    \end{equation*}
    Since $m>d/2$, by Sobolev embedding theorem we have $ H^{m+1} \hookrightarrow C^{1,\delta}$ and $ H^{m} \hookrightarrow C^{\delta}$ for some $\delta\in (0,1)$, and it is easy to check that all the terms involved above make senses. In addition, the cost  functional $J(u,b,p)$ herein is defined in a similar way to that in \cite[Theorem 2]{gomes2005variational} and \cite[Theorem 4]{bloch2000optimal} and in fact, we just add the terminal cost in form, but we shall discuss under a different framework.

\begin{prop}\label{prop-SOCP-NS}
  $(u,b,g)$ is a critical point of the above cost functional \eqref{SOCP}, if and only if $b(t,x)=u(t,X(t,x))$, $g(t,x)=\nabla^{\mathcal{T}}u(t,X(t,x))Y(t,x)$, and $u$ together with some pressure $p$ constitutes a strong solution to the Navier-Stokes equation \eqref{backward NS}.
\end{prop}

\begin{proof}
    For any $\delta u\in L^2(0,T;H^{m+1}_{\sigma})\cap C([0,T];H^m_{\sigma}) $, $\delta g \in L^2_{\sF}(0,T;H^{m})$ and  $ \delta b \in L^2_{\sF}(0,T;H^{m+1})$, let $(\delta X, \delta Y, \delta Z)\in L^2_{\sF}(0,T;L^2_{loc})\times \cS^2_{\sF}(0,T;L^2) \times L^2_{\sF}(0,T;L^2) $ solves the following FBSDE:
\begin{equation*}
  \left\{\begin{array}{l}
  \begin{split}
        \delta X(t,x)=&\int_0^t\!\!\! \delta b(t,x) \,ds;\\
        \delta Y(t,x)=&\int_t^T \!\!\!\delta g(s,x)\,ds
          -\int_{t}^T\!\!\!\delta Z(s,x)\,dW_s.
    \end{split}
    \end{array}\right.
    \end{equation*}
    Ito's formula yields that
    \begin{align}
    &  E\!
  \int_{\bR^d}\!\int_0^T\!\!\! \langle Y(t,x),\,b(t,x) \rangle  \,dtdx
-E\int_{\bR^d}\langle G(X(T,x)),\,X(T,x)\rangle\,dx\nonumber\\
=&
\int_{\bR^d}\!E\bigg(\int_0^T\!\!\! \Big(\langle X(t,x),\,g(t,x) \rangle -\sqrt{\nu}\,\textrm{tr}(Z(t,x))\Big) \,dt
-\langle Y(0,x),\,x\rangle\bigg)dx .    \label{eq-ito-formula}
    \end{align}
Note that in the above equality, the meaning of the right hand is implied by the left hand. Then through careful calculations we get the first variation
    \begin{align}
      \delta J:=&\frac{dJ(u+\eps\delta u,b+\eps\delta b,g+\eps\delta g )}{d\eps}\Big|_{\eps=0}\\
      =&
      \,E\! \int_{\bR^d}\!\int_0^T\!\!\!
      \left(\langle u,\,\delta u\rangle(t,x)-
      \langle Y,\delta u\rangle (t,X(t,x))\right)\,dtdx
      \nonumber\\
      &+E\! \int_{\bR^d}\!\int_0^T\!\!\!\langle \delta X(t,x),\,g(t,x)-\nabla^{\mathcal{T}}u(t,X(t,x))Y(t,x)\rangle\,dtdx \label{eq-vartn-deltaX}
      \\
      &+E\! \int_{\bR^d}\!\int_0^T\!\!\!\langle\delta Y(t,x),\,b(t,x)-u(t,X(t,x))\rangle\,dtdx
      \nonumber\\
      :=&\mathcal{R}_1+\mathcal {R}_2+\mathcal {R}_3,\nonumber
    \end{align}
    where we figure out \eqref{eq-vartn-deltaX} (or $\mathcal{R}_2$)  by inserting \eqref{eq-ito-formula} into the functional $J$. Then, by the arbitrariness of $\delta g$,
    \begin{align*}
      \mathcal{R}_3=\,E\! \int_{\bR^d}\!\int_0^T\!\!\!\langle\delta Y(t,x),\,b(t,x)-u(t,X(t,x))\rangle\,dtdx=0,
    \end{align*}
    which by standard denseness arguments implies that
    $$b(t,x)=u(t,X(t,x)).$$
    In a similar way, we obtain $\mathcal{R}_2=\mathcal{R}_1=0$ and
    $$g(t,x)=\nabla^{\mathcal{T}}u(t,X(t,x))Y(t,x).$$

    By the BSDE theory, there exists some $Z\in L^2_{\sF}(0,T;L^2)$ such that
    $(X,Y,Z)$ satisfies
    \begin{equation*}
  \left\{\begin{array}{l}
  \begin{split}
        X(t,x)=&\,x+\int_0^t\!\!\! u(s,X(s,x)) \,ds+\sqrt{\nu}\,W_t;\\
        Y(t,x)=&\,G(X(T,x))+\int_t^T \!\!\!\nabla^{\mathcal{T}}u(s,X(s,x))Y(s,x)\,ds
          -\int_{t}^T\!\!\!Z(s,x)\,dW_s.
    \end{split}
    \end{array}\right.
    \end{equation*}
    By Proposition \ref{prop FBSDEandPDE}, there exists some $\phi\in C([0,T];H^m)\cap L^2(0,T;H^{m+1})$ such that $Y(t,x)=\phi(t,X(t,x))$. In fact, it is easy to check that $\phi(t,x)=\overline{Y}_t(t,x)$ where $\overline{Y}$ satisfies BSDE \eqref{BSDE-Lag} (or \eqref{BSDE-Lag-vice}) with $f=0$ and $T_0=0$.

    Since $u$ is divergence free, by Lemma \ref{lem norm-equivalence} and Remark \ref{rmk_norm}, $X(t,\cdot)$ preserves the Lebesgue measure for all times. Noting that $\mathcal{R}_1=0$, we get
    \begin{align*}
      \! \int_{\bR^d}\!\int_0^T\!\!\!
   \langle u,\,\delta u\rangle(t,x)\,dtdx
      =&\,E \int_{\bR^d}\!\int_0^T\!\!\!
      \langle Y(t,x),\delta u(t,X(t,x))\rangle\,dtdx
      \\
      =&\,
      E \int_{\bR^d}\!\int_0^T\!\!\!
      \langle {\overline{Y}_t},\delta u\rangle (t,X(t,x))\,dtdx
      \\
      =& \int_{\bR^d}\!\int_0^T\!\!\!
      \langle {\overline{Y}_t},\delta u\rangle (t,x)\,dtdx.
    \end{align*}
Therefore,
$$
u(t,x)=\textbf{P}\,{\overline{Y}_t(t,x)},
$$
since $\delta u$ is divergence free and arbitrary.
Hence, by Proposition \ref{prop-webber}, $u$ together with some pressure $p$ constitutes a strong solution to the Navier-Stokes equation \eqref{backward NS}. The proof is complete.
\end{proof}

\section{Appendix}
\label{appendix}
\subsection{Proof of Lemma \ref{lem norm-equivalence}}
\label{append:1}

  It is sufficient for us to prove \eqref{eq_special} with $l=1$, from which \eqref{eq_time} follows by Fubini Theorem.

   First, taking a nonnegative function $\varphi\in C_c^{\infty}(\bR^d;\bR)$, we consider the following trivial FBSDE:
 \begin{equation}\label{FBSDE_norm1}
  \left\{\begin{array}{l}
  \begin{split}
  dX_r(t,x)&
            =b(r,X_r(t,x))\,dr+\sqrt{\nu}\,dW_r,\ \, t\leq r\leq s;\quad
  X_t(t,x)=x;\\
  dY_r(t,x)&
            =\sqrt{\nu}Z_r(t,x)\,dW_r,\ \, r\in[t,s];\quad
    Y_s(t,x)=\varphi(X_s(t,x)).
    \end{split}
  \end{array}\right.
\end{equation}
In view of Lemma \ref{lem_verif} and the proof therein, FBSDE \eqref{FBSDE_norm1} is a particular case with $\phi=0$ therein, and moreover, the assertions of Lemma \ref{lem_verif} still hold for \eqref{FBSDE_norm1}, as Lemma \ref{lem norm-equivalence} will never be involved in the proof of Lemma \ref{lem_verif} if $\phi=0$. Therefore, FBSDE \eqref{FBSDE_norm1} admits a unique solution such that for almost all $x\in\bR^d$,
$$\big(X_{\cdot}(t,x),Y_{\cdot}(t,x),Z_{\cdot}(t,x)\big)\in S^2(t,s;\bR^d)\times S^2(t,s;\bR^d) \times L_{\sF}^2(t,s;\bR^d),$$
and for this solution $(X,Y,Z)$, $\{Y_r(r,x),\,(r,x)\in(t,s)\times\bR^d\}\,\in\,L^{2}(t,s;H^{m+1})$,
  \begin{align*}
    Y_r(r,X_r(t,x))=&\,\varphi(X_s(t,x))
    -\sqrt{\nu}\int_r^sZ_{\tau}(\tau,X_{\tau}(t,x))\,dW_{\tau},\,a.s.
  \end{align*}
 and
  \begin{equation}\label{eq_norm3}
    Z_t(t,x)=\nabla Y_t(t,x),\, (Y_r(t,x),Z_r(t,x))=(Y_r,Z_r)(r,X_r(t,x)),\,a.s..
  \end{equation}
  In an obvious way, we have almost surely
  $$Y_r(t,x)=E\left[\varphi(X_s(t,x))\big|\sF_r\right]\geq 0,\quad \forall\,r\in[t,s].$$
   Define the following  equivalent probability measure
   $$
   d\mathbb{Q}^{t,x}
   =: \exp{\left(
   -\nu^{-{1\over2}}\int_{t}^s b(r,X_r(t,x))\,dW_s
   -\frac{1}{2}\nu^{-1}\int_{t}^T|b(r,X_r(t,x))|^2\,ds\right) }  d\mathbb{P}.
   $$
  In view of \eqref{eq_norm3},  FBSDE \eqref{FBSDE_norm1} reads
  \begin{equation}\label{FBSDE_norm2}
  \left\{\begin{array}{l}
  \begin{split}
  dX_r(t,x)&
            =\sqrt{\nu}\,dW'_r,\quad  t\leq r\leq s;\quad
  X_t(t,x)=x;\\
  -dY_r(t,x)&
            =Z_r(t,x)b(r,X_r(t,x))\,dr-\sqrt{\nu}Z_r(t,x)\,dW'_r\\
            &=(b\cdot\nabla)Y_r(r,X_r(t,x))\,dr-\sqrt{\nu}Z_r(t,x)\,dW'_r,\quad r\in[t,s];\\
    Y_s(t,x)&=\varphi(X_s(t,x)),\\
    \end{split}
  \end{array}\right.
\end{equation}
   where  $(W', \mathbb{Q}^{t,x})$ is a standard Brownian motion. Therefore,
   \begin{align*}
      \int_{\bR^d} Y_r(r,x)\,dx
     =&\,\int_{\bR^d}E_{\mathbb{Q}^{t,x}}\left[Y_r(r,X_r(t,x))\right]\,dx\\
     =&\, \int_{\bR^d}E_{\mathbb{Q}^{t,x}}\left[\varphi(X_s(t,x))\right]\,dx +
        \int_{\bR^d}\int_r^s E_{\mathbb{Q}^{t,x}}\left[ (b\cdot\nabla) Y_{\tau}(\tau,X_{\tau}(t,x))
        \right]\,d\tau dx\\
     =&\,\int_{\bR^d}\varphi(x)\,dx +
         \int_{\bR^d}\int_r^s(b\cdot\nabla)Y_{\tau}(\tau,x)\,d\tau dx\\
     =&\,\int_{\bR^d}\varphi(x)\,dx -\int_{\bR^d}\int_r^s (\textrm{div}\,b) Y_{\tau}(\tau,x)\,d\tau dx\\
     \leq&\,\int_{\bR^d}\varphi(x)\,dx
            +\int_r^s\|\textrm{div}\,b(\tau)\|_{L^{\infty}}\int_{\bR^d} Y_{\tau}(\tau,x)\,dxd\tau,\\
     &\textrm{or}\\
     \geq&\,\int_{\bR^d}\varphi(x)\,dx
            -\int_r^s\|\textrm{div}\,b(\tau)\|_{L^{\infty}}\int_{\bR^d} Y_{\tau}(\tau,x)\,dxd\tau.
   \end{align*}
   By Gronwall inequality, we have
   \begin{align*}
     \kappa\int_{\bR^d}\varphi(x)\,dx
     \leq\,
        \int_{\bR^d} Y_r(r,x)\,dx
     \leq\, \kappa^{-1} \int_{\bR^d}\varphi(x)\,dx, \quad \forall r\in[t,s]
   \end{align*}
   with
   $$\kappa:=e^{-\|\textrm{div}\,b\|_{L^1(t,s;L^{\infty})}}.$$
   Taking $r=t$, we get \eqref{eq_special}.

   For the general function $\varphi\in C_c^{\infty}(\bR^d;\bR)$ without the nonnegative assumption, we choose a positive Schwartz function $h$ and a nonnegative function $\widetilde{\varphi}\in C_c^{\infty}(\bR^d;\bR)$ such that
   $$
   \textrm{supp}\,\varphi\subset \,\left\{x\in \bR^d:\,\widetilde{\varphi}(x)=1\right\}.
   $$
    Set
    $$\varphi_{\eps}:=\left(\varphi^2+\eps h \right)^{\frac{1}{2}}\widetilde{\varphi},\,\,\textrm{for }\eps\in (0,1).$$
   Then in view of the above arguments, we have
   \begin{align*}
   \kappa \|\varphi\|_{L^1(\bR^d)}\leq\,
     \kappa \|\varphi_{\eps}\|_{L^1(\bR^d)}
                \leq   \int_{\bR^d} E\big[ |\varphi_{\eps}(X_s(t,x))| \big]\,dx
                        \leq \kappa^{-1}  \|\varphi_{\eps}\|_{L^1(\bR^d)}.
   \end{align*}
   Letting $\eps\rightarrow 0$, we conclude from Lebesgue dominant convergence theorem that \eqref{eq_special} holds for all $\varphi\in C_c^{\infty}(\bR^d;\bR)$.

   Finally, for any $\varphi\in L^1$, we choose a sequence $\{\varphi^n,n\in\mathbb{Z}^+\}\subset C_c^{\infty}(\bR^d;\bR)$ such that $\lim_{n\rightarrow \infty}\|\varphi-\varphi^n\|_{L^1}=0$. Then, by \eqref{eq_special}, $\{\varphi^n(X_s(t,x))\}$ is a Cauchy sequence in $L^1(\Omega\times\bR^d;\bR)$. It remains to show that  $\varphi(X_s(t,\cdot))$ is the limit.

   Through the above approximation, we can check that \eqref{eq_special} holds for any continuous function of a compact support. Therefore, if $A\subset \bR^d$ is a measurable, bounded subset of zero Lebesgue measure, then the $d\mathbb{P}\times dx$-measure of the set $\{(\omega,x)\in\Omega\times\bR^d: X_s(t,x)\in A\}$ is zero.
  Thus, the almost everywhere convergence of $\varphi^n$ to $\varphi$ in $\bR^d$ implies that of $\varphi^n(X_s(t,\cdot))$ to $\varphi(X_s(t,\cdot))$.

  Hence, $\varphi^n(X_s(t,x))$ converges to $\varphi(X_s(t,x))$ in $L^1(\Omega\times\bR^d;d\mathbb{P}\times dx)$.  Since \eqref{eq_special} holds for each $\varphi^n$, passing to the limit, relation \eqref{eq_special} holds for any $\varphi\in L^1$. We complete the proof.

\subsection{Proof of Lemma \ref{lem_verif}}

  For $m>d/2$,  $H^m\hookrightarrow C^{0,\delta},\, H^{m+1}\hookrightarrow C^{1,\delta}$ for some $\delta\in(0,1)$. By Theorems 3.4.1 and 4.5.1 of \cite{Kunita_book}, the \emph{forward} SDE is well posed for each $(t,x)\in [T_0,T]\times\bR^d$ and defines a stochastic flow of homeomorphisms. Moreover, by Lemma \ref{lem norm-equivalence} and Remark \ref{rmk_norm}, the backward SDE is also well posed for every $x\in\bR^d/F_t$ with Lebesgue's measure of $F_t$ being zero. Therefore,  FBSDE \eqref{FBSDE_verif1} has unique solution $(X,Y,Z)$ such that for each $(t,x)\in [T_0,T]\times(\bR^d/F_t)$,
  $$\big(X_{\cdot}(t,x),Y_{\cdot}(t,x),Z_{\cdot}(t,x)\big)\in S^2(T_0,T;\bR^d)\times S^2(T_0,T;\bR^d) \times L_{\sF}^2(T_0,T;\bR^d). $$

  For each $(t,x)\in[T_0,T)\times(\bR^d/F_t)$, define  the following equivalent probability measure:
   $$
  d\mathbb{Q}^{t,x}
  :=\exp{\left(
   -\frac{1}{\sqrt{\nu}}\int_{t}^T b(s,X_s(t,x))\,dW_s
   -\frac{1}{2\nu}\int_{t}^T|b(s,X_s(t,x))|^2\,ds \right)    } \, d\mathbb{P}.
   $$
  Then there is a  standard brownian motion $(W', \mathbb{Q}^{t,x})$ such that FBSDE \eqref{FBSDE_verif1} is written into the following form:
  \begin{equation}\label{FBSDE_verif11}
  \left\{\begin{array}{l}
  \begin{split}
  dX_s(t,x)&
            =\sqrt{\nu}\,dW'_s,\quad T_0\leq t\leq s\leq T;
  \quad X_t(t,x)=x;\\
  -dY_s(t,x)&
            =\Big(\phi(s,X_s(t,x))+Z_s(t,x)b(s,X_s(t,x))\Big)\,ds
            -\sqrt{\nu}Z_s(t,x)\,dW'_s;\\
    Y_T(t,x)&=\psi(X_T(t,x)).\\
    \end{split}
  \end{array}\right.
\end{equation}

 Choose $(b^n,\phi^n,\psi^n)\in C_c^{\infty}(\bR^{1+d};\bR^d)\times C_c^{\infty}(\bR^{1+d};\bR^d)\times C_c^{\infty}(\bR^d;\bR^d)$ such that
\begin{equation*}
\begin{split}
  &\lim_{n\rightarrow \infty}\Big\{\|b^n-b\|_{C([T_0,T];H^{m})}
  +\|b^n-b\|_{L^2(T_0,T;H^{m+1})}
  +\|\phi^n-\phi\|_{L^2(T_0,T;H^{m-1})}
  +\|\psi^n-\psi\|_{m}\Big\}=0,\\
  &\|b^n\|_{C([T_0,T];H^{m})}\leq C \|b\|_{C([T_0,T];H^{m})},\, \|\phi^n\|_{L^2(T_0,T;H^{m-1})}\leq C \|\phi\|_{L^2(T_0,T;H^{m-1})},\,\,\,
  \|\psi^n\|_{m}\leq C \|\psi\|_{m},
\end{split}
\end{equation*}
and $\|b^n\|_{L^2(T_0,T;H^{m+1})}\leq C \|b\|_{L^2(T_0,T;H^{m+1})}$, where $C$ is a universal constant  being independent of $n$. Let $(X,Y^n,Z^n)$ be the unique solution of FBSDE \eqref{FBSDE_verif11} with $(b,\phi,\psi)$ being replaced by $(b^n,\phi^n,\psi^n)$.
Then for each $n$, we have by the standard relationship between Markovian  BSDEs and PDEs  (for instance, see \cite{Antonelli93,HuPengFK95,MaProtterYong1994,PardouxTangFK99,PardouxZhang98,Peng1991_QPDE}),
$\{Y^n_r(r,x),\,(r,x)\in [t,T]\times\bR^d\}\,\in\,
 C([t,T];H^m)\cap L^{2}(t,T;H^{m+1})
  $,
and for each $x\in\bR^d$ and all $t\leq r\leq T,$
  \begin{align*}
  &{Y^n_r}(r,X_r(t,x))
  =\,\psi^n(X_T(t,x))+\!\int_r^T\!\!\!\left(\phi^n+{Z^n_s}b^n\right)(s,X_s(t,x))\,ds
  -\sqrt{\nu}\!
  \int_r^T\!\!\!{Z^n_s}(s,X_s(t,x))\,dW'_s,
  \\
    &{Z_t^n}(t,x)=\nabla {Y_t^n}(t        ,x), Y^n_r(t,x)={Y^n_r}(r,X_r(t,x)),
    Z^n_r(t,x)={Z^n_r}(r,X_r(t,x)).
  \end{align*}
  If we consider further
  \begin{equation*}
  \left\{\begin{array}{l}
  \begin{split}
  -d\bar{Y}^n(s,x)&=I_m(\phi^n+{Z^n_s}b^n)(s,X_s(t,x))\,ds-\sqrt{\nu}\bar{Z}^n_s(t,x)
  ,\  t\leq s\leq T;\\
    \bar{Y}_T(t,x)&=I_m\psi^n(X_T(t,x)),
    \end{split}
  \end{array}\right.
\end{equation*}
it is not difficult to show that for almost every $x\in\bR^d$ and all $s\in[t,T]$,
$$
\tilde{Y}^n_s(t,x)=I_m {Y^n_s}(s,X_s(t,x))\quad \text{and}\quad \bar{Z}^n_s(t,x)=I_m {Z^n_s}(s,X_s(t,x)),\quad a.s..
$$
Applying It\^o's formula, we have
\begin{equation*}
  \begin{split}
    &E_{\mathbb{Q}^{t,x}}\bigg[|\bar{Y}^n_{t}(t,x)|^2
    +\nu\int_{t}^{T}\!\!\!|\bar{Z}^n_s(t,x)|^2\,ds\bigg]
    \\
    =&
    2\!\int_{t}^{T}\!\!\!E_{\mathbb{Q}^{t,x}}\big[\langle \bar{Y}^n_s(t,x),
        \,I_m(\phi^n+{Z^n_s}b^n)(s,X_s(t,x))\}\rangle \big]\,ds +E_{\mathbb{Q}^{t,x}}\left[|I_m {Y^n_T}(T,X_{T}(t,x))|^2\right],\,a.e.x\in\bR^d.
  \end{split}
\end{equation*}
Thus, integrating  with respect to $x$ on both sides of the last equality, we obtain
\begin{equation}\label{eq:lemma6:iton}
  \begin{split}
    &\|{Y^n_t}(t)\|_m^2+\nu\!\int_{t}^{T}\!\!\!\|{Z_s^n}(s)\|_m^2\,ds
    =
    \ \|{Y_T^n}(T)\|_{m}^2
    +2\!\!\int_{t}^{T}\!\!\!\!
    \langle  \phi^n(s)+{Z_s^n}b^n(s),\,{Y_s^n}(s)
      \rangle_{m-1,m+1}\,ds.
  \end{split}
\end{equation}

Put
   $$
  (Y^{nk},Z^{nk},b_{nk},\phi_{nk},\psi_{nk})
  :=
  (Y^n-Y^k,Z^n-Z^k,b^n-b^k,\phi^n-\phi^k,\psi^n-\psi^k).
 $$
   For $n,k\in\mathbb{Z}^+$, we get by relation \eqref{eq:lemma6:iton} and Remark \ref{rmk_m=2}
  \begin{equation*}
  \begin{split}
    &\| {Y_s^{nk}}(s)\|_m^2+\nu\int_s^T\|{Z_r^{nk}}(r)\|_m^2\,dr
    \\
    =&\ \|\psi_{nk}\|_m^2
    +\int_s^T\!\!\!  2\langle \phi_{nk}(r)+  {Z_r^{n}}b_{nk}(r)+{Z_r^{nk}}b_k(r),\,   {Y_r^{nk}}(r)\rangle_{m-1,m+1}\,dr
    \\
    \leq& \
    \|\psi_{nk}\|^2_m
    +C\int_s^T\!\!\!\left(\|b_{nk}(r)\|_{m}^2\|{Y_r^{n}}(r)\|_{m}^2
    +\|\phi_{nk}(r)\|_{m-1}^2
    +\|b_{k}(r)\|_{m}^2\|{Y_r^{nk}}(r)\|_{m}^2
    \right)\,dr
    \\
    &    +\frac{\nu}{2}\int_s^T\!\!\!
    (\|{Z_r^{nk}}(r)\|_{m}^2+\|{Y_r^{nk}}(r)\|_m^2)\,dr
    \\
    \leq&\
    \|\psi_{nk}\|^2_m+\int_s^T\!\!
    \bigg[\frac{\nu}{2}\|{Z_r^{nk}}(r)\|_{m}^2+C\left( \|b_{nk}(r)\|_{m}^2+\|\phi_{nk}(r)\|_{m-1}^2+
    \|{Y_r^{nk}}(r)\|_{m}^2  \right)\bigg]\,dr,
  \end{split}
\end{equation*}
 where we have used the the following priori estimate by taking $(b^k,\phi^k,\psi^k)=0$ in the above,
 $$
\sup_{r\in[t,T]}\|{Y^n_r(r)}\|^2_m+\int_t^T \!\!\|{Z_r^n(r)}\|_m^2\,dr\leq C,
$$
with the constant $C$ being independent of $n$.
 Thus,
\begin{align}
    \sup_{s\in[t,T]}\|{Y_s^{nk}}(s)\|_m^2+\nu\int_{t}^T\|{Z_r^{nk}}(r)\|_m^2\,dr
    \leq&\
C\bigg(\|\psi_{nk}\|^2_m
    +
    \int_{t}^T\!\!\!\Big( \|b_{nk}(r)\|_{m}^2+\|\phi_{nk}(r)\|_{m-1}^2  \Big)dr
    \bigg)\nonumber\\
    &\ \longrightarrow 0\textrm{ as }n,k\rightarrow\infty. \nonumber
\end{align}
Hence, there exists $\xi\in C([t,T];H^m)\cap L^2(t,T;H^{m+1})$ such that
\begin{align}
   &\sup_{s\in[t,T]}\|{Y_s^{n}}(s)-\xi(s)\|_0^2
   +\nu\int_{t}^T\|{Z_r^{n}}(r)-\nabla\xi(r)\|_0^2\,dr
   \nonumber\\
    &=\sup_{s\in[t,T]} \int_{\bR^d}E_{\mathbb{Q}^{t,x}}|({Y_s^{n}}-\xi)(s,X_s(t,x))|^2dx
   +\nu\int_{t}^T\int_{\bR^d}E_{\mathbb{Q}^{t,x}} |({Z_r^{n}}-\nabla\xi)(r,X_r(t,x))|^2\,dxdr\nonumber\\
   &=\sup_{s\in[t,T]} \int_{\bR^d}E_{\mathbb{Q}^{t,x}}|{Y_s^{n}}(t,x)-\xi(s,X_s(t,x))|^2dx
   +\nu\int_{t}^T\int_{\bR^d}E_{\mathbb{Q}^{t,x}} |{Z_r^{n}}(t,x)-\nabla\xi(r,X_r(t,x))|^2\,dxdr\nonumber\\
   &\longrightarrow 0,\quad \textrm{as }n\rightarrow\infty.     \label{est_Ener}
\end{align}

  On the other hand, It\^o's formula yields that for each $(s,x)\in [t,T]\times \bR^d/F_t$
  \begin{align}
    &|Y^n_s(t,x)-Y_s(t,x)|^2+\nu\int_s^T \big|Z_r^n(t,x)-Z_r(t,x)\big|^2\,dr\nonumber\\
    =&\,   -\!\int_s^T\!\!\! 2\nu \langle Y^n_r(t,x)-Y_r(t,x),\, Z^n_r(t,x)-Z_r(t,x) \rangle\,dW'_r\nonumber\\
      &\,  +\!\int_s^T \!\!\!2\langle Y^n_r(t,x)-Y_r(t,x),\,\big(\phi^n-\phi\big)(r,X_r(t,x))+Z^n_r(t,x)b^n(r,X_r(t,x))
      \nonumber\\
      &\,\quad\ \,-Z_r(t,x)b(r,X_r(t,x)) \rangle\,dr+|\psi(X_T(t,x))-\psi^n(X_T(t,x))|^2.\label{eq_ito}
  \end{align}
  Using BDG and H\"{o}lder inequalities, we get
  \begin{align}
  &E_{\mathbb{Q}^{t,x}}\Big[\sup_{\tau\in[s,T]} |Y^n_{\tau}(t,x)-Y_{\tau}(t,x)|^2 + \nu\int_s^T \big|Z_r^n(t,x)-Z_r(t,x)\big|^2 \,dr   \Big]\nonumber\\
  \leq\,&
    E_{\mathbb{Q}^{t,x}}\Big[ |\psi(X_T(t,x))-\psi^n(X_T(t,x))|^2
    +\frac{1}{2}\sup_{\tau\in[s,T]} |Y^n_{\tau}(t,x)-Y_{\tau}(t,x)|^2
    \nonumber\\
    &
    +C \int_s^T |Y^n_r(t,x)-Y_r(t,x)|\Big(|\phi_n-\phi|(r,X_r(t,x))+\|b\|_{C([t,T];H^m)}|Z^n_r(t,x)-Z_r(t,x)|
    \nonumber\\
    &\quad\quad
    +\|b^n-b\|_{C([t,T];H^m)}|Z^n_r(t,x)|   \Big)\,dr
    +C\int_s^T \Big|Z_r^n(t,x)-Z_r(t,x)\Big|^2 \,dr
    \Big],\label{eq_itobdg}
  \end{align}
  with the constants $C$s being independent of $n$.  Combining \eqref{eq_ito} and \eqref{eq_itobdg}, we have
  \begin{align}
  &E_{\mathbb{Q}^{t,x}}\Big[\sup_{\tau\in[t,T]} |Y^n_{\tau}(t,x)-Y_{\tau}(t,x)|^2 + \nu\int_t^T \big|Z_r^n(t,x)-Z_r(t,x)\big|^2 \,dr   \Big]\nonumber\\
  \leq\,&
   C E_{\mathbb{Q}^{t,x}}\Big[\int_t^T \big(\big|(\phi_n-\phi)(r,X_r(t,x))\big|^2
   +\|b^n-b\|^2_{C([t,T];H^m)}|Z^n_r(t,x)|^2\big)\,dr\Big]
\nonumber\\
    &
    +CE_{\mathbb{Q}^{t,x}}\Big[ |\psi(X_T(t,x))-\psi^n(X_T(t,x))|^2,
    \nonumber
  \end{align}
  and thus,
  \begin{align}
  &\int_{\bR^d}\!E_{\mathbb{Q}^{t,x}}\sup_{\tau\in[t,T]} |Y^n_{\tau}(t,x)-Y_{\tau}(t,x)|^2dx
   + \nu\int_{\bR^d}\!\int_t^T\! E_{\mathbb{Q}^{t,x}}\big|Z_r^n(t,x)-Z_r(t,x)\big|^2 \,dr  dx \nonumber\\
  \leq\,&
  C\left(\|\psi-\psi^n\|_{m}^2+
    \|\phi^n-\phi\|^2_{L^2(T_0,T;H^m)}+ \|b^n-b\|^2_{C([t,T];H^m)}\int_{T_0}^T\|Z^n_r(r)\|_m^2dr\right)
    \nonumber\\
    &\longrightarrow 0, \quad \textrm{as }n\rightarrow\infty.
\label{est_Ito}
  \end{align}


Combining \eqref{est_Ito} and \eqref{est_Ener}, we have for almost every $x\in\bR^d$,  $(Y_t,Z_t)(t,x)=(\xi,\nabla\xi)(t,x)$, and $(Y,Z)_s(t,x)=(\xi,\nabla\xi)(s,X_s(t,x))$, a.s., for all $T_0\leq t\leq s\leq T$. Furthermore, by taking limits, we prove \eqref{PDE_Mild}, \eqref{eq expresion} and \eqref{eq:lemma6:energy}. The proof is complete.

\medskip

\bibliographystyle{siam}

\begin{thebibliography}{10}

\bibitem{albeverio2007generalized}
{\sc Albeverio, S., and Belopolskaya, Y.}
\newblock Generalized solutions of the cauchy problem for the navier-stokes
  system and diffusion processes.
\newblock {\em Cubo (Temuco) 12}, 2 (2010), 77--96.

\bibitem{Antonelli93}
{\sc Antonelli, F.}
\newblock Backward-forward stochastic differential equations.
\newblock {\em Ann. Appl. Probab. 3}, 3 (1993), 777--793.

\bibitem{Arnold-1966}
{\sc Arnold, V.}
\newblock Sur la g\'{e}om\'{e}trie diff\'{e}rentielle des groupes de {Lie} de
  dimension infinie et ses applications \`{a} l'hydrodynamique des fluides
  parfaits.
\newblock {\em Ann. Inst. Fourier (Grenoble) 16}, 1 (1966), 319--361.

\bibitem{BarlesLesigne_BSDEPDE97_inbook}
{\sc Barles, G., and Lesigne, E.}
\newblock \textrm{SDE}, \textrm{BSDE} and \textrm{PDE}.
\newblock In {\em Backward Stochastic Differential Equations}, vol.~364 of {\em
  Pitman Research Notes in Mathematics Series}. Harlow: Longman, 1997,
  pp.~47--80.

\bibitem{Bender-Zhang-08}
{\sc Bender, C., and Zhang, J.}
\newblock Time discretization and {M}arkovian iteration for coupled {FBSDE}s.
\newblock {\em Ann. Appl. Probab. 18}, 1 (2008), 143--177.

\bibitem{Bhattacharya-Waymire2003}
{\sc Bhattacharya, R.~N., Chen, L., Dobson, S., Guenther, R.~B., Orum, C.,
  Ossiander, M., Thomann, E., and Waymire, E.~C.}
\newblock Majorizing kernels and stochastic cascades with applications to
  incompressible \textrm{Navier-Stokes} equations.
\newblock {\em Trans. Amer. Math. Soc. 355}, 12 (2003), 5003--5040.

\bibitem{bloch2000optimal}
{\sc Bloch, A., Holm, D., Crouch, P., and Marsden, J.}
\newblock An optimal control formulation for inviscid incompressible ideal
  fluid flow.
\newblock In {\em Decision and Control, 2000. Proceedings of the 39th IEEE
  Conference on\/} (2000), vol.~2, IEEE, pp.~1273--1278.

\bibitem{Busnello1999}
{\sc Busnello, B.}
\newblock A probabilistic approach to the two-dimensional
  \textrm{Navier-Stokes} equations.
\newblock {\em Ann. Probab. 27}, 4 (1999), 1750--1780.

\bibitem{BusnelloFlandoliRomito05}
{\sc Busnello, B., Flandoli, F., and Romito, M.}
\newblock A probabilistic representation for the vorticity of a
  three-dimensional viscous fluid and for general systems of parabolic
  equations.
\newblock {\em Froc. Edinb. Math. Soc. 48}, 2 (2005), 295--336.

\bibitem{cheridito2007second}
{\sc Cheridito, P., Soner, H.~M., Touzi, N., and Victoir, N.}
\newblock Second-order backward stochastic differential equations and fully
  nonlinear parabolic {PDEs}.
\newblock {\em Commun. Pure Appl. Math. 60}, 7 (2007), 1081--1110.

\bibitem{Chorin73}
{\sc Chorin, A.~J.}
\newblock Numerical study of slightly viscous flow.
\newblock {\em J. Fluid Mech. 57}, 4 (1973), 785--796.

\bibitem{chorin1994vorticity}
{\sc Chorin, A.~J.}
\newblock {\em Vorticity and turbulence}.
\newblock Springer, 1994.

\bibitem{Constantin-Iyer-2008}
{\sc Constantin, P., and Iyer, G.}
\newblock A stochastic \textrm{Lagrangian} representation of the
  three-dimensional incompressible \textrm{Navier-Stokes} equations.
\newblock {\em Commun. Pure Appl. Math. 61}, 3 (2008), 330--345.

\bibitem{Constantin-Iyer-arxiv-2011}
{\sc Constantin, P., and Iyer, G.}
\newblock A stochastic \textrm{Lagrangian} approach to the
  \textrm{Navier-Stokes} equations in domains with boundary.
\newblock {\em Ann. Appl. Probab. 21}, 4 (2011), 1466--1492.
\newblock DOI: 10.1214/10-AAP731.

\bibitem{Cruzeiro-Shamarova-2009}
{\sc Cruzeiro, A.~B., and Shamarova, E.}
\newblock \textrm{Navier-Stokes} equations and forward-backward \textrm{SDEs}
  on the group of diffeomorphisms of a torus.
\newblock {\em Stoch. Proc. Appl. 119\/} (2009), 4034--4060.

\bibitem{cruzeiro2011_FB_burgers}
{\sc Cruzeiro, A.~B., and Shamarova, E.}
\newblock On a forward-backward stochastic system associated to the burgers
  equation.
\newblock In {\em Stoch. Anal. Financ. Appl.} 2011, pp.~43--59.

\bibitem{delarue2002existence}
{\sc Delarue, F.}
\newblock On the existence and uniqueness of solutions to {FBSDEs} in a
  non-degenerate case.
\newblock {\em Stoch. Proc. Appl. 99}, 2 (2002), 209--286.

\bibitem{Delarue-Menozzi-06}
{\sc Delarue, F., and Menozzi, S.}
\newblock A forward-backward stochastic algorithm for quasi-linear {PDE}s.
\newblock {\em Ann. Appl. Probab. 16}, 1 (2006), 140--184.

\bibitem{Delarue-Menozzi-08}
{\sc Delarue, F., and Menozzi, S.}
\newblock An interpolated stochastic algorithm for quasi-linear {PDE}s.
\newblock {\em Math. Comput. 77}, 261 (2008), 125--158.

\bibitem{DipernaLions1989}
{\sc DiPerna, R.~J., and Lions, P.-L.}
\newblock Ordinary differential equations, transport theory and
  \textrm{Sobolev} spaces.
\newblock {\em Invert. Math. 98\/} (1989), 511--547.

\bibitem{ebin1970groups}
{\sc Ebin, D., and Marsden, J.}
\newblock Groups of diffeomorphisms and the motion of an incompressible fluid.
\newblock {\em Ann. Math. 92}, 1 (1970), 102--163.

\bibitem{ElworthyLi1994}
{\sc Elworthy, K.~D., and Li, X.}
\newblock Formulae for the derivatives of heat semigroups.
\newblock {\em J. Funct. Anal. 125\/} (1994), 252--286.

\bibitem{EspositoMarra_1989}
{\sc Esposito, R., and Marra, R.}
\newblock Three-dimensional stochastic vortex flows.
\newblock {\em Math. Methods Appl. Sci. II\/} (1989), 431--445.

\bibitem{Esposito-Sciarretta1988}
{\sc Esposito, R., Marra, R., Pulvirenti, M., and Sciarretta, C.}
\newblock A stochastic \textrm{Lagrangian} picture for the three dimensional
  \textrm{Navier-Stokes} equations.
\newblock {\em Comm. Partial Diff. Eq. 13}, 12 (1988), 1601--1610.

\bibitem{Karoui_Peng_Quenez}
{\sc E\textrm{l Karoui}, N., Peng, S., and Quenez, M.~C.}
\newblock Backward stochastic differential equations in finance.
\newblock {\em Math. Finance 7}, 1 (1997), 1--71.

\bibitem{FriedmanPDEs}
{\sc Friedman, A.}
\newblock {\em Partial Differential Equations}.
\newblock Holt, Rinehart and Winston, New York, 1969.

\bibitem{GilbargTrud1983}
{\sc Gilbarg, D., and Trudinger, N.~S.}
\newblock {\em Partial Differential Equations of Second Order}, second
  edition~ed., vol.~224 of {\em Grundlehren der Mathematischen Wissenschaften}.
\newblock Springer-Verlag, Berlin, 1983.

\bibitem{gomes2005variational}
{\sc Gomes, D.}
\newblock A variational formulation for the {Navier-Stokes} equation.
\newblock {\em Commun. Math. Phys. 257}, 1 (2005), 227--234.

\bibitem{HuPengFK95}
{\sc Hu, Y., and Peng, S.}
\newblock Solution of forward-backward stochastic differential equations.
\newblock {\em Probab. Theory Relat. Fields 103\/} (1995), 273--283.

\bibitem{inoue1979new}
{\sc Inoue, A., and Funaki, T.}
\newblock On a new derivation of the {Navier-Stokes} equation.
\newblock {\em Commun. Math. Phys. 65}, 1 (1979), 83--90.

\bibitem{LeJanSznitman1997}
{\sc Jan, Y.~L., and Sznitman, A.~S.}
\newblock Stochastic cascades and 3-dimensional \textrm{Navier-Stokes}
  equations.
\newblock {\em Probab. Theory Relat. Fields 109\/} (1997), 343--366.

\bibitem{Kunita_book}
{\sc Kunita, H.}
\newblock {\em Stochastic Flows and Stochastic Differential Equations}.
\newblock Cambridge University Press, World Publishing Corp, 1990.

\bibitem{Ladyzhenskaia_68}
{\sc Ladyzenskaja, O.~A., Solonnikov, V.~A., and Ural'ceva, N.~N.}
\newblock {\em Linear and Quasi-linear Equations of Parabolic Type}.
\newblock AMS, Providence, 1968.

\bibitem{MaProtterYong1994}
{\sc Ma, J., Protter, P., and Yong, J.}
\newblock Solving forward-backward stochastic differential equations
  explicitly-- a four step scheme.
\newblock {\em Probab. Theory Relat. Fields 98\/} (1994), 339--359.

\bibitem{BertozziMajda2002}
{\sc Majda, A.~J., and Bertozzi, A.~L.}
\newblock {\em Vorticity and Incompressible Flow}, vol.~27 of {\em Cambridge
  Texts in Applied Mathematics}.
\newblock Cambridge University Press, Cambridge, 2002.

\bibitem{Navier1822}
{\sc Navier, C.}
\newblock M\'emoire sur les lois du mouvement des fluides.
\newblock {\em Mem. Acad. Sci. Inst. France 6}, 2 (1822), 375--394.

\bibitem{Nirenberg1959}
{\sc Nirenberg, L.}
\newblock On elliptic partial differential equations.
\newblock {\em Annali della Scoula Norm. Sup. Pisa 13\/} (1959), 115--162.

\bibitem{Ossiander2005}
{\sc Ossiander, M.}
\newblock A probabilistic representation of solutions of the incompressible
  \textrm{Navier-Stokes} equations in $\mathbb{R}^3$.
\newblock {\em Probab. Theory Relat. Fields 133\/} (2005), 267--298.

\bibitem{ParPeng_90}
{\sc Pardoux, E., and Peng, S.}
\newblock Adapted solution of a backward stochastic differential equation.
\newblock {\em Systems Control Lett. 14}, 1 (1990), 55--61.

\bibitem{ParPeng1992QPDE}
{\sc Pardoux, E., and Peng, S.}
\newblock Backward stochastic differential equations and quasilinear parabolic
  partial differential equations.
\newblock In {\em Stochastic Partial Differential Equations and Their
  Applications\/} (1992), B.~L. Rozovskii and R.~S. Sowers, Eds., vol.~176 of
  {\em Lect. Notes Control Inf. Sci.}, Berlin Heidelberg New York: Springer,
  pp.~200--217.

\bibitem{PardouxTangFK99}
{\sc Pardoux, E., and Tang, S.}
\newblock Forward-backward stochastic differential equations and quasilinear
  parabolic \textrm{PDEs}.
\newblock {\em Probab. Theory Relat. Fields 114\/} (1999), 123--150.

\bibitem{PardouxZhang98}
{\sc Pardoux, E., and Zhang, S.}
\newblock Generalized \textrm{BSDEs} and nonlinear \textrm{N}eumann boundary
  value problem.
\newblock {\em Probab. Theory Relat. Fields 110\/} (1998), 535--558.

\bibitem{Peng1991_QPDE}
{\sc Peng, S.}
\newblock Probabilistic interpretation for systems of quasilinear parabolic
  partial differential equations.
\newblock {\em Stoch. Stoch. Rep. 37\/} (1991), 61--74.

\bibitem{peng2000infinite}
{\sc Peng, S., and Shi, Y.}
\newblock Infinite horizon forward--backward stochastic differential equations.
\newblock {\em Stoch. Proc. Appl. 85}, 1 (2000), 75--92.

\bibitem{pengwu1999fully}
{\sc Peng, S., and Wu, Z.}
\newblock Fully coupled forward-backward stochastic differential equations and
  applications to optimal control.
\newblock {\em {SIAM} J. Control Optim. 37}, 3 (1999), 825--843.

\bibitem{QiuTangYou-SPA-2011}
{\sc Qiu, J., Tang, S., and You, Y.}
\newblock \textrm{2D} backward stochastic \textrm{Navier-Stokes} equations with
  nonlinear forcing.
\newblock {\em Stoch. Proc. Appl. 122\/} (2012), 334--356.

\bibitem{russo1999impulse}
{\sc Russo, G., and Smereka, P.}
\newblock Impulse formulation of the {Euler} equations: general properties and
  numerical methods.
\newblock {\em J. of Fluid Mech. 391\/} (1999), 189--209.

\bibitem{Stein1970}
{\sc Stein, E.~M.}
\newblock {\em Singular Integrals and Differentiability Properties of
  Functions}.
\newblock Princeton University Press, Princeton, 1970.

\bibitem{Stokes1849}
{\sc Stokes, G.}
\newblock On the theories of internal frictions of fluids in motion and of the
  equilibrium and motion of elastic solids.
\newblock {\em Trans. Camb. Phil. Soc. 8\/} (1849), 287--319.

\bibitem{Tang_05}
{\sc Tang, S.}
\newblock Semi-linear systems of backward stochastic partial differential
  equations in $\mathbb{R}^n$.
\newblock {\em Chin. Ann. Math. 26B}, 3 (2005), 437--456.

\bibitem{Temam84}
{\sc Temam, R.}
\newblock {\em Navier-Stokes Equations: Theory and Numerical Analysis},
  third~ed.
\newblock North-Holland-Amsterdam, New York, Oxford, 1984.

\bibitem{rT95}
{\sc Temam, R.}
\newblock {\em Navier-Stokes Equations and Nonlinear Functional Analysis},
  2~ed.
\newblock Society for Industrial and Applied mathematics, 1995.

\bibitem{TRiebel_83}
{\sc Triebel, H.}
\newblock {\em Theory of Function Spaces}, vol.~78 of {\em Monographs in
  Mathematics}.
\newblock Birkh$\ddot{a}$user, Basel, Boston, Stuttgart, 1983.

\bibitem{wang2012probabilistic}
{\sc Wang, J., and Zhang, X.}
\newblock Probabilistic approach for systems of second order quasi-linear
  parabolic pdes.
\newblock {\em J. Math. Anal. Appl. 388}, 2 (2012), 676--694.

\bibitem{yasue1983variational}
{\sc Yasue, K.}
\newblock A variational principle for the navier-stokes equation.
\newblock {\em J. Funct. Anal. 51}, 2 (1983), 133--141.

\bibitem{YongFK1997}
{\sc Yong, J.}
\newblock Finding adapted solutions of forward-backward stochastic differential
  equations: Method of continuation.
\newblock {\em Probab. Theory Relat. Fields 107\/} (1997), 537--572.

\bibitem{Zhang-bNS-2010}
{\sc Zhang, X.}
\newblock A stochastic representation for backward incompressible
  \textrm{Navier-Stokes} equations.
\newblock {\em Probab. Theory Relat. Fields 148\/} (2010), 305--332.

\end{thebibliography}

\end{document}